\documentclass[letterpaper,11pt]{article}
\usepackage[margin=0.94in]{geometry}

\usepackage{cmap} 
\usepackage[utf8]{inputenc}
\usepackage[english]{babel}
\usepackage[T1]{fontenc}
\usepackage{amsmath}
\usepackage{amssymb}
\usepackage{amsfonts}
\usepackage{bm}
\usepackage{bbm}
\usepackage{xcolor}
\definecolor{blueviolet}{rgb}{0.2, 0.2, 0.6}
\definecolor{webgreen}{rgb}{0,.5,0}
\definecolor{webbrown}{rgb}{.6,0,0}
\usepackage{setspace}
\usepackage[pdftex,
	bookmarks=false,
	colorlinks=true, 
	urlcolor=webbrown, 
	linkcolor=blueviolet, 
	citecolor=webgreen,
	pdfstartpage=1,
	pdfstartview={FitH},  
	bookmarksopen=false
	]{hyperref}
\allowdisplaybreaks
\usepackage{tikz}
\usepackage{braket}
\usepackage[numbers,sort&compress]{natbib}
\usepackage{amsthm}
\usepackage{dsfont}
\usepackage{listings}
\usepackage[capitalize]{cleveref}
\usepackage{appendix}

\usepackage{authblk}

\numberwithin{equation}{section}

\newtheorem{theorem}{Theorem}

\newtheorem{corollary}{Corollary}

\newtheorem{definition}{Definition}

\newtheorem{lemma}{Lemma}
\newtheorem{proposition}{Proposition}
\newtheorem{fact}{Fact}

\theoremstyle{definition}

\newcommand{\mg}[1]{\mathsf{\color{gray} #1}}

\DeclareMathOperator{\poly}{poly}

\newcommand{\PD}{\Pi^{\text{dist}}}
\newcommand{\cPD}{\Pi^{\text{c-dist}}}
\newcommand{\cPDLR}{\Pi^{\text{c-dist}}_{\mathsf{X}_\mathsf{L} \mathsf{Y}_\mathsf{R}}}

\newcommand{\vertiii}[1]{{\left\vert\kern-0.25ex\left\vert\kern-0.25ex\left\vert #1 \right\vert\kern-0.25ex\right\vert\kern-0.25ex\right\vert}}

\newcommand{\rom}[1]{\mathtt{\uppercase\expandafter{\romannumeral #1\relax}}}

\DeclareMathOperator*{\E}{{\mathbb{E}}}

\usepackage{xargs}
\usepackage[colorinlistoftodos,prependcaption,textsize=tiny]{todonotes}
\newcommandx{\lz}[2][1=]{\todo[linecolor=red,backgroundcolor=red!10,bordercolor=red,#1]{LZ: #2}}

 \usepackage{graphicx}
\usepackage{bm}
\usepackage{amsmath}
\usepackage{physics}
\usepackage{amssymb}
\usepackage{amsfonts}
\usepackage{amsthm}
\usepackage{bbm}
\usepackage{mathtools}
\usepackage{hyperref}
\usepackage{braket}
\usepackage[normalem]{ulem}
\usepackage{wrapfig}
\usepackage{tikz}
\usepackage{dsfont}
\usepackage{comment}
\usepackage{thmtools,thm-restate}
\usepackage[export]{adjustbox}

\newtheorem*{theorem*}{Theorem}

\newtheorem*{task*}{Task}
\newtheorem*{proposition*}{Proposition}

\newcommand{\bs}{\boldsymbol}
\newcommand{\ee}{\end{equation}}

\DeclareMathOperator{\Wg}{Wg}

\def\:={\,\raisebox{0.85pt}{.}\hspace{-2.78pt}\raisebox{2.85pt}{.}\!\!=\,}
\def\=:{\,=\!\!\raisebox{0.85pt}{.}\hspace{-2.78pt}\raisebox{2.85pt}{.}\,}

\newcommand{\dk}[1]{{ \color{orange} DK: #1 }}

\usepackage{authblk}

\newif\ifnumerics
\numericsfalse 

\begin{document}

\title{Hardness of recognizing phases of matter}

\author[*,1,2]{Thomas Schuster}
\author[*,3]{Dominik Kufel}
\author[3]{Norman Y. Yao}
\author[1,2]{Hsin-Yuan Huang}

\affil[1]{California Institute of Technology}
\affil[2]{Google Quantum AI}
\affil[3]{Harvard University}

\date{\today}

\maketitle

\begin{abstract}\normalsize 
We prove that recognizing the phase of matter of an unknown quantum state is quantum computationally hard. 
More precisely, we show that there exist  quantum states for which the quantum computational resources of any phase-recognition algorithm  must grow exponentially with the correlation range~$\xi$, i.e.~the light-cone size of a local unitary circuit that maps the state to a fixed point.
This exponential scaling renders phase recognition impractical even for moderate correlation ranges, and leads to super-polynomial computational time in the system size~$n$ whenever $\xi=\omega(\log n)$.
Our results encompass a substantial portion of all known phases of matter, including symmetry-breaking phases and symmetry-protected topological phases for any discrete on-site symmetry group in any spatial dimension. 
To establish this hardness, we extend the study of pseudorandom unitaries (PRUs) to quantum systems with symmetries. We prove that symmetric PRUs exist under standard cryptographic conjectures, and can be constructed in extremely low circuit depths. We also establish similar hardness for purely classical phases of matter.
Looking forward, our results raise a fundamental open question: if not symmetries, what properties of physical systems guarantee that their phases are efficient to recognize?
\end{abstract}

\pagenumbering{arabic} 
\setcounter{page}{1}

\addtocontents{toc}{\protect\setcounter{tocdepth}{0}}

\section{Introduction}

Phases of matter form a cornerstone of modern physics, and provide a fundamental organizing principle for understanding complex quantum systems. 
Classical phases of matter are famously characterized by the Landau paradigm, where phase transitions are driven by the spontaneous breaking of symmetries in the system~\cite{landau2013statistical}.
Beyond the Landau paradigm, quantum mechanics has unveiled entirely new phases of matter, including topological order and symmetry-protected topological phases~\cite{wen2004quantum,chen2013symmetry,zeng2019quantum}.
The ability to identify and characterize these diverse phases of matter is of fundamental interest across physics and information science and crucial for advancing quantum technologies~\cite{altman2021quantum,bauer2020quantum,semeghini2021probing,huang2021classical,clark2020observation,satzinger2021realizing,zhang2022digital,dupont2022quantum,iqbal2024non,leonard2023realization,haghshenas2025digital,will2025probing}.

A fundamental question underlying this scientific pursuit is the following: Given experimental access to a quantum state, how can one determine its phase of matter~\cite{jiang2012identifying,rodriguez2019identifying,cian2021many,herrmann2022realizing,huang2022provably,cian2022extracting,cong2019quantum,lake2022exact,cong2024enhancing,liu2023model,bouland2023public}? 
Despite tremendous interest, precisely answering this question has proven challenging. 
In many physical settings, phase recognition is known to be efficient, using either basic local measurements~\cite{sachdev1999quantum,huang2022provably} or more powerful quantum computational algorithms~\cite{cong2019quantum,lake2022exact,cong2024enhancing,liu2023model}.
On the other hand, all \emph{general} rigorous algorithms for phase recognition, without physical assumptions, have runtimes that scale exponentially in the range of correlations $\xi$ of the state of interest~\cite{huang2024learning,kim2024learning,landau2025learning}.
This scaling mirrors that of brute-force quantum state tomography on local patches of radius $\mathcal{O}(\xi)$ and becomes infeasible for physical systems with even a moderate correlation range.
Intriguingly, recent work has proven that for certain phases, this general computational hardness is fundamental: recognizing whether a quantum state has \emph{topological order} can require super-polynomial quantum computational resources whenever $\xi = \omega(\log n)$~\cite{schuster2024random}. 
This follows from the discovery that pseudorandom unitaries (PRUs)~\cite{ji2018pseudorandom,metger2024simple,chen2024efficient,ma2024construct} can be realized in extremely shallow geometrically-local circuits~\cite{schuster2024random}.
Nonetheless, this result leaves open an essential question: \emph{How hard is it to recognize phases of matter beyond topological order?}

In this work, we prove that a substantial portion of all known phases of matter are quantum computationally hard to recognize.
This includes symmetry-breaking phases and symmetry-protected topological phases for any discrete on-site symmetry group in any spatial dimension.
Our results apply equally to ground states and mixed states, and extend even to \emph{classical} states and phases of matter using distinct techniques.
In all of these settings, we prove that there exist quantum or classical states for which the quantum computational complexity of any phase-recognition algorithm must grow exponentially in the range of correlations $\xi$\footnote{The range of correlations $\xi$ refers to the maximum distance at which any connected correlation functions remain non-zero or above a small value.}.
Asymptotically, the computational runtime become super-polynomial in the system size $n$ whenever $\xi = \omega(\log n)$.
Our results apply to any definition of phases of matter that is invariant under shallow geometrically-local symmetric unitary circuits with light-cone size $\xi$~\cite{chen2010local,chen2011complete}.

Notably, the states we identify that exhibit this hardness have several peculiar properties.
Most crucially, while local-unitary equivalence preserves all long-range correlations, quasiparticles, and entanglement properties of a phase of matter, it can cause the interaction range of the parent Hamiltonian to grow.
In particular, the parent Hamiltonians of the states we identify consist of local terms that act on $\mathcal{O}(\xi)$ qubits.
Our results, and those of~\cite{schuster2024random}, therefore leave open a central physical question: What is the complexity of recognizing phases of matter in the ground states of \emph{constant-local} Hamiltonians?
Answering this question will demand new mathematical insights into ground states and phases beyond current knowledge based on spectral gaps and Lieb-Robinson bounds~\cite{chen2010local,chen2011complete,hastings2005quasiadiabatic,bravyi2006lieb,schuch2011classifying}.
More broadly, our results should also be regarded as a worst-case statement: there \emph{exist} quantum states whose phase of matter is impossible to recognize efficiently. 
Nonetheless, our daily experience shows that many familiar phases of matter, such as liquids, solids, and magnets, are easy to recognize in practice.
This raises further questions, in parallel to the one above.
If not symmetries, what properties guarantee that phases of matter can be recognized efficiently?
Is the locality of the parent Hamiltonian alone sufficient, or are further physical assumptions required?

To establish our results, we extend the notion of PRUs, and their formation in extremely low circuit depths, to incorporate a broad range of physical symmetry operations.
A PRU is an efficient quantum circuit that is indistinguishable from a completely random unitary by any polynomial-time quantum computation~\cite{ji2018pseudorandom,ma2024construct,schuster2025strong}.
Together with their close cousin, unitary designs~\cite{emerson2003pseudo,gross2007evenly,dankert2005efficient,dankert2009exact,brandao2016local, haah2024efficient, chen2024incompressibility,laracuente2024approximate,schuster2024random,cui2025unitary,zhang2025designs}, PRUs are central to numerous areas of quantum science, including quantum cryptography~\cite{ji2018pseudorandom,ananth2022cryptography,kretschmer2023quantum}, device benchmarking~\cite{emerson2005scalable,knill2008randomized,elben2023randomized}, quantum supremacy experiments~\cite{arute2019quantum, morvan2023phase, abanin2025constructive}, quantum learning theory~\cite{huang2021quantum,chen2021exponential,chen2021hierarchy,aharonov2021quantum,cotler2023information},  quantum gravity~\cite{sekino2008fast,hayden2007black,yoshida2017efficient,brown2023quantum,schuster2022many}, and many-body quantum dynamics~\cite{deutsch1991quantum,srednicki1994chaos,rigol2008thermalization,fisher2023random,nahum2017entgrowth,cotler2022fluctuations}.
A critical question in recent research is how to generalize PRUs to physically constrained classes of quantum dynamics, such as those with symmetries~\cite{marvian2022restrictions,kong2021charge,kong2022near,mitsuhashi2023clifford,li2024efficient,li2024designs,liu2024unitary,mitsuhashi2025unitary,schuster2024random,haah2025short,west2025no,grevink2025will}.
Several works have established no-go results, showing that low-depth constructions of PRUs are impossible in the presence of continuous~\cite{haah2025short} or time-reversal symmetries~\cite{schuster2024random,west2025no,grevink2025will}.
To achieve our results, we develop an extensive collection of tools to prove that nearly all features of PRUs \emph{do} extend to quantum systems with any discrete on-site symmetries.
Beyond their implications for recognizing phases of matter, our results thus also provide broader support for using PRUs as models of realistic quantum systems in the physical world.

The remainder of our manuscript is structured as follows. Section~\ref{sec:quantum-background} provides background on quantum phases of matter, and Section~\ref{sec:quantum-summary} summarizes our main results on recognizing quantum phases of matter. These results are founded upon our low-depth construction of symmetric PRUs, which we describe in Sections~\ref{sec:symPRU} and~\ref{sec:symPRUlowdepth}.
\ifnumerics
\else
We return to quantum phases of matter in Section~\ref{sec:purestate} to provide clarifying discussions, intuition, and detailed analyses that demonstrate how hardness can emerge even at small correlation lengths for common phase recognition approaches.
\fi
Section~\ref{sec:mixedstate} extends our results to mixed state quantum phases of matter. We then present our results on classical phases of matter in Section~\ref{sec:classical}. Section~\ref{sec:discussion} concludes with additional discussions and open questions.

\begin{figure}
    \centering
    \includegraphics[width=1.0\textwidth]{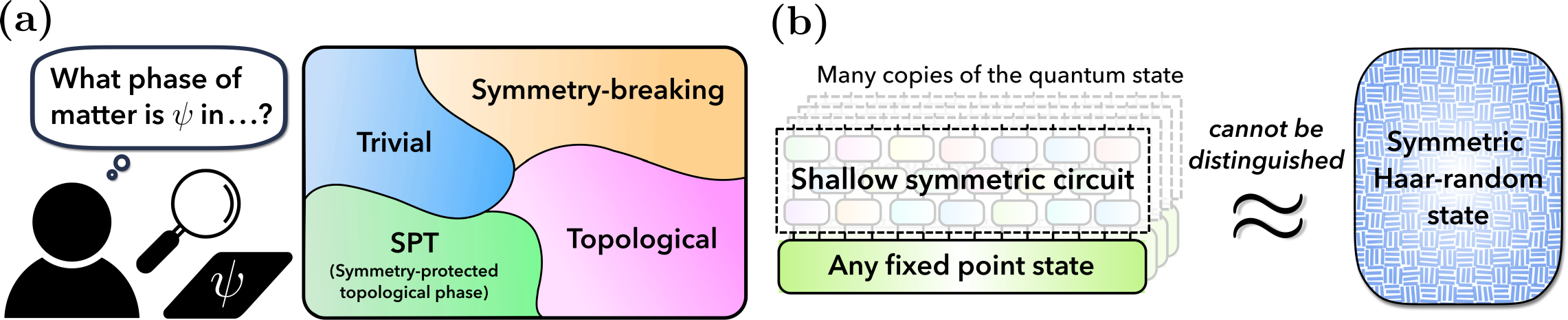}
    \caption{Illustration of our main results.
    \textbf{(a)} We consider the question: Given experimental access to many copies of a quantum state $\ket{\psi}$, can one recognize what phase of matter $\ket{\psi}$ is in?
    Examples of phases of matter include the trivial phase, symmetry-breaking phases, symmetry-protected topological (SPT) phases, and topological order.
    Our main result is that the complexity of recognizing phases of matter can grow exponentially in the correlation range $\xi$ of the quantum state, becoming super-polynomial in the system size $n$ as soon as $\xi = \poly(\log n)$. 
    \textbf{(b)} We achieve this result by extending the study of pseudorandom unitaries (PRUs) to quantum systems with discrete symmetries.
    We show that any fixed point state of any phase of matter can become indistinguishable from a symmetric Haar-random state after a low-depth symmetric circuit is applied.}
    \label{fig:1}
\end{figure}

\section{Quantum phases of matter} \label{sec:quantum}

In this section, we describe our main results on recognizing phases of matter in quantum systems.

\subsection{Background and definitions} \label{sec:quantum-background}

Quantum phases of matter capture properties of the ground states of gapped many-body quantum systems that remain invariant under smooth deformations of the system~\cite{wen2004quantum,zeng2019quantum}. 
The simplest example is a symmetry-breaking phase, in which the system is invariant under a symmetry $G$, yet features multiple nearly-degenerate ground states that transform into one another under the action of $G$~\cite{landau2013statistical}.
This contrasts with the trivial phase of a symmetric system, which features a unique ground state that respects the symmetry operation.
Quantum states can also exhibit novel phases beyond the symmetry-breaking paradigm.
These include topological orders~\cite{wen2004quantum,wen2017colloquium}, which can occur in the absence of any symmetry and are characterized by their patterns of long-range entanglement, and symmetry-protected topological (SPT) phases~\cite{chen2013symmetry,zeng2019quantum}, which are short-range entangled, yet in a manner that is non-trivial when the symmetry is considered.  We refer to Section~\ref{sec:purestate} for several concrete examples.

Our modern understanding of how to classify phases of matter centers upon ideas from quantum information~\cite{zeng2019quantum}.
Two quantum states are in the same phase of matter if and only if they are connected by (i) a low-depth local unitary circuit~\cite{chen2010local,chen2011complete,chen2013symmetry}, or (ii) a short-time adiabatic evolution under gapped Hamiltonians~\cite{hastings2005quasiadiabatic,bravyi2006lieb,schuch2011classifying}.
The two definitions are equivalent up to moderate overheads: any short-time adiabatic evolution can be Trotterized to a low-depth local unitary circuit up to a $\poly(\log n)$ overhead in circuit depth, while any depth-$\ell$ unitary circuit yields a gapped adiabatic path interpolating between an original parent Hamiltonian, $H_1$, and a final parent Hamiltonian, $H_2 = U H_1 U^\dagger$, up to a $\poly(\ell)$ increase in the Hamiltonian's locality.
Intuitively, low-depth circuits and short-time evolutions cannot change the structure of correlations and entanglement in a quantum state, owing to their limited light-cone.
Hence, the phase of matter remains invariant.
If the system is symmetric, the circuit or Hamiltonian must respect the symmetry operation as well.

Our results apply to any definition of phases of matter such that the following holds.
For each phase of matter, we can associate a representative ``fixed point'' state $\ket{\psi_0}$.
Then, we assume that any quantum state $\ket{\psi}$ that is related to $\ket{\psi_0}$ by a symmetric geometrically-local depth-$\ell$ unitary circuit $U$ is in the same phase as $\ket{\psi_0}$.
The parameter $\ell$ controls the maximum allowed circuit depth of the unitary $U$.
We also let $\xi$ denote the size of the light-cone of $U$.
Physically, $\xi$ upper bounds the maximum range of correlations in the quantum state $\ket{\psi}$; that is, the distance beyond which all connected correlation functions are zero or below a small value.
We note that this definition differs slightly from other common notions of correlation length, which concern the asymptotic exponential decay rate of connected correlation functions.
We have $\xi \leq \ell$ by definition.

\begin{figure}
    \centering
    \includegraphics[width=1.0\textwidth]{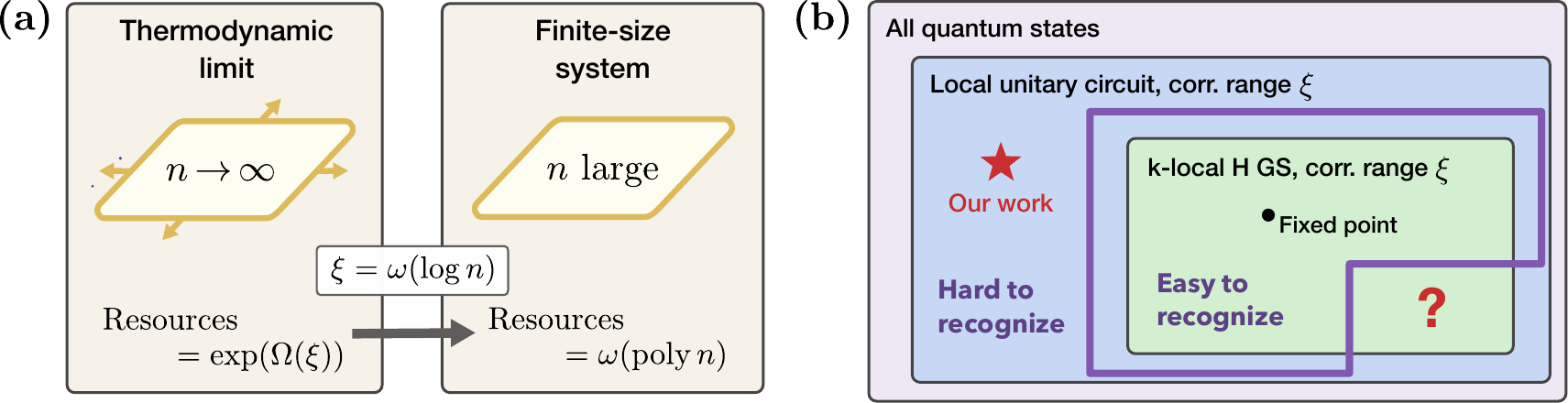}
    \caption{Additional details and contextualization of our results.
    \textbf{(a)} We consider two regimes.
    In the thermodynamic limit, $n \rightarrow \infty$, we show that the complexity of any phase recognition algorithm must grow exponentially in the correlation range $\xi$.
    In finite-size systems, the same exponential growth also applies, and implies that the complexity becomes super-polynomial in the number of qubits $n$ when $\xi = \omega(\log n)$.
    \textbf{(b)} Our results (red star) apply to quantum states that are related to fixed point states by local unitary circuits with light-cone size $\xi$ (blue region);
    this encompasses a larger set of states than definitions that restrict to ground states of constant-local Hamiltonians (green region). Whether there exist  ground states of constant-local Hamiltonians whose phases of matter become exponentially hard to recognize in the correlation length remains an open question (red question mark).
    For the sake of illustration, in this figure we have neglected the $\poly(\log n)$ overhead in light-cone size when translating adiabatic constant-local Hamiltonian evolution to local unitary circuits (see main text).}
    \label{fig:2}
\end{figure}

In condensed matter literature, the correlation range $\xi$ and circuit depth $\ell$ are often assumed to be constant with respect to the system size, $\xi, \ell = \mathcal{O}(1)$.
This is chosen because many properties of phases of matter become sharply defined only in the strict thermodynamic limit, $n \rightarrow \infty$.
In contrast, in order to study the complexity of recognizing phases of matter as a computational problem, one must identify a large but finite scaling parameter to study the complexity with respect to.
We consider two physically-motivated choices of scaling parameter in our work [Fig.~\ref{fig:2}(a)].
First, working in the strict $n \rightarrow \infty$ limit, we consider the complexity of phase recognition as a function of the correlation range $\xi$.
Our results will show that the resources of any phase recognition algorithm must grow exponentially in $\xi$.
We remark that in this limit, $n \rightarrow \infty$, a bounded quantum algorithm cannot access all $n$ qubits of a quantum state.
Second, we consider the standard computer scientific setting, and study the complexity of phase recognition as a function of the system size $n$ itself.
Consistent with the above scaling, our results will show that the resources of any phase recognition algorithm become super-polynomial in $n$ whenever $\xi = \omega(\log n)$.
For smaller correlation ranges, $\xi = \mathcal{O}(\log n)$, any phase of matter can be recognized in $\poly n$ time by performing brute-force local state tomography~\cite{huang2024learning,landau2025learning}.

%

A natural question is whether the scaling $\xi = \omega(\log n)$ is well-motivated physically.
We note that this scaling is not fundamental to our results; rather, it is an artifact of studying the complexity with respect to $n$ as opposed to $\xi$.
%
%
In favor of this scaling, we note that from a quantum-information perspective, all hallmark properties of phases of matter remain robust under any local unitary circuit with light-cone smaller than the system radius, $\xi = o(n^{1/d})$, where $d$ is the spatial dimension.
%
%
For example, this sub-extensive light-cone is sufficient to guarantee that all long-range correlations, entanglement, and quasiparticle excitations of a quantum state are precisely preserved, after coarse-graining over a distance $\xi$.
Hence, one can view these states as well-defined, verifiable examples of quantum phases of matter.
Moreover, from a practical perspective, an essential question for experiments is how the resources of phase recognition strategies grow with the correlation range $\xi$.
To address this question theoretically, one must allow $\xi$ to grow with the scaling parameter of any study.

On a more cautionary note, however, we note that the set of local unitary circuits with light-cone $\xi$ is much larger than the set of gapped adiabatic constant-local Hamiltonian evolutions  [Fig.~\ref{fig:2}(b)].
This follows because a local unitary circuit will generically increase the locality of a parent Hamiltonian from $\mathcal{O}(1)$ to $\mathcal{O}(\xi)$.
This opens the possibility that recognizing phases of matter in physical ground states, with large correlation ranges but constant-local Hamiltonians, may be meaningfully easier than recognizing phases of matter in complete generality (i.e.~under the local unitary circuit definition).
Our hardness results apply to the local-unitary-circuit definition and not to constant-local Hamiltonians.
This highlights that, if phase recognition is to be easy for all constant-local Hamiltonians, it must arise from a special structure of their ground states that, to our knowledge, is not yet understood.


\subsection{Summary of main results} \label{sec:quantum-summary}

Our main result establishes that the quantum computational complexity of recognizing quantum phases of matter grows exponentially in the correlation range $\xi$.
For technical reasons, the circuit depth $\ell$ for generating the states we consider is polynomial in the correlation range $\xi$; hence, the complexity grows as a stretched exponential in the circuit depth $\ell$.
This rapid growth implies that the complexity is super-polynomial in $n$ even for small values of $\xi = \omega(\log n)$ or $\ell = \poly(\log n)$.

\begin{theorem}[Hardness of recognizing quantum phases of matter] \label{thm:quantum}
    For any discrete on-site symmetry group $G$, distinguishing any quantum phase of matter from the trivial phase requires quantum algorithms with computational time scaling exponentially in the correlation range $\xi$. The computational time becomes super-polynomial in the system size $n$ whenever $\xi = \omega(\log n)$.
\end{theorem}

\noindent This lower bound on the hardness of phase recognition is optimal, since general rigorous algorithms exist with runtime growing exponentially in $\xi$~\cite{huang2024learning,landau2025learning}. These algorithms succeed by learning the entire circuit description of a shallow quantum state. Hence, our result proves that, in general settings, recognizing the phase of matter of a shallow quantum state is comparably hard to performing full state tomography. We discuss the extension of Theorem~\ref{thm:quantum} to mixed states in Section~\ref{sec:mixedstate}.

Before proceeding to our constructions, let us provide a high-level overview of our proof of Theorem~\ref{thm:quantum}.
We consider the following phase recognition task: one is given access to many copies of a pure quantum state $\ket{\psi}$ drawn from either the trivial phase or a non-trivial phase of matter. The goal is to recognize the phase of matter of $\ket{\psi}$ using arbitrary efficient quantum circuits and measurements on all copies of the state.
To prove that this task is computationally hard, we consider the complexity of recognizing a \emph{random example} of each phase of matter~\cite{schuster2024random}.
That is, for both the trivial and any non-trivial phase, we generate a random quantum state in the phase by applying a \emph{symmetric pseudorandom unitary} $U$ to the fixed point state, $\ket{\psi} = U \ket{\psi_0}$.
Crucially, we show that such unitaries can be formed in very low circuit depths.
By definition, applying a symmetric PRU to any state produces a state that is indistinguishable from Haar-random by any efficient quantum computation.
Hence, neither the random trivial state nor the random non-trivial state can be distinguished from Haar-random, which implies that the two states also cannot be distinguished from one another.
Therefore, recognizing the phase of matter is quantum computationally hard.
The states generated in our proof are the ground states of parent Hamiltonians $U H_0 U^\dagger$, where $H_0$ is the Hamiltonian of the fixed-point state and $U$ is the low-depth symmetric PRU.
%
%

%

%

In what follows, we first present the core technical results behind our proof by describing our definitions and constructions of low-depth symmetric PRUs.
\ifnumerics
\else
We then return to quantum phases of matter and provide additional discussions and analysis corroborating our result.
\fi
We provide an extension to mixed state phases of matter at the end of the section.

\subsection{Symmetric pseudorandom unitaries (PRUs)} \label{sec:symPRU}

Let us now introduce our construction of symmetric PRUs.
Pseudorandom unitaries (PRUs) are efficient random unitary ensembles that are indistinguishable from the Haar ensemble by any efficient quantum algorithm~\cite{ji2018pseudorandom,ma2024construct}.
Recently, the existence of PRUs that are secure against any subexponential-time quantum algorithms has been proven~\cite{ma2024construct} under a widely-accepted cryptographic conjecture: the quantum hardness of solving learning with errors (LWE)~\cite{regev2009lattices}.

In the presence of a symmetry group, one can naturally consider the symmetric Haar ensemble, containing all unitaries that commute with the symmetry operation~\cite{kong2021charge,kong2022near,mitsuhashi2023clifford,li2024efficient,li2024designs,liu2024unitary,mitsuhashi2025unitary}.
A symmetric PRU is then defined as an efficient ensemble of symmetric random unitaries that is indistinguishable from the symmetric Haar ensemble by any bounded-time quantum computation.
The existence of symmetric PRUs does not follow from the existence of standard PRUs, and, to date, no constructions of symmetric PRUs have been established.

Our first main technical result is a construction of symmetric PRUs for any discrete on-site symmetry in $\poly(n)$ circuit depth.
We will improve the circuit depth to $\poly(\log n)$ in the next subsection.

\begin{theorem}[Symmetric pseudorandom unitaries] \label{thm:poly}
    Under the conjecture that no subexponential-time quantum algorithm can solve LWE,
    symmetric pseudorandom unitaries with security against any $\exp(o(n))$-time quantum adversary exist, for every discrete on-site symmetry group $G$. The unitaries can be formed in polynomial circuit depth on any architecture. When $G$ is Abelian, the unitaries can be compiled from individually symmetric $\mathcal{O}(1)$-local gates.
\end{theorem}

\noindent Here, an on-site symmetry refers to any symmetry whose representation on $n$ qudits is a tensor product, $R_g = \bigotimes_{i} R^i_g$, for all $g \in G$. 
We refer to Appendix~\ref{sec: background} for a detailed introduction to the representation theory of symmetry groups in many-body quantum systems.
The proof of Theorem~\ref{thm:poly} is described below and given formally in Appendix~\ref{sec: symmetric PRUs}.

\paragraph{Construction overview.} Our construction of symmetric PRUs for Theorem~\ref{thm:poly} proceeds in three steps.
First, we decompose any $n$-qudit Hilbert space into a tensor sum of irreducible representations (``irreps'') of the symmetry group, $\mathcal{H} = \bigoplus_\lambda ( \mathcal{M}_\lambda \otimes \mathcal{F}_\lambda )$ [Fig.~\ref{fig:symPRU}(a)].
Each $\mathcal{F}_\lambda$ is the Hilbert space of the irrep $\lambda$, and each $\mathcal{M}_\lambda$ captures the multiplicity of the irrep.
In this basis, a symmetric Haar-random unitary is simply a tensor sum of random unitaries $U_\lambda$ in each symmetry sector, $U = \bigoplus_\lambda (U_\lambda \otimes \mathbbm{1}_{\mathcal{F}_\lambda})$.

Second, following this decomposition, we implement a sequence of \emph{controlled pseudorandom unitaries} $U_\lambda$ conditioned on the value of the system irrep $\lambda$.
This directly yields $U$ as above.
Prior to our work, the existence of controlled PRUs had been widely believed but not formally proven in the literature.
Hence, to implement this step, we prove that controlled PRUs exist under standard cryptographic assumptions and can be formed in polynomial circuit depth on any circuit geometry\footnote{During the preparation of this work, Tang and Wright~\cite{tang2025controlled} provided a distinct and independent proof of the existence of controlled pseudorandom unitaries.}.
Our construction applies to arbitrary tensor product Hilbert spaces, which is necessary due to the variable dimensionality of $\mathcal{M}_\lambda$.

Third, we prove that the entire construction in steps one and two can be implemented using \emph{individually symmetric} local gates whenever $G$ is Abelian.
Without this step, our construction would still yield symmetric PRUs in polynomial circuit depth.
However, the individual gates in the circuit might not be symmetric (only the circuit as a whole would be symmetric).
This is undesirable in several applications.
To resolve this, we prove a general result: any Abelian symmetric unitary $U$ can be implemented in polynomial depth via geometrically $5$-local symmetric gates whenever the unitary of each irrep $U_\lambda$ can be implemented in polynomial depth via geometrically 2-local gates.
This result uses a mild connectivity assumption on the circuit architecture.

\begin{figure}
    \centering
    \includegraphics[width=1.0\textwidth]{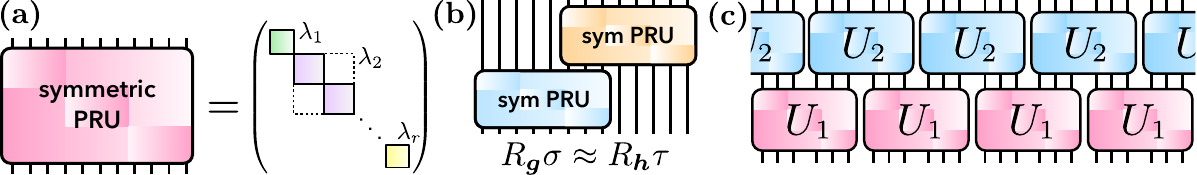}
    \caption{Schematic of our main results on symmetric pseudorandom unitaries (PRUs). 
    \textbf{(a)} Any symmetric unitary is block-diagonal in the irreducible representations (irreps) $\lambda_1,\ldots,\lambda_r$ of the symmetry group. Our $\poly n$-depth construction of symmetric PRUs  applies independent random unitaries in each block, by performing a sequence of controlled PRUs conditioned on the irrep of the entire system (Theorem~\ref{thm:poly}).
    \textbf{(b)} To construct symmetric PRUs in $\poly(\log n)$ depth, we first prove that two small symmetric PRUs can ``glue'' into a symmetric PRU on a larger system (Lemma~\ref{lemma:gluing}). This is used to prove that the two-layer circuit in which each unitary is drawn independently from a symmetric PRU ensemble is a low-depth symmetric PRU (Theorem~\ref{thm:polylog}). At a technical level, the symmetric gluing lemma follows by identifying the  terms $R_{\bs g} \sigma$ and $R_{\bs h} \tau$ arising from the first and second unitaries' twirls, where $\bs g, \bs h \in G^{\otimes k}$ are symmetry group elements and $\sigma, \tau \in S_k$ are permutation operators.
    \textbf{(c)} A translation-invariant modification of the two-layer circuit, which underlies our construction of $\poly(\log n)$-depth symmetric PRUs with translation invariance (Theorem~\ref{thm:TI}).}
    \label{fig:symPRU}
\end{figure}

\subsection{Symmetric PRUs in extremely low depth} \label{sec:symPRUlowdepth}

We now present our second main technical result, which shows that symmetric PRUs can be realized in extremely low circuit depths.

\paragraph{Why low-depth symmetric PRUs may not exist.} Before proceeding, we recall a potential counterargument against the possibility of low-depth symmetric random unitaries.
Symmetries of physical systems lead to the existence of conserved charges, which label how the system responds to the symmetry operation.
In a system with geometrically-local dynamics, charges are also conserved locally.
Thus, one might question: Is it possible for a low-depth symmetric unitary to shuffle charges around the system in the same manner as a deep random unitary?
This argument was used to rule out low-depth random unitaries in systems with any \emph{continuous} symmetry~\cite{haah2025short}.
In contrast, in what follows, we show that low-depth random unitaries can be naturally realized in systems with any discrete on-site symmetry.
Intuitively, charges of discrete symmetries can be locally created and annihilated in neutral combinations, which allows the dynamics to mimic those of a deep random unitary even at very low depths.

\vspace{1em}
At a high level, our construction of low-depth symmetric PRUs proceeds similarly to the construction of low-depth PRUs without symmetry~\cite{schuster2024random}.
We organize a system of $n$ qudits~\cite{fn10} along a 1D line, and consider a two-layer brickwork circuit composed of symmetric PRUs acting on small patches of $2\xi$ qudits each.
A translation-invariant version of this construction is shown in Fig.~\ref{fig:symPRU}(c). Without translation invariance, our construction uses different small random unitaries throughout each circuit layer instead of the same small random unitary.
We prove that the resulting two-layer circuit forms a symmetric PRU on all $n$ qudits by proving that one can ``glue'' the small PRUs together one at a time. Namely, we prove the following gluing lemma for symmetric random unitaries [Fig.~\ref{fig:symPRU}(b)]:

\begin{lemma}[Gluing symmetric random unitaries] \label{lemma:gluing}
    For any approximation error $\varepsilon \leq 1$, suppose each small unitary in the two-layer brickwork ensemble $\mathcal{E}$ is drawn from the symmetric Haar ensemble on $2\xi$ qudits. Then $\mathcal{E}$ forms an $\varepsilon$-approximate symmetric unitary $k$-design whenever $\xi \geq \log_2(nk^2|G|/\varepsilon)$.
\end{lemma}

\noindent The lemma references the notion of a symmetric unitary design, which we review in Appendix~\ref{sec: symmetric random unitaries}. 
By applying the gluing lemma $n/\xi$ times~\cite{schuster2024random}, we find that the two-layer  circuit forms a symmetric PRU whenever $\xi = \omega(\log n)$ (Appendix~\ref{sec: symmetric PRUs}).  
Moreover, as one increases $\xi$ above this value, the security becomes exponentially stronger in the value of $\xi$.

\begin{theorem}[Symmetric random unitaries in extremely low depth] \label{thm:polylog}
    Under the conjecture that no subexponential-time quantum algorithm can solve LWE,
    symmetric PRUs on $n$ qudits with security against any $\exp(o(\xi))$-time adversary can be constructed  with light-cone size $\mathcal{O}(\xi)$ and circuit depth $\poly(\xi)$, for any discrete on-site symmetry group $G$ and $\xi = \omega(\log n)$.
\end{theorem}

\noindent The circuit depth follows by instantiating each small PRU using Theorem~\ref{thm:poly}, which yields a circuit depth $\ell = \poly(\xi)$. We have $\ell = \poly(\log n)$ if $\xi = \poly(\log n)$.
As in Theorem~\ref{thm:poly}, when $G$ is Abelian, each unitary can be compiled from individually symmetric local gates.

\paragraph{Proof ideas for symmetric gluing lemma}
Our proof of the symmetric gluing lemma generalizes a large amount of technical machinery for random unitary analysis to the symmetric setting.
Although seemingly straightforward, this generalization requires several careful choices to enable a tractable analysis.
For the expert audience, we now summarize a few key intermediary steps and results. We begin by deriving an exact formula for the twirl of a symmetric Haar-random unitary in terms of a sum over permutation operators and tensor products of symmetry operators, $R_{\bs{g}} = \bigotimes_{i=1}^k R_{g_i}$. Then, we prove a much simpler \emph{approximate} formula for the twirl:
\begin{equation} \label{eq: approx sym Haar twirl}
    \E_U\left[ (U^{\otimes k}) \, A \, (U^{\otimes k})^\dagger \right] \approx \frac{1}{D^k} \sum_{\sigma \in S_k} \sum_{\bs{g} \in G^{\otimes k}} \tr( A \sigma^{-1} R_{\bs g}^{-1} ) \cdot R_{\bs g} \sigma.
\end{equation}
This replaces the sum over permutation operators in the approximate Haar twirl without symmetry~\cite{schuster2024random} with a sum over both permutation and symmetry group operators.
Third, we prove a key technical lemma for translating between different notions of error in symmetric unitary designs.
Finally, from these results, we provide a self-contained proof that $\varepsilon$-approximate symmetric unitary $k$-designs can be realized in $\poly(k) \cdot \poly \log (n/\varepsilon)$ depth for any discrete symmetry group $G$, without relying on any spectral gap analysis.
We expect that each of these results may be useful in future works, and refer the interested reader to Appendix~\ref{sec: symmetric random unitaries} for a full exposition and further details.

\paragraph{Translation-invariant symmetric PRUs.} Our third main technical result establishes the existence of \emph{translation-invariant} symmetric PRUs.
Beyond on-site symmetry groups, translation symmetry is by far the most common symmetry in the physical world and plays a crucial simplifying role in many quantum phenomena.
Our translation-invariant random unitary ensemble corresponds to the same two-layer circuit, but in which each unitary in the first layer is identical, and similarly for the second layer [Fig.~\ref{fig:symPRU}(c)].
This yields a symmetric random unitary that is invariant under translations by $a = 2\xi$ lattice sites.
We prove that the resulting unitary is indistinguishable from a translation-invariant symmetric Haar-random unitary by any efficient quantum algorithm.

\begin{theorem}[Translation-invariant symmetric pseudorandom unitaries] \label{thm:TI}
    Under the conjecture that no subexponential-time quantum algorithm can solve LWE, one-dimensional translation-invariant symmetric PRUs with security against any $\exp(o(\xi))$-time adversary can be constructed with light-cone size $\mathcal{O}(\xi)$ and circuit depth $\poly(\xi)$, for any discrete on-site symmetry group $G$ and $a = 2 \xi = \omega(\log n)$.
\end{theorem}

\noindent As in Theorems~\ref{thm:poly} and~\ref{thm:polylog}, when $G$ is Abelian, the translation-invariant symmetric PRUs can be compiled from individually symmetric local gates.
Extending this result to translations over a constant distance, $a = \mathcal{O}(1)$, remains an interesting open direction. 
The two-layer circuit explicitly breaks translation symmetry for any $a < 2\xi$.

Our proof of Theorem~\ref{thm:TI} is provided in Appendix~\ref{sec: TI} and proceeds along a similar path to Lemma~\ref{lemma:gluing}.
The primary difference is that the proof involves a significantly more intricate bound over a huge sum of permutation operators.
We introduce several helpful methods for counting permutations that enable us to complete this bound.

\begin{figure}
    \centering
    \includegraphics[width=1.0\textwidth]{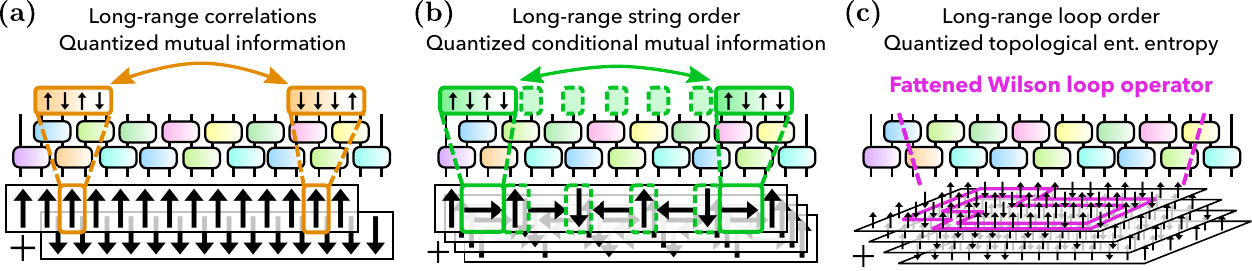}
    \caption{Three phases of matter and their order parameters after application of a low-depth symmetric PRU.
    \textbf{(a)} A symmetry-breaking phase (example: a GHZ state) features long-range correlations and a quantized non-zero mutual information.
    \textbf{(b)} A symmetry-protected topological  (SPT) phase (example: a 1D cluster state) features string order parameters and a quantized non-zero conditional mutual information.
    \textbf{(c)} Topological orders (example: a toric code state) feature loop order parameters and a quantized non-zero topological entanglement entropy.
    In all three cases, the order parameters and quantized correlations are retained after the low-depth PRU.
    However, they cannot be detected by any polynomial-time observer without knowledge of the unitary that was applied.}
    \label{fig:orderparameter}
\end{figure}

\subsection{Hardness of recognizing pure state quantum phases of matter} \label{sec:purestate}

Having established the existence of low-depth symmetric PRUs, we now return to their implications for recognizing quantum phases of matter.
Our main hardness result, Theorem~\ref{thm:quantum}, follows directly from combining our symmetric PRU constructions (Theorems~\ref{thm:polylog} and~\ref{thm:TI}) with the proof strategy outlined in Section~\ref{sec:quantum-summary}.
The key insight is that symmetric PRUs allow us to generate quantum states in well-defined phases that are computationally indistinguishable from Haar-random states, making phase recognition computationally hard.
Complete proof details are provided in Appendix~\ref{sec: phases of matter}.

In the remainder of this section, we explore the implications of this result through concrete examples, provide intuition for why phase recognition becomes hard, and analyze how our results apply to various extensions and practical scenarios.

\paragraph{Concrete examples.} To illustrate the power of our hardness result, consider the following four quantum states on $n$ qubits with symmetry $G = \mathbbm{Z}_2 \times \mathbbm{Z}_2$.
Here, $U$ denotes a low-depth symmetric PRU from our construction.
Our results prove that the four states, each representing a different phase of matter, are all computationally indistinguishable from one another:
\begin{enumerate}
    \item $\ket{\psi_{\text{triv}}} = U \ket{\text{product}}$, where  $\ket{\text{product}}$ is a product state respecting both symmetries. This state is in the trivial phase.
    \item $\ket{\psi_{\text{SSB}}} = U \ket{\text{GHZ}}$, where $\ket{\text{GHZ}}$ is a GHZ state that spontaneously breaks either one or both $\mathbbm{Z}_2$ symmetries. This state is in a spontaneous symmetry-broken phase~\cite{zeng2019quantum}.
    \item $\ket{\psi_{\text{SPT}}} = U \ket{\text{cluster}}$, where $\ket{\text{cluster}}$ is the 1D cluster state~\cite{zeng2019quantum}. This state is in a symmetry-protected topological phase.
    \item $\ket{\psi_{\text{topo}}} = U \ket{\text{TC}}$, where $\ket{\text{TC}}$ is the toric code state and respects both symmetries. This state is in a topological phase~\cite{zeng2019quantum}.
\end{enumerate}
Each state is in a well-defined and distinct phase of matter from all other states.
The different phases of matter are each characterized by their own unique patterns of entanglement and long-range correlations.
If an observer has knowledge of the unitary $U$, then the observer can easily use this knowledge of $U$ to detect these correlations, and verify that each state is in its given phase.
However, our results prove that without knowledge of $U$ (i.e.~when the quantum state is unknown to the observer), no efficient quantum algorithm can distinguish between any of the four families of states.

\paragraph{Why does phase recognition become hard?} While our theorem formally establishes that recognizing phases of matter is hard, our technical proof does not provide much intuition as to why this is the case. 
Let us address this intuition through the lens of one common approach to phase recognition: measuring order parameters (Fig.~\ref{fig:orderparameter}). 
Order parameters are observables whose expectation values diagnose a particular phase of matter.
For example, symmetry-breaking phases feature long-range correlations with respect to an order parameter that breaks the symmetry.
For the GHZ state, these correspond to the expectation values, $\bra{\psi} Z_j Z_i \ket{\psi}$, for distant sites $i, j$.
Similarly, one-dimensional SPT phases feature string-like order parameters, which are equal to a symmetry operator in their bulk while breaking a symmetry at each edge.
For the cluster state, these correspond to the expectation values, $\bra{\psi} Z_i X_{i+1} X_{i+3} \ldots X_{j-3}  X_{j-1} Z_j \ket{\psi}$, for $j = i \text{ mod } 2$.
While the order parameters of many toy-model systems happen to be simple, this is not guaranteed to be the case.

The key intuition behind our results is that, without prior knowledge about a quantum state, its order parameters can be quantum computationally hard to identify.
For example, when considering spontaneous breaking of a $\mathbbm{Z}_2$ symmetry, there is no property of the symmetry or phase of matter that fixes the order parameter to be a simple one-body $Z$ operator (as in the GHZ state).
Rather, after conjugation by a low-depth symmetric PRU, all possible order parameters become scrambled superpositions of many symmetry-breaking operators acting on regions of $\mathcal{O}(\xi)$ qubits [Fig.~\ref{fig:orderparameter}(a)].
The overlap of these scrambled order parameters with any simple one-body operator is typically exponentially small in $\xi$, which makes the phase of matter hard to recognize using simple one-body expectation values.
At the same time, the full $\mathcal{O}(\xi)$-body order parameters contain exponentially many random and non-commuting terms, and hence are hard to reconstruct in any experiment without knowledge of $U$.
This causes order-parameter-based approaches for recognizing phases of matter to require exponential resources as $\xi$ increases.
Analogous arguments apply to the order parameters of trivial, SPT, and topological phases; see Fig.~\ref{fig:orderparameter}(b,c) for an illustration.

\paragraph{What about entanglement-based approaches?} Beyond order parameters, phase recognition often relies on measuring characteristic entanglement signatures that distinguish different phases of matter.
However, these approaches encounter the same fundamental scaling limitations with the correlation range $\xi$.
Consider the Kitaev-Preskill construction \cite{kitaev2006topological} for measuring the topological entanglement entropy (TEE): $S_{\text{topo}}=S_{AB}+S_{BC} + S_{AC} - S_A - S_B - S_C -S_{ABC}$.
This quantity is a powerful tool designed to distinguish topological phases ($S_{\text{topo}} > 0$) from trivial phases ($S_{\text{topo}} = 0$); see \cite{kim2023universal} for a rigorous version of TEE that addresses spurious cases \cite{williamson2019spurious, kato2020toy}.
When the correlation range is $\xi$, the topological entanglement entropy needs to be measured on regions of radius $\mathcal{O}(\xi)$.
The sample complexity for measuring entropies on such regions grows exponentially in $\xi$.
Similarly, the mutual information $I(A:B) = S_{A} + S_B - S_{AB}$ for detecting long-range correlations in symmetry-breaking phases, and the conditional mutual information $I(A:B|C)$ for capturing long-range string correlations in SPT phases, both require measuring entropies on regions with radius $\mathcal{O}(\xi)$.
This leads to the same exponential scaling in the sample complexity, confirming that entanglement-based phase recognition approaches suffer from the same fundamental hardness as order-parameter approaches.

\paragraph{Implications for existing phase recognition algorithms.} Our hardness results establish the optimality of existing rigorous algorithms for phase recognition~\cite{huang2024learning,kim2024learning,landau2025learning}, which achieve runtime exponential in $\xi$.
This proves that no algorithm can fundamentally improve upon this scaling.
%
%
On the other hand, most common approaches to phase recognition are highly heuristic and leverage assumptions and prior knowledge about the quantum state of interest~\cite{jiang2012identifying,rodriguez2019identifying,cian2021many,herrmann2022realizing,cian2022extracting}.
Our results show that there exist quantum states where all such approaches must fail; we emphasize that this does not preclude their success in specific practical settings.


A recent class of algorithms based on the renormalization group represents a particularly interesting case~\cite{cong2019quantum,cong2024enhancing,lake2022exact}.
These algorithms operate in a hierarchal manner, first extracting small signals from local properties and gradually zooming out to larger scales to amplify these signals and determine the phase of matter. 
In particular, Ref.~\cite{lake2022exact} provides an explicit renormalization-group-inspired algorithm for recognizing phases of matter, and proves that the algorithm produces a correct answer for generic quantum states as long as the system size and circuit depth are sufficiently large.
They also show that this convergence occurs in circuit depth $d = \mathcal{O}(\log \xi)$ under various additional physical assumptions.
Our results do not contradict the main claims in Ref.~\cite{lake2022exact}. 
The mechanism behind the algorithm in Ref.~\cite{lake2022exact} is akin to nested majority voting, and its success hinges on the absence of near degeneracies in local voting outcomes. 
Our symmetric PRUs produce quantum states that exhibit degeneracies in the local voting outcomes that are exponentially small in $\xi$.
Amplifying these small signals requires a circuit depth $d = \mathcal{O}(\xi)$ in the algorithm of Ref.~\cite{lake2022exact}.
This translates to a number of gates $\exp(\mathcal{O}(\xi))$ owing to the algorithm's hierarchical architecture, consistent with our bound.
It remains open to rigorously prove if the algorithm in Ref.~\cite{lake2022exact} always attains this optimality without any assumptions.

\paragraph{What about non-symmetric ground states?} In our proof of Theorem~\ref{thm:quantum} and the examples above, we focus on states $\ket{\psi}$ that are symmetric in each phase.
This is always possible because there exists a symmetric representative of any phase of matter, since the ground state of a symmetric Hamiltonian is symmetric.
However, our results also imply strong limitations on recognizing non-symmetric representatives of a phase as well.
For example, consider the state $U \ket{0^n}$, which explicitly breaks the $\mathbbm{Z}_2$ symmetry $\prod_i X_i$, instead of $U \ket{\text{GHZ}}$ where $\ket{\text{GHZ}} = (\ket{0^n}+\ket{1^n})/\sqrt{2}$, which preserves the symmetry.
Our results prove that the only information one can efficiently extract from the state is the total population in each symmetry sector; see Appendix~\ref{sec: phases of matter} for details.
This allows distinguishing explicit symmetry-breaking states like $U \ket{0^n}$ from the trivial phase.
However, it still rules out distinguishing such states from non-symmetric states in other phases of matter, such as SPT phases~\cite{fn1}.
Moreover, Theorem~\ref{thm:quantum} implies that \emph{obtaining} the explicitly symmetry-broken state $U \ket{0^n}$ from the symmetric state $U \ket{\text{GHZ}}$ is computationally hard without knowledge of $U$.
If this were not true, then one could efficiently learn the phase of matter.

\paragraph{What about states generated from individually symmetric local gates?} Following our discussion of symmetric PRUs, our results for Abelian symmetry groups are slightly stronger than those for non-Abelian groups.
For Abelian symmetries, our results apply even when each individual constant-local gate in the unitary $U$ is required to be symmetric.
For non-Abelian symmetries, our results apply only when each individual $2\xi$-local patch is symmetric [see Fig.~\ref{fig:symPRU}(c)], but not necessarily each individual constant-local gate comprising each patch.
To achieve this in the Abelian case, we require a mild assumption that one can perform symmetric gates between nearest-neighbor triplets of qudits on which the symmetry acts; see Appendix~\ref{sec: symmetric PRUs}.
This avoids unusual counter-examples where no symmetric qudits can interact directly with one another~\cite{fn9}.

\subsection{Hardness of recognizing mixed state quantum phases of matter} \label{sec:mixedstate}

We conclude this section by discussing the extension of our results to mixed state quantum phases of matter.
Traditionally, mixed state phases have been studied in the context of finite-temperature Gibbs states of quantum or classical Hamiltonians~\cite{landau2013statistical}.
They have also received renewed interest recently~\cite{lee2023quantum,bao2023mixed,wang2025intrinsic,sang2024mixed,sang2025stability,coser2019classification,xue2024tensor,ma2025symmetry,yang2025topological}, sparked by experimental realizations affected by noise and decoherence.

The precise definition and classification of mixed state phases of matter remains an active research topic~\cite{lee2023quantum,bao2023mixed,wang2025intrinsic,sang2024mixed,sang2025stability,coser2019classification,xue2024tensor,ma2025symmetry,yang2025topological}.
One common definition proceeds as follows~\cite{sang2024mixed,coser2019classification,sang2025stability,xue2024tensor}.
Two mixed quantum states $\rho_1$ and $\rho_2$ are in the same phase of matter if and only if there exists a pair of local channels, $\mathcal{C}_1$ and $\mathcal{C}_2$, such that $\mathcal{C}_1(\rho_1) = \rho_2$ and $\mathcal{C}_2(\rho_2) = \rho_1$.
The use of quantum channels instead of unitary transformations is needed to enact equivalence transformations between mixed states with different spectra. 
The use of two channels instead of one follows because quantum channels are, in general, not invertible.
If the system is symmetric, then one also requires that each channel is covariant under the symmetry transformation, $\mathcal{C}(\cdot) = R_g^\dagger \mathcal{C}(R_g (\cdot) R_g^\dagger) R_g$ for all $g \in G$.
It was recently argued that additional smoothness constraints on the channels $\mathcal{C}_1$ and $\mathcal{C}_2$ should also be imposed~\cite{yang2025topological}.

Similar to the pure state setting, our results apply to any definition of mixed state phases of matter such that the following holds.
For each phase of matter, we can associate a representative fixed point mixed state $\rho_0$.
Then, we assume that any state $\rho$ that is related to $\rho_0$ by symmetric geometrically-local quantum channels as above is in the same phase as $\rho_0$.
We let $\ell$ denote the maximum allowed circuit depth of the channels and $\xi$ the light-cone radius.
We require only that the allowed channels include symmetric geometrically-local unitary circuits, and tensor product operations with ancilla qubits that are maximally mixed and unentangled with the system.

Our main result shows that mixed state phases of matter are as hard to recognize as pure state phases of matter.
Namely, we prove that it is quantum computationally hard to distinguish any mixed state phase of matter from a mixture of maximally mixed states in each symmetry charge sector.
The latter should be viewed as belonging to the trivial phase.
The population in each charge sector is inherited from the fixed point of the mixed state phase.

\begin{theorem}[Hardness of recognizing mixed state phases of matter] \label{thm:mixed}
    For any discrete on-site symmetry group $G$, distinguishing any mixed state phase of matter from a mixture of maximally mixed states in each symmetry charge sector requires quantum computational time scaling exponentially in the correlation range $\xi$. The time becomes super-polynomial in the system size $n$ whenever $\xi = \omega(\log n)$.
\end{theorem}

\noindent This theorem holds for all types of mixed state phases of matter, including both topological and symmetry-breaking mixed state phases.
This complements recent results on the hardness of recognizing strong-to-weak spontaneous symmetry breaking~\cite{feng2025hardness}.
As with Theorem~\ref{thm:quantum}, our proof of Theorem~\ref{thm:mixed} follows directly from our construction of low-depth symmetric PRUs.

\begin{figure}
    \centering
    \includegraphics[width=0.95\textwidth]{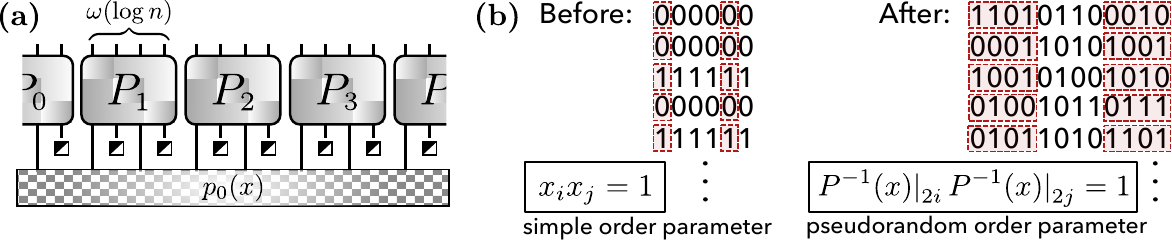}
    \caption{Extension of our results to classical phases of matter.
    \textbf{(a)} We consider a reversible one-layer brickwork circuit composed of symmetric pseudorandom permutations acting on patches of $\xi = \omega(\log n)$ bits each.
    The input to the circuit is any fixed point probability distribution, $p_0(x)$, tensored with $n$ mixed ancilla bits (black-white squares).
    We prove that for any input $p_0(x)$, the output of the circuit is indistinguishable from the maximally mixed distribution, i.e.~the trivial phase.
    \textbf{(b)} Samples from the $\mathbbm{Z}_2$-symmetry-breaking fixed point distribution before and after the circuit. The circuit transforms the order parameter from a simple single-body observable into a pseudorandom function on $2\xi = \poly(\log n)$ bits. This hides the phase of matter from any $\poly n$-time observer.
    }
    \label{fig:classical}
\end{figure}

\section{Classical phases of matter} \label{sec:classical}

While our primary focus is on quantum phases of matter, we find that, surprisingly, the hardness of recognizing phases of matter extends to purely classical states and phases as well.
The hallmark example of a classical phase of matter occurs in the two-dimensional Ising model, $H = \sum_{\langle i,j \rangle} Z_i Z_j$, at finite temperature $1/\beta$~\cite{landau2013statistical}.
The Gibbs state defines a classical probability distribution, 
\begin{equation}
    p_\beta(x) = \frac{\bra{x} e^{-\beta H} \ket{x}}{\sum_x \bra{x} e^{-\beta H} \ket{x}},
\end{equation}
where $x \in \{0,1\}^n$.
At high temperatures ($\beta < \beta_c \approx 0.441$), the distribution is in the trivial phase and features only short-range correlations.
At low temperatures ($\beta > \beta_c$), the distribution becomes symmetry-breaking, leading to long-range correlations across the entire system.

Similar to quantum systems, we can formalize the definition of classical phases of matter through local circuit equivalences.
Following our discussion of mixed state quantum phases of matter, two classical probability distributions $p_1(x)$ and $p_2(x)$ are in the same phase of matter if and only if there exists a pair of geometrically-local stochastic maps $\mathcal{M}_1$ and $\mathcal{M}_2$ such that $\mathcal{M}_1(p_1) = p_2$ and $\mathcal{M}_2(p_2) = p_1$.
We let $\ell$ denote the circuit depth of the maps and $\xi$ the light-cone radius.
This precisely mirrors the definition of mixed state quantum phases of matter, but with all objects restricted to be purely classical.
For systems possessing a symmetry, both maps must be covariant under the symmetry operation.
We refer to Appendix~\ref{sec: classical} for further details.

Our main result establishes the computational hardness of recognizing classical phases of matter:

\begin{theorem}[Hardness of recognizing classical phases of matter] \label{thm:classical}
    For any discrete on-site symmetry group $G$, recognizing whether a classical probability distribution is in a trivial versus symmetry-breaking phase requires a quantum or classical computational time scaling exponentially in the correlation range $\xi$. The time becomes super-polynomial in the system size $n$ whenever $\xi = \omega(\log n)$.
    \end{theorem}

\noindent We focus on symmetry-breaking phases since topological phases are not present in classical systems.

As in the quantum setting, we prove Theorem~\ref{thm:classical} by considering the complexity of recognizing a random example of each classical phase of matter.
Our random example is generated by applying a random reversible circuit $\mathcal{M}$ to a fixed point distribution $p_0(x)$ of each phase, as depicted in Fig.~\ref{fig:classical}(a).
The circuit $\mathcal{M}$ consists of only a single layer: a tensor product of $2n/\xi$ symmetric pseudorandom permutations (PRPs)~\cite{zhandry2016note} applied to disjoint patches of $\xi$ bits across the system.
Before applying the circuit $\mathcal{M}$, we tensor the fixed point distribution $p_0(x)$ with $n$ ancilla bits (one per physical bit), each initialized to be uniformly random (i.e.~in the maximally mixed state, $\mathbbm{1}/2$).
This operation adds entropy to the system without changing its phase of matter.

We prove Theorem~\ref{thm:classical} by showing that for any fixed point distribution $p_0(x)$, the resulting random ensemble is exponentially hard to distinguish from the maximally mixed distribution by any classical or quantum observer.
The intuition behind our result parallels the quantum case; nonetheless, it is somewhat unexpected given that classical circuits in one-dimension cannot appear globally random in any depth sublinear in the system size~\cite{schuster2024random}.
Consider the fixed point distribution for the $\mathbbm{Z}_2$-symmetry-broken phase: $p_0(x) = \frac{1}{2} \delta_{x, 0^n} + \frac{1}{2} \delta_{x, 1^n}$.
This features long-range correlations $\langle f_i(x) f_j(x) \rangle = 1$ with respect to the order parameter $f_i(x) = (-1)^{x_i}$.
After conjugating by PRPs on each patch, the distribution still features long-range correlations, but now with respect to a \emph{scrambled} order parameter $f'_i(x) = f_i(P^{-1}(x))$, where $P$ is the permutation acting on the patch containing bit $i$.
Since $P$ is a PRP, the scrambled order parameter $f'_i$ is a pseudorandom function on the patch, making it very difficult to learn for any classical or quantum algorithm.
Hence, as in the quantum case, the random probability distribution maintains genuine long-range symmetry-breaking order, yet the parameter signaling this order is impossible to detect efficiently.

The technical details of our proof are relatively straightforward.
We consider the distribution of bitstrings $x^{(1)},x^{(2)},\ldots, x^{(k)}$ obtained by querying the unknown distribution $k = \poly(n)$ times.
Due to the entropy of the ancilla bits, the $k$ bitstrings are \emph{distinct} from one another, i.e.~$x^{(i)}\neq x^{(j)}$ for all $i < j$~\cite{metger2024simple}, with high probability. 
Moreover, this distinctness also holds at the level of each individual patch: the set of $k$ $\xi$-bit strings obtained by restricting $x^{(1)},x^{(2)},\ldots, x^{(k)}$ to any patch of $\xi$ bits is distinct with high probability whenever $\xi = \omega(\log n)$~\cite{cui2025unitary}.
The key observation is that the probability of any given set of $k$ patchwise-distinct bitstrings appearing is uniform across all such sets due to the pseudorandom permutations on each patch.
Therefore, the distribution of interest cannot be distinguished from the uniform distribution over all sets of $k$ bitstrings, which corresponds to the trivial phase of matter. 
We refer to Appendix~\ref{sec: classical} for complete details.

Our classical construction relies crucially on the probability distributions of interest having a high local entropy, which is provided by the maximally mixed ancilla bits.
This high local entropy allows one to locally hide information about the order parameters of the phase.
This intuition suggests that a similar one-layer tensor-product circuit should be able to hide the phase of matter of highly-mixed quantum states as well.
However, this approach cannot extend to pure quantum states, where a simple purity test would distinguish any state generated by any tensor-product circuit from a truly random state.
One could also measure mutual information between distant patches of the quantum state to enable phase recognition.
Hence, the hardness of recognizing \emph{pure state} quantum phases of matter arises crucially from the ability of quantum circuits to appear globally random at very low depth, which does not occur in classical systems.

\section{Discussions} \label{sec:discussion}

Our results raise several questions regarding both the recognition of phases of matter, and symmetric random unitaries more broadly.
First, as discussed in the introduction, how does one reconcile our results with the observation that phases of matter are often efficiently recognizable in practice?
Is there a fundamental property of the ground states encountered in the physical world that guarantees this efficiency?
In particular, the most natural limitation in our results is that the locality of the parent Hamiltonian grows with the correlation range $\xi$.
However, even a $2$-local Hamiltonian can possess ground states with arbitrarily large correlation ranges $\xi$, since the $\xi$ will typically scale with the inverse of the spectral gap.
This raises a crucial question: Can one prove that the ground states of constant-local Hamiltonians exhibit different physical behaviors that enable efficient phase recognition with computational time scaling only polynomially in the correlation range $\xi$?
And if so, does this hold for every ground state, or only for a large class typically encountered in practice?

Indeed, in the same direction, arguments from conformal field theory suggest that the complexity of recognizing many phases of matter should grow only polynomially in the correlation length, and that this scaling can be achieved by measuring extremely simple fixed point order parameters~\cite{sachdev1999quantum}.
For example, the magnetization of a spin chain (a one-body Pauli operator) is predicted to decay polynomially in the correlation length $\xi$ as one approaches a phase transition.
Thus, one can detect the phase by measuring the magnetization with only a polynomial number of samples.
This prediction is tremendously simplifying, yet it goes far beyond our current mathematical understanding of quantum ground states and phases of matter via short-time adiabatic evolutions.
Is it possible to improve our understanding to encompass these simplifying properties?

Finally, beyond phases of matter, our results make progress on extending ideas from pseudorandom unitaries to physical systems that appear in the natural world.
This provides evidence that other important properties of random unitaries, such as quantum sampling advantages~\cite{arute2019quantum, morvan2023phase}, barren plateaus in variational algorithms~\cite{mcclean2018barren}, and classical shadow tomography~\cite{huang2020predicting}, should also persist at low depths in quantum systems with any discrete symmetry.
An important remaining question concerns quantum systems with continuous symmetries, such as charge or energy conservation.
While simple light-cone arguments rule out the complete formation of low-depth random unitaries in these settings~\cite{haah2025short}, is it possible that other features of low-depth random unitaries might persist?

\vspace{0.5em}
\subsection*{Acknowledgments:} We are grateful to  Samuel J. Garratt, Soumik Ghosh, Jeongwan Haah, Jonas Haferkamp, Isaac Kim, John Preskill, Norbert Schuch, and Nathanan Tantivasadakarn for valuable discussions and insights.
We are especially grateful to Shankar Balasubramanian, Soonwon Choi, Ethan Lake, and Yu-Jie Liu for several discussions and feedback that helped us to significantly improve the clarity of our results.
T.S. acknowledges
support from the Walter Burke Institute for Theoretical Physics at Caltech.
T.S. and H.H. acknowledge support the U.S. Department of Energy, Office of Science, National Quantum Information Science Research Centers, Quantum Systems Accelerator.
The Institute for Quantum Information and Matter, with which T.S. and H.H. are affiliated, is an NSF Physics Frontiers Center (NSF Grant PHY-2317110).
H.H. acknowledge support from the Broadcom Innovation Fund.

\vspace{2.5em}
\appendix

\addtocontents{toc}{\protect\setcounter{tocdepth}{2}}

\noindent 
\textbf{\LARGE{}Appendices}
\vspace{2em}

\noindent Our Appendices are organized as follows.
%
%
In Appendix~\ref{sec: symmetric random unitaries}, we present our main results on symmetric Haar-random unitaries and designs.
In Appendix~\ref{sec: symmetric PRUs}, we present our main results on symmetric pseudorandom unitaries.
In Appendix~\ref{sec: controlled PRUs}, we present our construction of controlled pseudorandom unitaries, which is used in our construction of symmetric pseudorandom unitaries in Appendix~\ref{sec: symmetric PRUs}.
In Appendix~\ref{sec: TI}, we present our results on translation-invariant random unitaries.
In Appendix~\ref{sec: phases of matter}, we present the implications of these results for recognizing quantum phases of matter.
In Appendix~\ref{sec: classical}, we present our independent results on the hardness of recognizing classical phases of matter.

\tableofcontents


\section{Symmetric random unitaries} \label{sec: symmetric random unitaries}

In this section, we present the full details of our results on gluing symmetric random unitaries.
We begin in Appendix~\ref{sec: background} with a basic introduction to  symmetry groups and their representations on qudit Hilbert spaces.
We then turn to random symmetric unitaries in Appendix~\ref{sec: exact haar symmetric}, and provide an explicit formula for the twirl over a Haar-random symmetric unitary.
In Appendix~\ref{sec: approx haar symmetric}, we present our first main result by providing a much simpler approximate expression for the symmetric Haar twirl.
Both our exact and approximate expressions hold for any unitary symmetry group, discrete or continuous.
In Appendix~\ref{sec: symmetric unitary designs}, we provide the definition of a symmetric unitary $k$-design, following the standard definition for unitary $k$-designs without symmetry.
In Appendix~\ref{sec: gluing symmetric}, we present our second main result by proving that random symmetric unitaries ``glue'' together whenever they overlap on a small region.
Our gluing lemma applies to any discrete on-site symmetry group, and forms the basis of our results on generating low-depth symmetric pseudorandom unitaries in Appendix~\ref{sec: symmetric PRUs}.
%

\subsection{Background on representations of symmetry groups} \label{sec: background}

We begin by providing a short introduction to the representation theory of symmetry groups in quantum many-body systems. We refer to standard textbooks for a comprehensive review~\cite{dresselhaus2007group}.

\emph{Unitary representations of symmetry groups.}---Consider any group $G$ and any finite Hilbert space $\mathcal{H}$.
A unitary representation $R$ of $G$ on $\mathcal{H}$ is a homomorphism mapping each element $g \in G$ to a unitary operator $R_g$ acting on $\mathcal{H}$.
A standard fact of mathematics states that any unitary representation on a finite Hilbert space can be decomposed as a tensor sum of \emph{irreducible representations}, or ``irreps'', as
\begin{equation} \label{eq: irrep decomp}
    \mathcal{H} = \bigoplus_\lambda  \big( \mathcal{M}_\lambda \otimes \mathcal{F}_\lambda \big), \,\,\,\,\,\,\,\, R_g = \bigoplus_\lambda \big( \mathbbm{1} \otimes R^\lambda_g \big).
\end{equation}
Here, the tensor sum is over all irreps $\lambda$ of the symmetry group. 
Each irrep $\lambda$ corresponds to a representation $R^\lambda$ of the group acting on a Hilbert space $\mathcal{F}_\lambda$.
By definition, the irreducible representation cannot be decomposed into a tensor sum of any sub-representations.
We denote the dimension of each irrep as $d_\lambda = |\mathcal{F}_\lambda|$.
Each irrep might also appear many times in the decomposition.
To capture this, each irrep is also associated with a Hilbert space $\mathcal{M}_\lambda$, whose dimension is equal to the number of times that the irrep appears, $D_\lambda \equiv | \mathcal{M}_\lambda |$.

In our work, we are interested in random symmetric unitary operators $U$ on $\mathcal{H}$, which commute with all elements of the symmetry group.
From Schur's lemma, any such unitary can be written in the block-diagonal form,
\begin{equation} \label{eq: U lambda}
    U = \bigoplus_\lambda \big( U_\lambda \otimes \mathbbm{1}_{\mathcal{F}_\lambda} \big),
\end{equation}
where each $U_\lambda$ acts on $\mathcal{M}_\lambda$.
The unitary acts trivially on the $\mathcal{F}_\lambda$ registers, which is necessary in order to commute with all symmetry group elements.
One can also easily see that any unitary of the form above is symmetric, i.e.~it commutes with all symmetry group elements.

\emph{Discrete symmetry groups.}---In quantum many-body systems, the Hilbert space factorizes into a tensor product of small components, $\mathcal{H} = \bigotimes_{i=1}^n \mathcal{H}_i$, where each $\mathcal{H}_i$ is the Hilbert space of an individual qubit or qudit of the system.
The most common symmetry operations in many-body physical systems fall into two categories.
In the first, \emph{on-site} symmetries, symmetry group elements act in a parallel on all qudits at a time.
That is, the representation, $R_g$, of a symmetry group element $g \in G$ on $\mathcal{H}$ is given by the tensor product of local representations of $g$ on each local Hilbert space, $R_g = \bigotimes_{i=1}^n R^i_g$.
In the second, \emph{lattice} symmetries, symmetry group elements act to permute the qudits (and their associated Hilbert spaces) amongst one another.
In this work, we will primarily focus on on-site symmetries. 
We extend our results to one particular lattice symmetry, translation-invariance, in Appendix~\ref{sec: TI}.

Let us now turn to the representation theory of discrete on-site symmetry groups in particular.
Any finite discrete group possesses a finite number of irreducible representations.
%
In our work, we focus for simplicity on the most common representation of a discrete symmetry group, known as the \emph{regular} representation.
As discussed in the following paragraphs, the tensor product of any representation of a symmetry group with a qudit in the regular representation yields a a larger regular representation. 
Hence, for our results on phases of matter, we can assume without loss of generality that the symmetry acts in the regular representation on each qudit. 
If it does not, then we can take the tensor product of each qudit with a symmetric ancilla qudit where the symmetry acts in the regular representation on the ancilla. 
This does not change the phase of matter, and yields a larger on-site qudit Hilbert space where the symmetry acts in the regular representation.

Let us now describe the regular representation and its properties.
For any finite discrete group $G$, the regular representation $R^{\text{reg}}$ of $G$ acts on a Hilbert space, $\mathcal{F}_{\text{reg}} = \{ \ket{g} | g \in G\}$.
The Hilbert space has dimension $|G|$, and is spanned by computational basis states labeled by each element $g \in G$.
The regular representation of $G$ acts on these basis states via right multiplication, $R^{\text{reg}}_g \ket{h} = \ket{g h}$.
One useful property of the regular representation is that the symmetry group operators, $R^{\text{reg}}_g$, are orthogonal to one another.
In particular, we have 
\begin{equation}
    \frac{1}{|G|} \tr( R^{\text{reg}}_g (R^{\text{reg}}_h)^{\dagger} ) = \frac{1}{|G|}\tr(R^{\text{reg}}_{gh^{-1}}) = \delta_{g,h},
\end{equation}
for any $g,h \in G$.
From this property, one can show that the regular representation has an irrep decomposition [Eq.~(\ref{eq: irrep decomp})] in which every irrep $\lambda$ appears with multiplicity equal to its dimension $d_\lambda$~\cite{dresselhaus2007group}.
This fact also yields the convenient formula, $\sum_\lambda d_\lambda^2 = |G|$~\cite{dresselhaus2007group}.

As previously mentioned, a second useful property of the regular representation concerns its behavior under tensor products.
Namely, the regular representation is a \emph{fixed point} under the tensor product.
To elaborate, consider the tensor product of the regular representation, $\mathcal{F}_{\text{reg}}, R_g^{\text{reg}}$, of a group $G$, with any other representation, $\mathcal{F}_2, R_g^2$, of $G$.
From these two original representations, one can naturally define a third representation, $R_g = R_g^{\text{reg}} \otimes R_g^2$, which acts on the combined Hilbert space $\mathcal{F}_{\text{reg}} \otimes \mathcal{F}_2$.
As for any representation, the tensor product representation can be decomposed as a sum of irreps.
For a general pair of original representations, this decomposition can be complicated to determine.
However, when either one of the input representations is regular, one finds that the tensor product representation has an irrep decomposition identical to that of the regular representation, up to an overall constant factor increase in the multiplicity.
That is, the irrep decomposition of the tensor product representation contains each irrep $\lambda$ of $G$ with multiplicity $d_\lambda |\mathcal{F}_2|$.
This implies that there exists an isomorphism, $\mathcal{F}_{\text{reg}} \otimes \mathcal{F}_2 \cong \mathbb{C}^{|\mathcal{F}_2|} \otimes \mathcal{F}_{\text{reg}}$, such that the symmetry group acts only on the rightmost register, $R_g = \mathbbm{1}_{|\mathcal{F}_2|} \otimes R^{\text{reg}}_g$.

We can now return to the quantum many-body setting.
Suppose that each local Hilbert space $\mathcal{H}_i$ transforms under the regular representation of $G$.
That is, we can write $\mathcal{H}_i \cong \mathcal{F}_{\text{reg},i} \otimes \mathcal{A}_i$, where $R_g^i = R_g^{\text{reg},i} \otimes \mathbbm{1}_{\mathcal{A}_i}$ and $\mathcal{A}_i$ is an ancilla register of arbitrary dimension $d'$.
Repeating the tensor product procedure in the previous paragraph $n-1$ times, one finds that the many-body Hilbert space $\mathcal{H} = \bigotimes_{i=1}^n \mathcal{H}_i$ can be decomposed as,
\begin{equation}
    \mathcal{H} = \bigotimes_{i=1}^n \mathcal{H}_i = \bigoplus_{\lambda} \big( \mathcal{M}_\lambda \otimes \mathcal{F}_\lambda \big), \,\,\,\,\,\,\,\,\,\,\,\,\, \text{where } \,\,\, D_\lambda \equiv |\mathcal{M}_\lambda| = d_\lambda |G|^{n-1} d'^n = d_\lambda D / |G|,
\end{equation}
where $D = |G|^n d'^n$ is the many-body Hilbert space dimension, and the symmetry operator $R_g = \bigotimes_{i=1}^n R_g^i$ acts as $R_g = \bigoplus_\lambda (\mathbbm{1}_{\mathcal{M}_\lambda} \otimes R^\lambda_g)$.
As a sanity check, one can easily compute $\sum_\lambda d_\lambda D_\lambda = (D/|G|) \sum_\lambda d_\lambda^2 = D$.
If desired, one can also write the many-body Hilbert space as $\mathcal{H} \cong \mathbb{C}^{D/|G|} \otimes \mathcal{F}_{\text{reg}}$, where $R_g$ acts in the regular representation on the second register.
This shows that the symmetry operators remain orthogonal to one another on the many-body Hilbert space, $\tr(R_g R_h^{-1}) = D \delta_{g,h}$, as one would expect.

\emph{Additional notation.}---We conclude this introduction by setting up several definitions for the analysis that follows.
We let $P_\lambda$ denote the projector onto the irrep $\lambda$, i.e.~$P_\lambda = \mathbbm{1}_{\mathcal{M}_\lambda} \otimes \mathbbm{1}_{\mathcal{F}_\lambda}$.
We also define a set of computational basis states spanning each irrep Hilbert space, $\mathcal{F}_\lambda = \{ \ket{ a } | a = 1,\ldots,d_\lambda \}$.
Any choice of basis is valid.
We let $P^\lambda_{k\ell}$ denote the transition operator between basis states $a$ and $b$, i.e.~$P^\lambda_{ab} = \mathbbm{1}_{\mathcal{M}_\lambda} \otimes \dyad{k}{\ell}$.
Both the set of $P^\lambda_{k \ell}/\sqrt{D_\lambda}$ and the set of $R_g/\sqrt{D}$ form complete orthonormal bases for the set of operators in $\text{span}\{ R_g \}$~\cite{dresselhaus2007group}.
This implies the equality,
\begin{equation} \label{eq: R P regular}
    \sum_{\lambda k \ell} \frac{1}{D_\lambda} \tr( (\cdot) P^{\lambda}_{\ell k}) P^{\lambda}_{k \ell} = \frac{1}{D} \sum_g \tr( (\cdot) R_g^{-1} ) R_g.
\end{equation}
This will be useful for exchanging summations over the two sets of operators later on.

\subsection{The exact symmetric Haar twirl} \label{sec: exact haar symmetric}

In this section, we derive an explicit and exact formula for the twirl over a symmetric Haar-random unitary.
Our results apply to any unitary representation $R$ of any symmetry group $G$ on any Hilbert space $\mathcal{H}$.

Recall from the previous subsection [Eq.~(\ref{eq: U lambda})] that any symmetric unitary $U$ can be written in the form,
\begin{equation} 
    U = \bigoplus_\lambda \big( U_\lambda \otimes \mathbbm{1}_{\mathcal{F}_\lambda} \big),
\end{equation}
where $U_\lambda \in U(D_\lambda)$, the unitary group on $D_\lambda$ states.
The Haar measure on the group of symmetric unitaries is the unique measure that is invariant under left and right multiplication, $U \rightarrow U V$ and $U \rightarrow V U$, by any symmetric unitary $V$.
One can easily see that this corresponds to drawing each $U_\lambda$ independently from the Haar measure on $U(D_\lambda)$.
We denote the resulting $G$-symmetric Haar ensemble as $H_G$.




We begin by providing an exact formula for the twirl over $k$ copies of a symmetric Haar-random unitary.
This exact formula is rather complicated; we will provide a much simpler approximate formula in the following subsection.
To set up notation,  we decompose the $k$-copy Hilbert space into a tensor product of $k$ irreps, $\mathcal{H}^{\otimes k} = \bigoplus_{\bs \lambda} ( \mathcal{M}_{\bs \lambda} \otimes \mathcal{F}_{\bs \lambda} )$, where $\bs{\lambda} = (\lambda_1,\ldots,\lambda_k)$ is a length-$k$ vector denoting the irrep in each of the $k$ copies, and $\mathcal{M}_{\bs \lambda} = \bigotimes_{j=1}^k  \mathcal{M}_{\lambda_j}$, $\mathcal{F}_{\bs \lambda} = \bigotimes_{j=1}^k  \mathcal{F}_{\lambda_j}$.
We define the basis operators $P^{\bs{\lambda}}_{\bs{k} \bs{\ell}} = \bigotimes_{j=1}^k P^{\lambda_j}_{k_j \ell_j}$, where $\bs k, \bs \ell$ are also length-$k$ vectors.

\begin{proposition}[Exact symmetric Haar twirl]\label{prop: exact Haar}
    For any unitary representation of any symmetry group $G$.
    The twirl over $k$ copies of a symmetric Haar-random unitary is,
    \begin{equation} \label{eq: exact Haar}
    \Phi^{G}_{H}(\rho) \equiv \E_{U \sim H_G} \left[ U^{\otimes k} \rho U^{\dagger, \otimes k} \right]
    =
    \sum_{\sigma}
    \sum_{\bs{\lambda} \bs{k} \bs{\ell}}
    \sum_{ \,\,\,\,\, \tau_{\alpha_\lambda} }
    \tr( \rho \sigma^{-1} P^{\bs{\lambda}}_{\sigma(\bs{\ell}) \bs{k}} ) \cdot  \prod_\lambda \Wg(\tau_{\alpha_\lambda}; D_\lambda) \cdot P^{\bs{\lambda}}_{\bs{k} \tau(\bs{\ell})} \, \tau,
    \end{equation}
    for any $k \leq \min_\lambda D_\lambda$.
    Here, the permutation $\tau \in S_k$ is defined as $\tau = \big( \bigotimes_{\alpha_\lambda} \tau^{}_{\alpha_\lambda} \big) \sigma \in S_k$, where each $\tau_{\alpha_\lambda} \in S_{k_\lambda}$ permutes the copies $\alpha_\lambda = \{ j | \lambda_j = \lambda\}$ that contain a given irrep $\lambda$, and $k_\lambda = |\alpha_\lambda|$.
\end{proposition}
\noindent We recall that $\Wg(\tau_{\alpha_\lambda}; D_\lambda) = \Wg_{\mathbbm{1}_{\alpha_\lambda},\tau_{\alpha_\lambda}}(D_\lambda)$ denotes the Weingarten matrix element between $\tau_{\alpha_\lambda}$
 and the identity permutation on $k_\lambda$ copies.
 
\begin{proof}
We write $\rho = \bigoplus_{\bs{\lambda},\bs{\lambda}'} \rho_{\bs{\lambda},\bs{\lambda}'}$ and $U = \bigoplus_{\bs{\lambda}} ( U_{\bs{\lambda}} \otimes \mathbbm{1}_{\mathcal{F}_{\bs{\lambda}}} )$, where $U_{\bs{\lambda}} = \otimes_{i=1}^k U_{\lambda_i}$.
We have,
\begin{equation} \label{eq: Haar block by block}
\begin{split}
    \Phi^{G}_{H}(\rho) & \equiv \bigoplus_{\bs{\lambda},\bs{\lambda}'} \E_{U_{\bs{\lambda}} \sim H} \left[ ( U_{\bs{\lambda}} \otimes \mathbbm{1}_{\mathcal{F}_{\bs{\lambda}}} ) \, \rho_{\bs{\lambda},\bs{\lambda}'} \, ( U_{\bs{\lambda}'}^\dagger \otimes \mathbbm{1}_{\mathcal{F}_{\bs{\lambda}'}} )  \right]  \\
    & \equiv \bigoplus_{\bs{\lambda},\bs{\lambda}'} \bigg( \prod_\lambda \delta_{k^{}_\lambda,k'_\lambda}  \bigg)  \E_{U_{\bs{\lambda}} \sim H} \left[ ( U_{\bs{\lambda}} \otimes \mathbbm{1}_{\mathcal{F}_{\bs{\lambda}}} ) \, \rho_{\bs{\lambda},\bs{\lambda}'} \, ( U_{\bs{\lambda}'}^\dagger \otimes \mathbbm{1}_{\mathcal{F}_{\bs{\lambda}'}} )  \right]  \\
    & = \bigoplus_{\bs{\lambda},\bs{\lambda}'} \bigg( \prod_\lambda \delta_{k^{}_\lambda,k'_\lambda}  \bigg)  \sum_{\sigma_{\alpha'_\lambda \rightarrow \alpha_{\lambda}}} 
    \sum_{\tau_{\alpha'_\lambda \rightarrow \alpha_{\lambda}}} 
    \tr_{\mathcal{M}_{\bs{\lambda}}}( \rho_{\bs{\lambda},\bs{\lambda}'} \, \sigma^{-1}_{\mathcal{M}} )
    \cdot 
    \prod_\lambda \Wg_{\sigma_{\alpha'_\lambda \rightarrow \alpha_{\lambda}}, \tau_{\alpha'_\lambda \rightarrow \alpha_{\lambda}}}(D_\lambda)
    \cdot
    \tau_{\mathcal{M}} \\
    & = \bigoplus_{\bs{\lambda}} \sum_{\sigma} \sum_{\,\,\,\, \tau_{\alpha_\lambda}} 
    \tr_{\mathcal{M}_{\bs{\lambda}}}( \rho_{\bs{\lambda},\sigma^{-1}(\bs{\lambda})} \, \sigma^{-1}_{\mathcal{M}} )
    \cdot 
    \prod_\lambda \Wg(\tau_{\alpha_\lambda}; D_\lambda)
    \cdot
    \tau_{\mathcal{M}} \\
    & = \bigoplus_{\bs{\lambda}} \sum_{\sigma} \sum_{\,\,\,\, \tau_{\alpha_\lambda}}  
    \tr( \rho_{\bs{\lambda},\sigma^{-1}(\bs{\lambda})} \, \big[ \sigma^{-1}_{\mathcal{M}} \otimes \dyad{\bs{\ell}}{\bs{k}}_{\mathcal{F}_{\bs{\lambda}}} \big] )
    \cdot 
    \prod_\lambda \Wg(\tau_{\alpha_\lambda}; D_\lambda)
    \cdot
    \big[ \tau_{\mathcal{M}} \otimes \dyad{\bs{k}}{\bs{\ell}}_{\mathcal{F}_{\bs{\lambda}}} \big] \\
    & = \bigoplus_{\bs{\lambda}} \sum_{\sigma} \sum_{\,\,\,\, \tau_{\alpha_\lambda}}  
    \tr( \rho_{\bs{\lambda},\sigma^{-1}(\bs{\lambda})} \, \sigma^{-1} \big[ \mathbbm{1}_{\mathcal{M}_{\bs{\lambda}}} \otimes \dyad{\sigma(\bs{\ell})}{\bs{k}}_{\mathcal{F}_{\bs{\lambda}}} \big] )
    \cdot 
    \prod_\lambda \Wg(\tau_{\alpha_\lambda}; D_\lambda)
    \cdot
    \big[ \mathbbm{1}_{\mathcal{M}_{\bs{\lambda}}} \otimes \dyad{\bs{k}}{\tau(\bs{\ell})}_{\mathcal{F}_{\bs{\lambda}}} \big] \tau \\
    & = \sum_{\bs{\lambda}  \bs{k} \bs{\ell}} \sum_{\sigma} \sum_{\,\,\,\, \tau_{\alpha_\lambda}}  
    \tr( \rho \, \sigma^{-1} P^{\bs{\lambda}}_{\sigma(\bs{\ell}),\bs{k}} )
    \cdot 
    \prod_\lambda \Wg(\tau_{\alpha_\lambda}; D_\lambda)
    \cdot
    P^{\bs{\lambda}}_{\bs{k},\tau(\bs{\ell})} \tau
\end{split}
\end{equation}
In the second line, we note that the expectation value vanishes unless there are an equal number $k_\lambda$ of $U_\lambda$ as there are $k'_\lambda$ of $U_\lambda^\dagger$ (where $k'_\lambda$ denotes the number of copies with irrep $\lambda$ in $\bs \lambda'$).
In the third line, we compute the Haar average over $k_\lambda$ copies for each irrep $\lambda$. The expression in terms of Weingarten functions is valid as long as $k_\lambda \leq D_\lambda$ for each $\lambda$; this is guaranteed when $k \leq \min_\lambda D_\lambda$.
This produces a sum over permutations $\sigma = \bigotimes_\lambda \sigma_{\alpha'_\lambda \rightarrow \alpha_\lambda}$, $\tau = \bigotimes_\lambda \tau_{\alpha'_\lambda \rightarrow \alpha_\lambda}$, which map the copies $\alpha'_\lambda$ with irrep $\lambda$ in $\bs \lambda'$ to the copies $\alpha_\lambda$ with irrep $\lambda$ in $\bs \lambda$. The permutations act on the $\mathcal{M}$ registers; specifically, they are isometries between $\mathcal{M}_{\bs \lambda'}$ and $\mathcal{M}_{\bs \lambda}$.
In the fourth line, we note that we can combine the sums over $\bs \lambda'$ (with $k'_\lambda = k_\lambda$) and $\sigma_{\alpha'_\lambda \rightarrow \alpha_\lambda}$ into a single summation over all $\sigma \in S_k$. 
We also re-label the summation over $\tau$, using $\tau_{\alpha_\lambda} \equiv \tau_{\alpha'_\lambda \rightarrow \alpha_\lambda} (\sigma_{\alpha'_\lambda \rightarrow \alpha_\lambda})^{-1}$.
In the fifth line, we insert the complete basis of operators on $\mathcal{F}_{\bs \lambda}$.
In the sixth line, we insert the identity $\mathbbm{1} = \sigma^{-1} \sigma$ on the $\mathcal{F}_{\bs \lambda}$ register, and similar for $\tau$.
In the final line, we re-write the expression using the operators $P^{\bs \lambda}_{\bs k \bs \ell}$.
This concludes the proof.
\end{proof}

While our proof of Proposition~\ref{prop: exact Haar} is direct and fairly short, it can be somewhat opaque due to the various tensor sums and products.
To this end, we also present a second, more heuristic derivation of Proposition~\ref{prop: exact Haar} below.
Our method is inspired by a similar technique for deriving the Haar twirl without symmetries.
In particular, we first enumerate the operators that commute with $U^{\otimes k}$ for every $U$.
We then define $\Phi^{G}_{H}(\rho)$ as the map which projects $\rho$ onto the span of these operators.

Loosely speaking, there are two classes of operators that commute with any symmetric $U^{\otimes k}$.
First, the permutation operators $\sigma \in S_k$ commute with any tensor product unitary, symmetric or not.
Second, since $U$ is symmetric, any tensor product of local symmetry operators $R_{\bs g} = \otimes_{j=1}^k R_{g_j}$ also commutes with $U^{\otimes k}$.
Any product of an operator from the first set and the second set also commutes with $U^{\otimes k}$.
Hence, the commutant is equal to the set of all such products, $\{ R_{\bs g} \, \sigma \}$.
In what follows, it will be more convenient to replace the local symmetry operators $R_g$ with the irrep basis elements $P^\lambda_{k\ell}$. 
This is valid because any $R_g$ can be written as a linear combination of $P^\lambda_{k\ell}$ and vice versa.
Thus, the commutant is equivalently spanned by the set $\{ P^{\bs \lambda}_{\bs k \bs \ell} \, \sigma \}$.

To project an operator $\rho$ onto the commutant, we need to know the overlaps of each operator in the commutant.
These are given by the following expression,
\begin{equation} \label{eq: G P}
    G^{\bs{\lambda} \bs{k} \bs{\ell},\bs{\lambda}' \bs{k}' \bs{\ell}'}_{\sigma, \tau} = \tr( P^{\bs{\lambda}}_{\bs{k} \bs{\ell}} \, \sigma  \tau^{-1} P^{\bs{\lambda}'}_{\bs{\ell}' \bs{k}'} ) = \delta_{\bs{\lambda},\bs{\lambda}'}  \delta_{\bs{k},\bs{k}'} \cdot \delta_{\sigma^{-1}(\bs{\lambda}),\tau^{-1}(\bs{\lambda})} \delta_{\sigma^{-1}(\bs{\ell}),\tau^{-1}(\bs{\ell}')} \cdot \prod_\lambda G(\sigma^{}_\lambda \tau^{-1}_\lambda;D_\lambda),
\end{equation}
where $\sigma_\lambda, \tau_\lambda$ refer to the restriction of $\sigma, \tau$ to input and output indices in representation $\lambda$ (i.e.~$\sigma_\lambda$ is a shorthand for $\sigma_{\sigma^{-1}(\alpha_\lambda) \rightarrow \alpha_\lambda}$ in the notation used in our proof of Proposition~\ref{prop: exact Haar}).
When viewed as a matrix, the inverse of $G^{\bs{\lambda} \bs{k} \bs{\ell},\bs{\lambda}' \bs{k}' \bs{\ell}'}_{\sigma, \tau}$ is,
\begin{equation} \label{eq: W P}
    \Wg^{\bs{\lambda} \bs{k} \bs{\ell},\bs{\lambda}' \bs{k}' \bs{\ell}'}_{\sigma, \tau} = \delta_{\bs{\lambda},\bs{\lambda}'} \delta_{\bs{k},\bs{k}'} \cdot \delta_{\tau^{-1}(\bs{\lambda}),\sigma^{-1}(\bs{\lambda})} \delta_{\tau^{-1}(\bs{\ell}),\sigma^{-1}(\bs{\ell}')} \cdot \prod_\lambda \Wg(\sigma_\lambda \tau^{-1}_\lambda;D_\lambda).
\end{equation}
We will verify this in the following paragraph.
The inverse matrix allows us to immediately write an expression for the symmetric Haar twirl,
\begin{equation} \label{eq: Phi W}
    \Phi^{G}_{H}(\rho)
    =
    \sum_{\sigma,\tau}
    \sum_{\bs{\lambda} \bs{k} \bs{\ell}}
    \sum_{\bs{\lambda}' \bs{k}' \bs{\ell}'}
    \tr( \rho \sigma^{-1} P^{\bs{\lambda}}_{\bs{\ell} \bs{k}} ) \cdot  \Wg^{\bs{\lambda} \bs{k} \bs{\ell},\bs{\lambda}' \bs{k}' \bs{\ell}'}_{\sigma, \tau} \cdot P^{\bs{\lambda}'}_{\bs{k}' \bs{\ell}'} \, \tau.
\end{equation}
One can see that $\Phi^{G}_{H}(\rho)$ is the unique map that projects $\rho$ onto the span of the commutant.
Inserting our expression for $\Wg^{\bs{\lambda} \bs{k} \bs{\ell},\bs{\lambda}' \bs{k}' \bs{\ell}'}_{\sigma, \tau}$ into the above, and re-indexing $\bs{\ell} \rightarrow \sigma(\bs{\ell})$, we find,
\begin{equation} 
    \Phi^{G}_{H}(\rho)
    =
    \sum_{\sigma,\tau}
    \sum_{\bs{\lambda} \bs{k} \bs{\ell}}
    \delta_{\tau^{-1}(\bs{\lambda}),\sigma^{-1}(\bs{\lambda})} \cdot
    \tr( \rho \sigma^{-1} P^{\bs{\lambda}}_{\sigma(\bs{\ell}) \bs{k}} ) \cdot  \prod_\lambda \Wg(\sigma_\lambda \tau^{-1}_\lambda;D_\lambda) \cdot P^{\bs{\lambda}}_{\bs{k} \tau(\bs{\ell})} \, \tau.
\end{equation}
We can then note that the delta function $\delta_{\tau^{-1}(\bs{\lambda}),\sigma^{-1}(\bs{\lambda})}$ amounts to restricting the sum over $\tau$ to those of the form $\tau = \big( \bigotimes_{\alpha_\lambda} \tau^{}_{\alpha_\lambda} \big) \sigma \in S_k$ (as in Proposition~\ref{prop: exact Haar}).
Implementing this substitution, i.e.~replacing $\sum_\tau \delta_{\tau^{-1}(\bs{\lambda}),\sigma^{-1}(\bs{\lambda})} \rightarrow \sum_{\tau_{\alpha_\lambda}}$ and $\sigma_\lambda \tau^{-1}_\lambda \rightarrow \tau^{-1}_{\alpha_\lambda}$, the expression above reduces precisely to our desired formula, Eq.~(\ref{eq: exact Haar}).

We conclude by showing that $\Wg^{\bs{\lambda} \bs{k} \bs{\ell},\bs{\lambda}' \bs{k}' \bs{\ell}'}_{\sigma, \tau}$ is indeed the inverse of $G^{\bs{\lambda} \bs{k} \bs{\ell},\bs{\lambda}' \bs{k}' \bs{\ell}'}_{\sigma, \tau}$.
We compute,
\begin{equation}
\begin{split}
    \sum_\tau \sum_{\bs{\lambda}' \bs{k}' \bs{\ell}'} & G^{\bs{\lambda} \bs{k} \bs{\ell},\bs{\lambda}' \bs{k}' \bs{\ell}'}_{\sigma, \tau} \Wg^{\bs{\lambda}' \bs{k}' \bs{\ell}',\bs{\lambda}'' \bs{k}'' \bs{\ell}''}_{\tau, \pi} \\
    & = 
    \delta_{\bs{\lambda},\bs{\lambda}''} \delta_{\bs{k},\bs{k}''} \delta_{\sigma^{-1}(\bs{\ell}),\pi^{-1}(\bs{\ell}'')} \delta_{\sigma^{-1}(\bs{\lambda}),\pi^{-1}(\bs{\lambda})}
    \cdot \sum_\tau 
    \delta_{\sigma^{-1}(\bs{\lambda}),\tau^{-1}(\bs{\lambda})} \cdot
     \prod_\lambda G_{\sigma_\lambda, \tau_\lambda}(D_\lambda) \Wg_{\tau_\lambda, \pi_\lambda}(D_\lambda) \\
     & = 
    \delta_{\bs{\lambda},\bs{\lambda}''} \delta_{\bs{k},\bs{k}''} \delta_{\sigma^{-1}(\bs{\ell}),\pi^{-1}(\bs{\ell}'')} \delta_{\sigma^{-1}(\bs{\lambda}),\pi^{-1}(\bs{\lambda})}
    \cdot \prod_\lambda \bigg( \sum_{\tau_\lambda}
      G_{\sigma_\lambda, \tau_\lambda}(D_\lambda) \Wg_{\tau_\lambda, \pi_\lambda}(D_\lambda) \bigg) \\
      & = 
    \delta_{\bs{\lambda},\bs{\lambda}''} \delta_{\bs{k},\bs{k}''} \delta_{\sigma^{-1}(\bs{\ell}),\pi^{-1}(\bs{\ell}'')} \delta_{\sigma^{-1}(\bs{\lambda}),\pi^{-1}(\bs{\lambda})}
    \cdot \prod_\lambda \delta_{\sigma_\lambda,\pi_\lambda} \\
    & = 
    \delta_{\bs{\lambda},\bs{\lambda}''} \delta_{\bs{k},\bs{k}''} \delta_{\bs{\ell},\bs{\ell}''} \delta_{\sigma,\pi}. \\
\end{split}
\end{equation}
In the first line, we note the delta functions in Eq.~(\ref{eq: G P}) and Eq.~(\ref{eq: W P}) enforce $\bs{\lambda}=\bs{\lambda}'=\bs{\lambda}''$, $\bs{k}=\bs{k}'=\bs{k}''$, and $\sigma^{-1}(\bs{\ell})=\tau^{-1}(\bs{\ell}')=\pi^{-1}(\bs{\ell}'')$.
This allows us to compute the sums over $\bs{\lambda}',\bs{k}',\bs{\ell}'$, which eliminates the delta functions involving these variables.
In the second line, the condition $\tau^{-1}(\bs{\lambda}) = \sigma^{-1}(\bs{\lambda})$ restricts the sum over $\tau$ to permutations such that $\tau^{-1}$ maps each copy $i$ with an irrep $\lambda_i$ to a copy $j$ such that $\lambda_{\sigma(j)} = \lambda_i$.
In other words, the permutation $\tau$ maps each subset of input copies with an irrep $\lambda$, $\{ i : \lambda_i = \lambda\}$, to a fixed subset of output copies $\{ j : \lambda_{\sigma(j)} = \lambda \}$.
Thus, the sum over permutations $\tau$ reduces to a product of individual sums, over permutations $\tau_\lambda \in S_{k_\lambda}$ that map $\{ i : \lambda_i = \lambda\}$ to $\{ j : \lambda_{\sigma(j)} = \lambda \}$, where $k_\lambda = | \{ i : \lambda_i = \lambda\} | = | \{ j : \lambda_{\sigma(j)} = \lambda \} |$.
In the third and fourth line, we find that each sum over $\tau_\lambda$ produces a delta function $\delta_{\sigma_\lambda,\pi_\lambda}$, which causes the remaining expression to reduce to the identity matrix.
This completes our alternate derivation of Proposition~\ref{prop: exact Haar}.

\subsection{Approximate symmetric Haar twirl} \label{sec: approx haar symmetric}

The expression for the exact symmetric Haar twirl is somewhat daunting, owing to the numerous summations and Weingarten matrix elements.
In this section, we show that a much simpler formula suffices when the Hilbert space dimension is large.
Our approximate formula holds for any discrete symmetry group $G$ acting in the regular representation on the Hilbert space.
%
%

\begin{lemma}[Approximate symmetric Haar twirl] \label{lemma: approximate symmetric Haar twirl}
    For any discrete on-site symmetry group $G$.
    The symmetric Haar twirl $\Phi^{G}_{H}$ is approximated by the map,
    \begin{equation} \label{eq: approx Haar}
    \Phi^{G}_{a}(\rho) 
    =
    \frac{1}{D^k} \sum_{\sigma}
    \sum_{\bs{g}}
    \tr( \rho \sigma^{-1} R_{\bs{g}}^{-1} ) \,  R_{\bs{g}} \, \sigma,
    \end{equation}
    up to relative error, $(1-\varepsilon)\Phi^G_H \preceq \Phi^G_\mathcal{E} \preceq (1+\varepsilon) \Phi^G_H$, with $\varepsilon = |G|k^2/D$ for any $k^2 \leq D/|G|$.
    Here, $\mathcal{A} \preceq \mathcal{B}$ denotes that $\mathcal{B}-\mathcal{A}$ is a completely positive map.
\end{lemma}
\noindent We remark that the approximate symmetric Haar twirl can be equivalently written in terms of the irreducible subspace projectors, $P^{\bs{\lambda}}_{\bs{k} \bs{\ell}}$, as,
\begin{equation} \label{eq: approx Haar P}
\Phi^{G}_{a}(\rho) 
=
\sum_{\sigma}
\sum_{\bs{\lambda} \bs{k} \bs{\ell}}
\tr( \rho \sigma^{-1} P^{\bs{\lambda}}_{\bs{\ell} \bs{k}} ) \cdot \frac{1}{\prod_{\lambda} D_\lambda^{k_\lambda}} \cdot P^{\bs{\lambda}}_{\bs{k} \bs{\ell}} \, \sigma.
\end{equation}
This follows from writing $R_{\bs g} = \bigotimes_{i=1}^n R_{g_i}$ and $P^{\bs \lambda}_{\bs k \bs \ell} =  \bigotimes_{i=1}^n P^{\lambda_i}_{k_i \ell_i}$, and noting that within each copy, we have $\frac{1}{D} \sum_{g} \tr( (\cdot) R_g^{-1}) R_g =  \sum_{\lambda k \ell} \frac{1}{D_\lambda} \tr( (\cdot) P^\lambda_{\ell k} ) P^\lambda_{k \ell}$ from Eq.~(\ref{eq: R P regular}). 

We will break our proof of Lemma~\ref{lemma: approximate symmetric Haar twirl} into two parts.
First, we provide a general lemma for bounding the relative error of any random unitary ensemble with respect to the approximate symmetric Haar twirl, $\Phi^{G}_{a}$.
This lemma will also be used in our proof of the gluing lemma for symmetric unitaries (Lemma~\ref{lemma: AB BC to ABC app}) later on.
Second, we apply this lemma to prove Lemma~\ref{lemma: approximate symmetric Haar twirl}.


\subsubsection{Bounding the relative error via the additive error on the EPR state}

In practice, it is usually much easier to bound the additive error of a random unitary ensemble rather than the relative error.
To this end, we provide a helpful lemma for converting from additive to relative error, with respect to the approximate symmetric Haar twirl, $\Phi^{G}_{a}$.

\begin{lemma}[Relative error from EPR states] \label{lemma: relative error to additive}
    For any discrete on-site symmetry group $G$.
    Let $\mathcal{E}$ be any symmetric unitary ensemble, and $\Phi_{\mathcal{E}}$ its twirl on $\mathcal{H}^{\otimes k}$.
    The twirl is approximated by $\Phi^{G}_{a}$ up to relative error, $(1-\varepsilon) \Phi^{G}_{a} \preceq \Phi_{\mathcal{E}} \preceq (1+\varepsilon) \Phi^{G}_{a}$, where
    \begin{equation} \label{eq: relative error EPR}
        \varepsilon = \frac{k! |G|^k}{D^{2k}} \big\lVert \big[ \delta \Phi  \otimes \mathbbm{1} \big](P_{\text{\emph{EPR}}})\big\rVert_\infty.
    \end{equation}
    Here, $\delta \Phi \equiv \Phi_{\mathcal{E}} - \Phi^{G}_{a}$, and $P_{\text{\emph{EPR}}}$ is the projector onto the EPR state on $\mathcal{H}^{\otimes k} \otimes \mathcal{H}^{\otimes k}$.
\end{lemma}
\noindent Our proof closely resembles that of Lemma~2 in Ref.~\cite{schuster2024random} for random unitaries without symmetries.
\begin{proof}
    Let $\rho^G_a \equiv [\Phi^{G}_{a} \otimes \mathbbm{1}](P_{\text{EPR}})$ and $\rho_{\mathcal{E}} \equiv [\Phi_{\mathcal{E}} \otimes \mathbbm{1}](P_{\text{EPR}})$. Our proof proceeds in four steps.

    \vspace{3mm}
    \noindent (1) $\rho^G_a$ is equal to $k! |G|^k/D^{2k}$ times a projector,
    \begin{equation}
        P \equiv \frac{1}{k! |G|^k} \sum_{\sigma}
    \sum_{\bs{g}}
    R^{}_{\bs{g}} \, \sigma \otimes R^{*}_{\bs{g}} \, \sigma
    \end{equation}
    Hence, on the projected subspace, $\rho^G_a$ has a flat spectrum with eigenvalue  $k! |G|^k/D^{2k}$.
    To show that $P$ is a projector, we compute that it squares to one,
    \begin{equation} \label{eq: P squares to 1}
\begin{split}
    P^2
    & =
    \frac{1}{k!^2 |G|^{2k}} \sum_{\sigma,\sigma'} \sum_{\bs{g},\bs{g}'}
    R^{}_{\bs{g}} \, \sigma \, R^{}_{\bs{g}'} \, \sigma' \otimes 
    R^{*}_{\bs{g}} \, \sigma \, R^{*}_{\bs{g}'} \, \sigma' \\
    & =
    \frac{1}{k!^2 |G|^{2k}} \sum_{\sigma,\sigma'} \sum_{\bs{g},\bs{g}'}
    R^{}_{\bs{g} \cdot \sigma(\bs{g}')} \, \sigma \sigma' \otimes 
    R^{*}_{\bs{g} \cdot \sigma(\bs{g}')} \, \sigma \sigma' \\
    & =
    \frac{1}{k!^2 |G|^{2k}} \sum_{\sigma,\sigma'} \sum_{\bs{g},\bs{g}'}
    R^{}_{\bs{g} \cdot \bs{g}'} \, \sigma \sigma' \otimes 
    R^{*}_{\bs{g} \cdot \bs{g}'} \, \sigma \sigma' \\
    & =
    \frac{1}{k! |G|^{k}} \sum_{\sigma} \sum_{\bs{g}}
    R^{}_{\bs{g} } \, \sigma \otimes 
    R^{*}_{\bs{g} } \, \sigma  =  P.  \\
\end{split}
\end{equation}
In the second line, we re-index $\bs{g}' \rightarrow \sigma^{-1}(\bs{g}')$.
In the fourth line, we compute the sums over $\sigma'$ and $\bs{g}'$ by re-indexing $\sigma \rightarrow \sigma \sigma'^{-1}$ and $\bs{g} \rightarrow \bs{g} \bs{g}'^{-1}$. The sums then give $\sum_{\sigma'} 1 = k!$ and $\sum_{\bs{g}'} 1 = |G|^k$.

    \vspace{3mm}
    \noindent (2) $\rho_{\mathcal{E}}$ has support entirely within the support of $P$. 
    This follows because $P (U^{\otimes k} \otimes \mathbbm{1}) \ket{\Psi_{\text{EPR}}} = (U^{\otimes k} \otimes \mathbbm{1}) P \ket{\Psi_{\text{EPR}}} = (U^{\otimes k} \otimes \mathbbm{1}) \ket{\Psi_{\text{EPR}}}$ for any $U \sim \mathcal{E}$, where $\dyad{\Psi_{\text{EPR}}} \equiv P_{\text{EPR}}$.
    The first equality follows because any symmetric tensor product unitary $U^{\otimes k}$ commutes with any permutation $\sigma$ and any tensor product $R_{\bs{g}}$ of symmetry operators.
    The second equality follows because $(R^{}_{\bs{g}} \sigma \otimes R^{*}_{\bs{g}} \sigma) \ket{\Psi_{\text{EPR}}} = (R_{\bs{g}} \sigma \sigma^{-1} R^{-1}_{\bs{g}} \otimes \mathbbm{1}) \ket{\Psi_{\text{EPR}}} = \ket{\Psi_{\text{EPR}}}$.

    \vspace{3mm}
    \noindent (3) Steps (1) and (2) immediately imply that the twirl has relative error $\varepsilon$ [Eq.~(\ref{eq: relative error EPR})] on the EPR state.

    \vspace{3mm}
    \noindent (4) The relative error on the EPR state upper bounds the relative error on any state.
    This follows because we can express $\Phi(\rho) = D^{2k}\tr_2( (\mathbbm{1} \otimes \rho^T) [ \Phi \otimes \mathbbm{1} ](P_{\text{EPR}}))$ for any $\Phi$, where  the trace is over the second copy of $\mathcal{H}^{\otimes k}$.
\end{proof}

\subsubsection{Proof of Lemma~\ref{lemma: approximate symmetric Haar twirl}}

With Lemma~\ref{lemma: relative error to additive} in hand, we will now prove Lemma~\ref{lemma: approximate symmetric Haar twirl}.
Let $\rho^G_a = [\Phi^G_{a}\otimes \mathbbm{1}](P_{\text{EPR}})$ and $\rho^G_H = [\Phi^{G}_{H} \otimes \mathbbm{1}](P_{\text{EPR}})$ denote the approximate and exact Haar twirls applied to the EPR state.
For the purposes of this proof, it will be convenient to write $\rho^G_H$ by inserting the EPR state $\rho \rightarrow P_{\text{EPR}}$ into the third line of Eq.~(\ref{eq: Haar block by block}).
To do so, we first decompose the EPR state into a tensor sum of EPR states on each irreducible representation, 
\begin{equation} \label{eq: EPR irrep}
    P_{\text{EPR}} = \bigoplus_{\bs \lambda, \bs \lambda'} \bigg( \dyad*{\Psi_{\text{EPR}}^{\mathcal{M}_{\bs \lambda}}}{\Psi_{\text{EPR}}^{\mathcal{M}_{\bs \lambda'}}} \otimes \dyad*{\Psi_{\text{EPR}}^{\mathcal{F}_{\bs \lambda}}}{\Psi_{\text{EPR}}^{\mathcal{F}_{\bs \lambda'}}} \bigg) \cdot \prod_\lambda (d_\lambda D_\lambda / D)^{k_\lambda/2} \cdot \prod_{\lambda'} (d_{\lambda'} D_{\lambda'} / D)^{k'_{\lambda'}/2}. 
\end{equation}
Here, $\bs \lambda$ labels the irrep of the ket and $\bs{\lambda}'$ of the bra, and $k_\lambda, k'_{\lambda'}$ indicate the multiplicity of an irrep $\lambda, \lambda'$ in $\bs \lambda, \bs \lambda'$.
The multiplicative factors at the right ensure that each irrep has trace equal to its dimension $d_\lambda D_\lambda$ divided by the total Hilbert space dimension $D$ (since the EPR state has flat support over the total Hilbert space).
This yields an expression for $\rho^G_H$,
\begin{equation} \label{eq: Haar approx derivation}
\begin{split}
    \rho^G_H
    & = \bigoplus_{\bs{\lambda}} \sum_{\sigma, \tau_{\alpha_\lambda} }
    \tr_{\mathcal{M}_{\bs{\lambda}}}( [ P_{\text{EPR}} ]_{\bs{\lambda},\sigma^{-1}(\bs{\lambda})} \, [\sigma^{-1}]_{\mathcal{M}_{\sigma^{-1}(\bs \lambda)} \rightarrow \mathcal{M}^{}_{\bs \lambda}}  )
    \cdot 
    \prod_\lambda \Wg(\tau_{\alpha_\lambda}; D_\lambda)
    \cdot
    [\tau]_{\mathcal{M}^{}_{\bs{\lambda}} \rightarrow \mathcal{M}_{\sigma^{-1}(\bs \lambda)}}, \\
    & = \frac{|G|^k}{D^{2k}} \bigoplus_{\bs{\lambda}} \sum_{\sigma, \tau_{\alpha_\lambda}}  \left( \prod_\lambda \left( D_\lambda^{k_\lambda} \Wg(\tau_{\alpha_\lambda}; D_\lambda) \right) \cdot 
    \bigg( [\tau \otimes \sigma ]_{\mathcal{M}^{\otimes 2}_{\bs{\lambda}} \rightarrow \mathcal{M}^{\otimes 2}_{\sigma^{-1}(\bs \lambda)}} \otimes \dyad*{\Psi_{\text{EPR}}^{\mathcal{F}_{\bs{\lambda}}}}{\Psi_{\text{EPR}}^{\mathcal{F}_{\sigma^{-1}(\bs{\lambda})}}} \bigg)  \right) \\
    & = \frac{|G|^k}{D^{2k}} \sum_{\sigma} \left( \bigoplus_{\bs{\lambda}} \left( \sum_{\,\,\,\, \tau_{\alpha_\lambda}} \prod_\lambda \left( D_\lambda^{k_\lambda} \Wg(\tau_{\alpha_\lambda}; D_\lambda) \right) \cdot 
    \bigg( [ (\tau \sigma^{-1}) \otimes \mathbbm{1}]_{\mathcal{M}^{\otimes 2}_{\bs{\lambda}}} \otimes \dyad*{\Psi_{\text{EPR}}^{\mathcal{F}_{\bs{\lambda}}}}{\Psi_{\text{EPR}}^{\mathcal{F}_{\bs{\lambda}}}} \bigg)  \right)\right) (\sigma \otimes \sigma). \\
\end{split}
\end{equation}
We add subscripts to the various operators to make it clear which subspaces they map between.
In the second line, we apply the identity,
$\tr_{\mathcal{M}_{\bs{\lambda}}}( [ P_{\text{EPR}} ]_{\bs{\lambda},\sigma^{-1}(\bs{\lambda})} \, [\sigma^{-1}]_{\mathcal{M}_{\sigma^{-1}(\bs \lambda)} \rightarrow \mathcal{M}^{}_{\bs \lambda}}  ) = \prod_\lambda (1/D_\lambda)^{k_\lambda} \cdot [\sigma]_{\mathcal{M}^{}_{\bs \lambda} \rightarrow \mathcal{M}_{\sigma^{-1}(\bs \lambda)}}$, to take the partial trace over $\mathcal{M}_{\bs \lambda}$ on the left side of the the EPR projector.
This follows from the standard identity for the partial trace over the left half of an EPR pair, $\tr_1( [A \otimes \mathbbm{1}] P_{\text{EPR}} ) = A^T / D$.
We also use $\prod_\lambda (d_\lambda D_\lambda / D)^{k_\lambda} \cdot \prod_\lambda (1/ D_\lambda)^{k_\lambda} = (|G|^k/D^k) \prod_\lambda D_\lambda^{k_\lambda}$ to simplify the constant factors arising from Eq.~(\ref{eq: EPR irrep}) and this expression.
In the final line, we pull out a factor $\sigma \otimes \sigma$, which will be convenient later on.

We can obtain an analogous expression for the EPR state under the approximate Haar twirl, by replacing $\prod_\lambda D_\lambda^{k_\lambda} \Wg(\tau_{\alpha_\lambda}; D_\lambda) \rightarrow \delta_{\sigma,\tau}$.
(To be precise, this follows from Eq.~(\ref{eq: approx Haar P}) and manipulations exactly analogous to the final four lines of Eq.~(\ref{eq: Haar block by block}).)
This yields,
\begin{equation} \label{eq: Haar approx derivation 2}
\begin{split}
    \rho^G_a
    & = \frac{|G|^k}{D^{2k}} \sum_{\sigma} \left( \bigoplus_{\bs{\lambda}} 
    \big( \mathbbm{1}_{\mathcal{M}^{\otimes 2}_{\bs{\lambda}}}  \otimes \dyad*{\Psi_{\text{EPR}}^{\mathcal{F}_{\bs{\lambda}}}}{\Psi_{\text{EPR}}^{\mathcal{F}_{\bs{\lambda}}}} \big) \right) (\sigma \otimes \sigma).  \\
\end{split}
\end{equation}
Taking the difference between the two states and applying the triangle inequality, with $\lVert \bigoplus_{\bs{\lambda}} A_{\bs{\lambda}} \rVert_\infty = \max_{\bs{\lambda}} \lVert A_{\bs{\lambda}} \rVert_\infty$, we find
\begin{equation} \label{eq: Haar approx derivation 3}
\begin{split}
    \lVert \rho^G_H - \rho^G_a \rVert_\infty
    & \leq \frac{|G|^k}{D^{2k}} \sum_\sigma \max_{\bs{\lambda} } \left( \sum_{\,\,\,\, \tau_{\alpha_\lambda}}
    \bigg| \prod_\lambda D_\lambda^{k_\lambda} \Wg(\tau_{\alpha_\lambda}; D_\lambda)  - \delta_{\sigma,\tau} \bigg| \right) \\
    & = \frac{|G|^k}{D^{2k}} \sum_\sigma \max_{\bs{\lambda} } \left( \sum_{\,\,\,\, \tau_{\alpha_\lambda}}
    \bigg(  \prod_\lambda D_\lambda^{k_\lambda} \big| \! \Wg(\tau_{\alpha_\lambda}; D_\lambda) \big| - \delta_{\sigma,\tau} \bigg) \right) \\
    & = \frac{|G|^k}{D^{2k}} \sum_\sigma \max_{\bs{\lambda} } \left( 
      \prod_\lambda D_\lambda^{k_\lambda} \frac{(D_\lambda - k_\lambda)!}{D_\lambda!}  - 1 \right) \\
    & \leq \frac{|G|^k}{D^{2k}} \sum_\sigma \max_{\bs{\lambda} } \left( 
      \prod_\lambda (1+k_\lambda^2/D_\lambda)  - 1 \right) \\
    & = \frac{|G|^k}{D^{2k}} \sum_\sigma \frac{k^2}{\min_\lambda D_\lambda}  = \frac{k!|G|^k}{D^{2k}} \frac{|G| k^2}{D}. \\
\end{split}
\end{equation}
In the second line, we use that the diagonal Weingarten elements are greater than one, $D_\lambda^{k_\lambda} \Wg(\mathbbm{1}_{\alpha_\lambda}; D_\lambda) \geq 1$, and in the third line, we use that $\sum_{\tau_{\alpha_\lambda}} | \! \Wg(\tau_{\alpha_\lambda}; D_\lambda)| = (D_\lambda-k_\lambda)!/D_\lambda! \leq 1+k_\lambda^2/D_\lambda$~\cite{aharonov2021quantum}.
In the final line, we note that the maximum is achieved when $\bs{\lambda}$ contains $k$ copies of the same value of $\lambda$, for $\lambda$ such that $D_\lambda$ is minimal.
We also have $\min_\lambda D_\lambda = \min_\lambda d_\lambda D / |G| \geq D / |G|$.
Applying Lemma~\ref{lemma: relative error to additive} completes our proof. \qed



\subsection{Symmetric unitary designs} \label{sec: symmetric unitary designs}

In many cases, one would like quantify whether a given symmetric random unitary ensemble is ``close'' to the symmetric Haar ensemble.
One way to do so is through the notion of a symmetric unitary $k$-design.
We say that an ensemble of symmetric unitaries forms an $\varepsilon$-approximate symmetric unitary $k$-design if it reproduces the $k$-th moment of the symmetric Haar ensemble up to relative error $\varepsilon$.
\begin{definition}[Approximate symmetric unitary designs] \label{def: strong unitary design}
For any symmetry group $G$.
A symmetric unitary ensemble $\mathcal{E}$  is an $\varepsilon$-approximate symmetric unitary $k$-design if
\begin{align}
    (1-\varepsilon) \, \Phi^{G}_{H} \, \preceq \, \Phi_{\mathcal{E}} \, \preceq \, (1+\varepsilon) \, \Phi^{G}_{H},
\end{align}
where the quantum channel $\Phi_{\mathcal{E}}(\cdot)$ is defined via
\begin{equation} \label{eq:twirling-channel}
	\Phi_{\mathcal{E}}(A) \coloneqq \E_{U \sim \mathcal{E}} \left[ U^{\otimes k} A U^{\dagger, \otimes k} \right],
\end{equation}
and similarly for the symmetric Haar ensemble. 
\end{definition}
\noindent We will use the notion of symmetric unitary designs to state our gluing lemma in the following subsection.

\subsection{Proof of Lemma~\ref{lemma:gluing}: Gluing symmetric random unitaries} \label{sec: gluing symmetric}

We now arrive at one of our main results on symmetric random unitaries, the symmetric gluing lemma (Lemma~\ref{lemma:gluing}).
The lemma shows that small symmetric random unitaries can ``glue'' together to form larger symmetric random unitaries in much the same manner as random unitaries without symmetries~\cite{schuster2024random}.
Let first state a formal version of Lemma~\ref{lemma:gluing} with tighter error bounds.
%

\begin{lemma}[Gluing two symmetric random unitaries; formal version of Lemma~\ref{lemma:gluing}] \label{lemma: AB BC to ABC app}
For any discrete on-site symmetry group $G$.
Let $A$, $B$, $C$ be three disjoint subsystems. 
Consider a symmetric random unitary given by $V_{ABC} = U_{AB} U_{BC}$, where $U_{AB}$ and $U_{B C}$ are drawn from $\varepsilon_{AB}$ and $\varepsilon_{BC}$-approximate symmetric unitary $k$-designs, respectively. 
Then $V_{ABC}$ is an $\varepsilon$-approximate symmetric unitary $k$-design with
\vspace{-2mm}
\begin{equation} \label{eq: ABC error}
    1+\varepsilon = (1+\varepsilon_{AB})(1 + \varepsilon_{BC}) 
    \left( 1-\frac{k^2|G|}{2 D_{AB}} \right)^{-1}
    \left( 1-\frac{k^2|G|}{2 D_{BC}} \right)^{-1}
    e^{k(k-1)|G|/2D_B} \left( 1+ \frac{k^2 |G|}{D_{ABC}} \right),
\end{equation}
for any $k^2 \leq D_B/|G|$. Here, $D_\alpha = 2^{|\alpha|}$ is the Hilbert space dimension of subsystem $\alpha$.
\end{lemma}
\noindent Our proof of the gluing lemma utilizes the simple approximate expression for the symmetric Haar twirl derived in the previous subsections.

From the gluing lemma, we can immediately show that symmetric two-layer symmetric circuit forms a unitary $k$-design.
\begin{theorem}[Gluing small symmetric random unitaries] \label{thm:main-design}
    For any discrete on-site symmetry group and any approximation error $\varepsilon \leq 1$.
    Suppose each small random unitary in the two-layer brickwork ensemble $\mathcal{E}$ is drawn from an $\frac{\varepsilon}{n}$-approximate symmetric unitary $k$-design on $2\xi$ qubits with circuit depth $d$. Then $\mathcal{E}$ forms an $\varepsilon$-approximate symmetric unitary $k$-design on $n$ qubits with depth $2d$, whenever the local patch size satisfies $\xi \geq \log_{|G|}(nk^2|G|/\varepsilon)$.
\end{theorem}
\noindent This will form the basis of our proof of Theorem~\ref{thm:polylog}, on low-depth symmetric PRUs, in the next section.

In what follows, we first prove the symmetric gluing lemma and then the two-layer circuit theorem.

\begin{proof}[Proof of Lemma~\ref{lemma: AB BC to ABC app}]
    Let $\delta_{AB} = k^2|G|/2D_{AB}/(1-k^2|G|/2D_{AB})$, $\delta_{BC} = k^2|G|/2D_{BC}/(1-k^2|G|/2D_{BC})$; elementary algebra gives $1+\delta_{AB} = (1-k^2|G|/2D_{AB})^{-1}$ and similar for $\delta_{BC}$.
    From Lemma~\ref{lemma: approximate symmetric Haar twirl}, we can approximate the twirl over $V_{ABC}$ by $(\Phi_{a}^{G})_{AB} \circ (\Phi_{a}^{G})_{BC}$ up to relative error $(1+\varepsilon_{AB})(1+\delta_{AB})(1+\varepsilon_{BC})(1+\delta_{BC})-1$.
    To proceed, let
    \begin{equation}
        \rho^G_a = [\Phi^{G}_{a} \otimes \mathbbm{1}](P_{\text{EPR}}) = \frac{1}{D^{2k}} \sum_{\sigma}
    \sum_{\bs{g}}
    R_{\bs{g}} \, \sigma \otimes R^*_{\bs{g}} \, \sigma,
    \end{equation}
    and
    \begin{equation}
    \begin{split}
        \rho_{\mathcal{E},a} & \equiv \big[ \big( (\Phi_{a}^{G})_{AB} \circ (\Phi_{a}^{G})_{BC} \big) \otimes \mathbbm{1} \big] ( P_{\text{EPR}} ) \\
        & = \frac{1}{D^k} \frac{1}{D_{AB}^k D_{BC}^k} \sum_{\sigma \tau} \sum_{\bs{g} \bs{g}'}
    G^{\bs{g},\bs{g}'}_{\sigma,\tau} \cdot
    [ ( R_{\bs{g}} \sigma)_{AB} \otimes (R_{\bs{g}'} \tau)_{C} ] \otimes [ (R^*_{\bs{g}} \sigma)_A \otimes (R^*_{\bs{g}'} \tau)_{BC} ],
    \end{split}
    \end{equation}
    where we denote $G^{\bs{g},\bs{g}'}_{\sigma,\tau} = \tr_B( (R^{}_{\bs{g}} \sigma^{} \tau^{-1} R_{\bs{g}'}^{-1})_B )$.
    We note that the diagonal terms in $\rho_{\mathcal{E},a}$, with $\sigma = \tau$,  have $G^{\bs{g},\bs{g}'}_{\sigma,\sigma} = \prod_{i=1}^k \tr_{B_i}( R^{}_{g^{}_i} R^{-1}_{g_i'} ) = D_B^k \, \delta_{\bs{g},\bs{g}'}$, and thus yield precisely the expression for $\rho^G_a$.
    Applying the triangle inequality to the remaining terms, which have $\sigma \neq \tau$, we find
    \begin{equation} \label{eq: sym gluing proof}
    \begin{split}
        \lVert \rho^G_a - \rho_{\mathcal{E},a} \rVert_\infty
        & \leq
        \frac{1}{D^k} \frac{1}{D_{AB}^k D_{BC}^k} \sum_{\sigma \neq \tau} \sum_{\bs{g} \bs{g}'} \big| G^{\bs{g},\bs{g}'}_{\sigma,\tau} \big| \\
        & =
        \frac{1}{D^k} \frac{1}{D_{AB}^k D_{BC}^k} \sum_{\sigma \neq \tau}
        \sum_{\bs{g} \bs{g}'}
        \bigg| \prod_{\ell \in \sigma\tau^{-1}} \tr_B \bigg( \prod_{i \in \ell} R^{-1}_{g'_i} R^{}_{g_i} \bigg) \bigg| \\
        & \leq
        \frac{1}{D^k} \frac{1}{D_{AB}^k D_{BC}^k} \sum_{\sigma \neq \tau}
        \sum_{\bs{g} \bs{g}'}
        D_B^{k-|\sigma \tau^{-1}|}
         \prod_{\ell \in \sigma\tau^{-1}} \delta\bigg( \prod_{i \in \ell} g'_i g_i^{-1} = e\bigg) \\
         & =
        \frac{1}{D^{2k}} \sum_{\sigma \neq \tau}
        D_B^{-|\sigma \tau^{-1}|}
         |G|^{k+|\sigma \tau^{-1}|} \\
         & =
        \frac{k! |G|^k}{D^{2k}} \sum_{\sigma \neq \mathbbm{1}}
        (|G|/D_B)^{|\sigma|} \leq \frac{k!|G|^k}{D^{2k}} \left( e^{k(k-1) |G|/2 D_B} -1 \right). \\
    \end{split}
    \end{equation}
    In the second line, we write $G^{\bs{g},\bs{g}'}_{\sigma,\tau}$ as a product of traces over each cycle $\ell$ in $\sigma \tau^{-1}$.
    In the third line, we note that the trace over each cycle is zero unless the product of group elements appearing the cycle is the identity.
    Moreover, when this condition is satisfied, the trace is equal to $D_B$.
    The total number of cycles is $\#(\sigma \tau^{-1}) = k-|\sigma\tau^{-1}|$.
    In the fourth line, we compute the sum over $\bs{g}, \bs{g}'$. This yields $|G|^{2k-\#(\sigma \tau^{-1})} = |G|^{k+|\sigma \tau^{-1}|}$, where $\#(\sigma \tau^{-1})$ counts the number of delta functions that restrict the sum.
    In the final line, we use the well-known bound $\sum_\sigma D^{-|\sigma|} \leq e^{k(k-1)/2D}$ for $D \rightarrow D_B / |G|$~\cite{harrow2023approximate, schuster2024random}.

    We have $\lVert \rho_{\mathcal{E}} - \rho^G_a \rVert_\infty \leq (1+\varepsilon_{AB})(1+\delta_{AB})(1+\varepsilon_{BC})(1+\delta_{BC}) \lVert \rho_{\mathcal{E},a} - \rho^G_a \rVert_\infty$, since $\rho_{\mathcal{E},a}$ approximates $\rho_{\mathcal{E}}$ up to relative error.
    Applying Lemma~\ref{lemma: relative error to additive} (which bounds the relative error between $\Phi_\mathcal{E}$ and $\Phi^G_a$) and Lemma~\ref{lemma: approximate symmetric Haar twirl} (which bounds the relative error between $\Phi^G_a$ and $\Phi^G_H$) completes the proof.
\end{proof}

\begin{proof}[Proof of Theorem~\ref{thm:main-design}]
    Now that the symmetric gluing lemma is established, our proof exactly mirrors that of the proof of Theorem~1 in Ref.~\cite{schuster2024random}. 
    We apply Lemma~\ref{lemma: AB BC to ABC app} patch-by-patch $m-2$ times, where $m = n/\xi$ is the number of local patches.
    After all $m-2$ applications, one finds that the two-layer ensemble forms an approximate symmetric unitary $k$-design with error
    \begin{equation} \label{eq: total error}
    \begin{split}
        (1+\varepsilon/n)^{m-1} & \cdot \left( \left( 1-k^2|G|/2q^2 \right)^{-1}  \left( 1-k^2|G|/2q^2 \right)^{-1}  e^{k(k-1)|G|/2q}  \big(1+k^2|G|/q^3 \big) \right)^{m-2} - 1, \\
    \end{split}
    \end{equation}
    where we let $q \equiv |G|^{\xi}$ denote the Hilbert space dimension of a single patch. We also use $D_B \geq q$, $D_{AB}, D_{BC} \geq q^2$, and $D_{ABC} \geq q^3$.
    Using the inequalities $1/(1-x/2) \leq 1 + x$  and $1+x \leq e^x$ several times, the above expression is upper bounded by,
    \begin{equation} \label{eq: total error 2}
    \begin{split}
        \left( \ldots \right) \leq e^{(m-1) \varepsilon/n + (m-2) f(k,q)} -1  \leq \frac{1}{\log(2)} \left( (m-1) \varepsilon / n + (m-2) f(k,q) \right).
    \end{split}
    \end{equation}
    The second inequality follows from $e^{x}-1 \leq x/\log(2)$ for $x \leq \log(2)$, and we abbreviate,
    \begin{equation}
        f(k,q) = \frac{k(k-1)|G|}{2q} + 2 \frac{k^2|G|}{q^2} + \frac{k^2|G|}{q^3},
    \end{equation}
    which arise from the four error terms in the gluing lemma.
    
    We would like to show that the total error, Eq.~(\ref{eq: total error}), is less than $\varepsilon$.
    We take $k \geq 2$ and $n \geq 3\xi$ since otherwise the theorem holds trivially.
    By assumption, we have $q \geq nk^2|G|/\varepsilon$.
    Combined with $k \geq 2$, $n \geq 3\xi$, $|G| \geq 1$, and $\varepsilon \leq 1$, this implies that $\xi \geq 5$ and so $q \geq 32$.
    The first term in Eq.~(\ref{eq: total error}) is therefore less than,
\begin{equation}
    \frac{(m-1) \varepsilon}{ n \log(2) } \leq \frac{\varepsilon}{ \xi \log(2) }
    \leq \frac{\varepsilon}{5 \log(2) } = 0.29 \varepsilon,
\end{equation}
since $m \leq n / \xi \leq n/5$.
Applying $q \geq nk^2|G|/\varepsilon$ to the second term in Eq.~(\ref{eq: total error 2}), we find
\begin{equation}
\begin{split}
    \frac{(m-2)}{\log(2)} f(k,q) & \leq \frac{n}{\xi \log(2)}   \left( \frac{\varepsilon}{n} + 2\frac{\varepsilon}{nq} + \frac{\varepsilon}{n q^2}  \right) \leq  \varepsilon \cdot \frac{1.07}{5 \log(2)} = 0.31 \varepsilon.
\end{split}
\end{equation}
These error bounds are exactly the same as in Ref.~\cite{schuster2024random} for random unitaries without symmetries.
The sum of the two error terms is less than $0.6 \varepsilon \leq \varepsilon$ as claimed.
\end{proof}





\section{Symmetric pseudorandom unitaries} \label{sec: symmetric PRUs}

In this section, we present the full details of our results on symmetric pseudorandom unitaries.
We begin in Appendix~\ref{sec: sym PRU notation} by reviewing several key definitions and notation.
In Appendix~\ref{sec: sym PRU}, we present our first construction of symmetric pseudorandom unitaries, which has circuit depth $\poly n$ for any discrete on-site symmetry group $G$.
We further prove that our construction can be compiled in terms of constant-local symmetric gates whenever $G$ is Abelian.
In Appendix~\ref{sec: sym PRUs low depth}, we combine this result with our gluing theorem from the previous section.
This leads to our second construction of symmetric pseudorandom unitaries, which has circuit depth $\poly(\log n)$ on any circuit geometry.

\subsection{Definition of symmetric PRUs} \label{sec: sym PRU notation}

We begin by stating the formal definition of a symmetric pseudorandom unitary ensemble.
Namely, a symmetric unitary ensemble is a symmetric PRU ensemble if it is indistinguishable from the symmetric Haar ensemble by any bounded-time quantum experiment.
\begin{definition}[Symmetric pseudorandom unitaries] \label{def: PRU-t(lambda)}
For any symmetry group $G$ acting on $n$ qubits.
An infinite sequence $\{ \mathcal{U}_n \}_{n \in \mathbb{N}}$ of $n$-qubit symmetric unitary ensembles $\mathcal{U}_n = \{ U_{\mathsf{key}} \}_{\mathsf{key} \in \mathcal{K}_{n}}$ for the key space $\mathcal{K}_{n}$ is a symmetric pseudorandom unitary secure against any $t(n)$-time adversary if it satisfies the following.
\begin{itemize}
    \item \textbf{Efficient computation:} There exists a $\mathrm{poly}(n)$-time quantum algorithm that implements the $n$-qubit unitary $U_{\mathsf{key}}$ for all $\mathsf{key} \in \mathcal{K}_{n}$.
    \item \textbf{Indistinguishability from Haar:} Any quantum algorithm $\mathcal{A}^{U}(1^{n})$ that runs in time $\leq t(n)$, queries an $n$-qubit unitary $U$ for any number of times, and outputs $\{0, 1\}$ satisfies
    \begin{equation} \label{eq:distinguishability-adv}
        \left| \Pr_{\mathsf{key} \sim \mathcal{K}_{n}}\left[\mathcal{A}^{U_{\mathsf{key}}}(1^{n}) = 1 \right] - \Pr_{U \sim H_G}\left[ \mathcal{A}^{U}(1^{n}) = 1 \right] \right| \leq \mathrm{negl}(n),
    \end{equation}
    where $H_G$ is the symmetric Haar ensemble and $\mathrm{negl}(n)$ is the negligible function, i.e., a function that is asymptotically smaller than any inverse-polynomial function $1 / \mathrm{poly}(n)$.
\end{itemize}
A $t(n)$-time algorithm $\mathcal{A}^{U}(1^{n})$ is referred to as a $t(n)$-time adversary.
The difference between the probability under $\mathcal{U}_n$ and the symmetric Haar ensemble is referred to as the advantage of $\mathcal{A}^{U}(1^{n})$.
\end{definition}
\noindent This is identical to the definition of a pseudorandom unitary without symmetry when $G$ is trivial.

Our construction of symmetric PRUs will use the notion of a controlled pseudorandom unitary, which are defined as follows.
We prove the existence of controlled PRUs in Appendix~\ref{sec: controlled PRUs}.
\begin{definition}[Controlled pseudorandom unitaries]
    A unitary ensemble is a controlled pseudorandom unitary ensemble secure against any $t(n)$-time adversary if it cannot be distinguished from the controlled Haar ensemble by any $t(n)$-time quantum experiment.
\end{definition}
\noindent For brevity, we state the definition colloquially; the full details follow identically to Definition~\ref{def: PRU-t(lambda)}.
Here, the controlled Haar ensemble refers to a random unitary $U^c = \dyad{0} \otimes \mathbbm{1} + \dyad{1} \otimes U$, where the first qubit is a control register and $U$ is drawn from the Haar measure on the remaining $n$ qubits.

\subsection{Proof of Theorem~\ref{thm:poly}: Symmetric PRUs in polynomial depth} \label{sec: sym PRU}

We can now present our polynomial-depth construction of symmetric pseudorandom unitaries.
Our construction applies to any discrete on-site symmetry group $G$. 
Our construction can be compiled in circuit depth $\poly n$.
%
This being said, the individual circuit gates in our compilation may not necessarily respect the symmetry operation.
For Abelian discrete on-site symmetries in particular, we prove a stronger statement that our construction can be compiled in circuit depth $\poly n$ using only symmetric geometrically 5-local gates.

Our construction is simple following the irreducible decomposition of $\mathcal{H}$.
We proceed in three steps.
First, we apply a unitary circuit that performs the irreducible decomposition, $\mathcal{H} \rightarrow \bigoplus_\lambda \mathcal{M}_\lambda \otimes \mathcal{F}_\lambda$.
%
%
We also compute the irrep index $\lambda$ onto a small additional register of qubits.
Second, conditioned on the small register being in the state $\lambda$, we apply a controlled pseudorandom unitary on the register $\mathcal{M}_\lambda$.
%
%
We repeat this for each irrep $\lambda$.
Finally, we apply the inverse of the unitary circuit that was applied in the first step.
%

%
Our main result is that this construction forms a symmetric pseudorandom unitary.
\begin{theorem}[Symmetric pseudorandom unitaries] \label{thm: sym PRU}
For any discrete on-site symmetry group $G$,
the unitary ensemble described above is a symmetric pseudorandom unitary secure against any sub-exponential time quantum adversary.
\end{theorem}
\noindent We prove the theorem in Appendix~\ref{sec: sym PRU proof} below. 
Our proof relies on several key ingredients whose proofs are provided later sections:
the efficient implementation of pseudorandom functions on general local tensor product spaces (Appendix~\ref{sec: PRFs LTP}), the efficient generation of unitary $k$-designs on general local tensor product spaces (Appendix~\ref{sec: designs LTP}), and, leveraging these two ingredients, the existence of controlled PRUs (Appendix~\ref{sec: controlled PRUs}).
We separate the proofs of these components since they may be of more general interest beyond our current discussion.

For Abelian discrete on-site symmetries, we prove an even stronger statement.
\begin{proposition}[Compilation with symmetric local gates] \label{prop: compile sym}
    For any Abelian discrete on-site symmetry $G$, each random unitary in the ensemble above can be compiled in circuit depth $\text{\emph{poly}}(n)$ using symmetric geometrically 5-local gates.
\end{proposition}
\noindent Our proof involves a detailed analysis of the unitary circuit in the first step of our symmetric PRU construction, and is provided in Appendix~\ref{sec: compiling constant local}.
The combination of Theorem~\ref{thm: sym PRU} and Proposition~\ref{prop: compile sym} yields Theorem~\ref{thm:poly} of the main text.

\subsubsection{Proof of Theorem~\ref{thm: sym PRU}: Symmetric PRUs in polynomial depth} \label{sec: sym PRU proof}

    Our proof consists of two steps.
    First, we show that the irreducible decomposition in the first and third steps of our construction can be implemented efficiently.
    This follows from a relatively straightforward approach, in which we use Clebsch-Gordon coefficients to gradually build up the irreducible decomposition one qudit at a time.
    Second, we show that controlled PRUs exist and can be implemented in $\poly n$ circuit depth on arbitrary local tensor product spaces.
    This more general setting is necessary because the Hilbert spaces $\mathcal{M}_\lambda$ may not have uniform qudit dimensions.

    Let us begin with the first proof step.
    We work in a 1D circuit geometry for simplicity; our results can be extended to any circuit geometry using Lemma~9 and~10 of Ref.~\cite{schuster2024random}.
    Our symmetric pseudorandom unitary has the form,
    \begin{equation}
        U = W^\dagger \cdot V^\dagger \cdot \prod_\lambda U^c_\lambda \cdot V \cdot W,
    \end{equation}
    Here, we let $W$ denote a unitary circuit that enacts the isomorphism $\mathcal{H} = \bigotimes_{i=1}^n ( \mathcal{F}_{\text{reg},i} \otimes \mathcal{A}_i ) \cong \mathbb{C}^{D/|G|} \otimes \mathcal{F}_{\text{reg}}$ described in Appendix~\ref{sec: background}.
    Then, we let $V$ denote a circuit that implements the isomorphism $\mathcal{F}_{\text{reg}} \cong \bigoplus_\lambda ( \mathbb{C}^{d_\lambda} \otimes \mathcal{F}_\lambda )$ on the output of $W$.
    The product $V W$ corresponds to the first step of our symmetric PRU construction.

    To implement $W$ efficiently, we recall the behavior of the regular representation of a group under the tensor product.
    Each qudit Hilbert space has the form $\mathcal{H}_i = \mathcal{F}_{\text{reg}} \otimes \mathcal{A}$ where $|\mathcal{A}| = d'$ and the symmetry group acts in the regular representation on the first register.
    In Appendix~\ref{sec: background}, we reviewed that the tensor product of two regular representations of a group is also regular.
    This allows one to define an isomorphism, 
    \begin{equation}
        \mathcal{F}_{\text{reg}} \otimes \mathcal{F}_{\text{reg}} \cong \mathbb{C}^{|G|} \otimes \mathcal{F}_{\text{reg}},
    \end{equation}
    where the symmetry operator acts only on the second register on the right side.
    We let $W'$ denote the unitary map that implements this isomorphism.
    Since $W'$ acts on only a constant-size Hilbert space, it can be implemented in constant depth.
    To construct the full irreducible decomposition $W$, we apply $W'$ one-by-one in a ladder-like fashion down the 1D chain.
    That is, we take $W = W'_{n-1,n} \ldots W'_{3,4} W'_{2,3} W'_{1,2}$, where $W'_{i,i+1}$ denotes the application of $W'$ to the $\mathcal{F}_{\text{reg}}$ registers of qudits $i$ and $i+1$.
    After all $n-1$ applications, the unitary $W$ implements the isomorphism
    \begin{equation}
        \big( \mathcal{F}_\text{reg} \otimes \mathcal{A} \big)^{\otimes n} \cong
        \big(\mathbb{C}^{|G|} \otimes \mathcal{A} \big)^{\otimes n-1} \otimes \mathcal{F}_\text{reg} \otimes \mathcal{A},
    \end{equation}
    where the symmetry operator, $W R_g W^\dagger$ where $R_g = \otimes_i R_g^i$, acts only on the $\mathcal{F}_\text{reg}$ register of the right side.
    Since each $W'_{i,i+1}$ can be implemented in constant depth, the unitary $W$ can be implemented in depth $\mathcal{O}(n)$.

    To implement the unitary $V$, we will introduce several ancilla registers.
    %
    Let $\Lambda = | \{ \lambda \} |$ denote the total number of irreps.
    We introduce two ancilla registers of dimension $d_\lambda$ for each $\lambda$, as well as a final ancilla register of dimension $\Lambda$.
    %
    %
    We let $V$ act on the final $\mathcal{F}_{\text{reg}}$ register and the ancilla registers as, 
    \begin{equation}
    \begin{split}
        V & \left( \ket{\lambda; x_\lambda, k_\lambda}_{\mathcal{F}_{\text{reg}}} 
        \otimes \ket{0}_\Lambda 
        \otimes \big( \bigotimes_{\lambda'} \ket{0}_{d_{\lambda'}} 
        \otimes \ket{0}_{\mathcal{F}_{\lambda'}} \big)  \right) \\
        & \quad \quad \quad \quad=
        \ket{0}_{\mathcal{F}_{\text{reg}}} 
        \otimes \ket{\lambda}_\Lambda  
        \otimes \big( \ket{x}_{d_{\lambda}} \otimes \ket{k}_{\mathcal{F}_{\lambda}} \big) 
        \otimes \big( \bigotimes_{\lambda' \neq \lambda} \ket{0}_{d_{\lambda'}} \otimes \ket{0}_{\mathcal{F}_{\lambda'}} \big).
    \end{split}
    \end{equation}
    Here, $x_\lambda$ labels which multiplicity the state belongs to and $k_\lambda$ labels the state in the irrep.
    The unitary $V$ can be implemented in constant depth since it acts on a constant-size Hilbert space.
    In total, the product $VW$ implements the isomorphism,
    \begin{equation}
        \big( \mathcal{F}_\text{reg} \otimes \mathcal{A} \big)^{\otimes n} 
        \cong
        \bigoplus_\lambda 
        \left[ \big(\mathbb{C}^{|G|} \otimes \mathcal{A} \big)^{\otimes n-1}  \otimes \mathcal{A} \otimes \mathbb{C}^{d_\lambda} \right] \otimes \mathcal{F}_\lambda,
    \end{equation}
    where we identify the Hilbert space in square brackets as $\mathcal{M}_\lambda$.
    Note that $\mathcal{M}_\lambda$ has a local tensor product structure, but with a non-uniform set of local Hilbert space dimensions $|G|, d', d_\lambda$.

    It remains only to show that the second step of our symmetric pseudorandom unitary construction can be implemented efficiently.
    In Theorem~\ref{thm: cRLFC} of Appendix~\ref{sec: controlled PRUs}, we show that controlled PRUs exist on any Hilbert space that factorizes as $\mathcal{H} = \mathcal{H}_L \otimes \mathcal{H}_R$.
    Here, we have in mind $\mathcal{H} \rightarrow \mathcal{M}_\lambda$ for each $\lambda$.
    The security of our controlled PRU is determined by the dimensions of $\mathcal{H}_L$ and $\mathcal{H}_R$; in particular, when both spaces have dimension exponential in $n$, we achieve security against any sub-exponential time quantum adversary.
    Our construction assumes only two ingredients: an efficient implementation of a quantum-secure pseudorandom function on $\mathcal{H}_L$ and $\mathcal{H}_R$, and an efficient implementation of an approximate unitary 2-design on $\mathcal{H}$.
    In Appendix~\ref{sec: PRFs LTP}, we show that existing results allow one to generate quantum-secure pseudorandom functions in $\poly n$ depth on any local tensor product space.
    In Appendix~\ref{sec: designs LTP}, we introduce a new, nested application of the two-layer circuit construction from Ref.~\cite{schuster2024random} which allows one to generate unitary $k$-designs in $\poly(\log n)$ depth on any local tensor product space.
    Since we have shown that $\mathcal{M}_\lambda$ is a local tensor product space, these results prove that controlled PRUs with sub-exponential security can be implemented in $\poly n$ depth on each $\mathcal{M}_\lambda$.
    This completes our proof. \qed

\subsubsection{Pseudorandom functions on any local tensor product space} \label{sec: PRFs LTP}

In this brief section, we show that existing results imply that the first three components, $F$, $O_L$, and $O_R$, of our controlled PRU construction can be realized in $\poly n$ depth in any local tensor product space.
This extends known results for systems composed solely of qubits~\cite{ma2024construct,schuster2024random}.
Here, a local tensor product space is given by $\mathcal{H} = \bigotimes_{i=1}^n \mathbbm{C}^{q_i}$, for arbitrary constant-sized local qudit dimensions $q_i$.

\begin{fact} \label{fact: F}
    One can generate the operator $F$ in circuit depth $\poly n$ on any local tensor product space and any circuit geometry.
\end{fact}

\begin{proof}
Ref.~\cite{zhandry2021PRF} constructs a quantum-secure pseudorandom function $f: [D] \rightarrow \{0,1\}$ for $D = 2^\ell$.
The function can be compiled in circuit depth $\poly n$ on any circuit geometry~\cite{schuster2024random}.
Setting $\ell = \lceil D \rceil$ and restricting the action of the function to the first $D$ bits yields $F$.
\end{proof}

\begin{fact} \label{fact: OL OR}
    One can generate the operators $O_L, O_R$ in circuit depth $\poly n$ on any local tensor product space and any circuit geometry.
\end{fact}

\begin{proof}
    Ref.~\cite{zhandry2021PRF} constructs a quantum-secure pseudorandom function $f_L$ for $D_L = 2^\ell$ and $D_R = p$, where $p$ is any $\poly n$-bit integer.
    The function can be compiled in circuit depth $\poly n$ on any circuit geometry~\cite{schuster2024random}.
    Setting $\ell = \lceil D_L \rceil$ and restricting the action of the function to the first $D_L$ bits yields $O_L$, and similar for $O_R$.
\end{proof}

\subsubsection{Unitary designs on any local tensor product space} \label{sec: designs LTP}

In this section, we show that unitary $k$-designs can be generated in $\poly(\log n)$ depth on any local tensor product space.
To do so, we introduce a \emph{nested} variant of the two-layer construction from Ref.~\cite{schuster2024random}, in which the small random unitaries of the two-layer circuit are themselves decomposed into smaller two-layer circuits.
This enables a stand-alone proof that unitary designs can be realized in extremely low depth, which does not rely on any existing results on the convergence of local random circuits to designs.
This is convenient for our purposes, since it avoids the need to perform a new spectral gap analysis for general local tensor product spaces.

Our main result is the following.
\begin{proposition}[Unitary designs on any local tensor product space]
    One can generate $\varepsilon$-approximate unitary $k$-designs in circuit depth $\mathcal{O}( \log(n) \cdot \log(1/\varepsilon) \cdot \log(nk/\varepsilon)^4 \cdot k^8 )$ on any local tensor product space and any circuit geometry. 
\end{proposition}

\begin{proof}
    We arrange the $n$ qudits into a 1D line.
    We then divide the qubits into $n/\xi$ contiguous patches, so that each patch has Hilbert space dimension at least $2^\xi$ for a chosen $\xi$.
    We then sub-divide each patch into $\xi / \zeta$ contiguous smaller patches, so that each smaller patch has Hilbert space dimension at least $2^\zeta$ for some $\zeta < \xi < n$.
    We consider a two-layer circuit on the large patches, repeated $\beta$ times.
    Each small random unitary of the two-layer circuit is generated, in turn, by a repeated two-layer circuit on the smaller patches, repeated $\alpha$ times.
    Each small random unitary in the small two-layer circuits is taken to be Haar-random.
    The construction has circuit depth $\mathcal{O}(\alpha \cdot \beta \cdot 2^{4\zeta})$, since a Haar-random unitary on a Hilbert space of dimension $\mathcal{O}(2^{2\zeta})$ can be compiled in circuit depth $\mathcal{O}(2^{4\zeta})$.

    To determine the required values of $\alpha$, $\beta$, $\xi$, $\zeta$, let us begin by analyzing the smaller repeated two-layer circuits.
    From Theorem~1 of Ref.~\cite{schuster2024random}, each small two-layer circuit within each patch forms a unitary $k$-design with relative error $\varepsilon'_0 \leq \xi k^2/2^\zeta$.
    Let us set $\zeta \geq \log_2(3 \xi k^2)$ so that $\varepsilon'_0 \leq 1/3$.
    From Fact~\ref{fact: rel error square} below, this implies that the repeated small two-layer circuit within each patch has relative error less than $2^{\alpha-1} \varepsilon_0'^\alpha \leq (2/3)^\alpha$.
    We can then set $\alpha \geq \log_{3/2}(1/\varepsilon')$ to obtain a $\varepsilon'$-approximate unitary $k$-design on each patch.
    %

    We can now apply Theorem~1 of Ref.~\cite{schuster2024random} again, to analyze the larger two-layer circuits.
    To do so, we set $\varepsilon' \leq \varepsilon_0/n$ and $\xi \geq \log_2(n k^2/\varepsilon)$.
    This guarantees that the nested two-layer circuit is an $\varepsilon_0$-approximate unitary $k$-design on all $n$ qubits.
    Similar to the previous paragraph, let us set $\varepsilon_0 \leq 1/3$.
    From Fact~\ref{fact: rel error square}, the repeated large two-layer circuit has relative error less than $2^{\beta-1} \varepsilon_0^\beta \leq (2/3)^\beta$.
    We can then set $\beta \geq \log_{3/2}(1/\varepsilon)$ to obtain a $\varepsilon$-approximate unitary $k$-design on all $n$ qubits.
    In total, our design has circuit depth $\mathcal{O}(\alpha \cdot \beta \cdot 2^{4\zeta}) = \mathcal{O}( \log(1/\varepsilon') \cdot \log(1/\varepsilon) \cdot (3\xi k^2)^4 ) = \mathcal{O}( \log(n) \cdot \log(1/\varepsilon) \cdot \log(nk/\varepsilon)^4 \cdot k^8 )$.
    This completes the proof.
\end{proof}

Our proof above uses the following fact.
The fact allows one to exponentially suppress the relative error of any unitary design by repeatedly applying the design several times in a row.

\begin{fact} \label{fact: rel error square}
    If $\Phi_{\mathcal{E}}$ has relative error $\varepsilon$, then $\Phi_{\mathcal{E}} \circ \Phi_{\mathcal{E}}$ has relative error less than $2 \varepsilon^2$.
\end{fact}

\begin{proof}
    We write $\Phi_{\mathcal{E}} = \Phi_H + \delta \Phi$, where $\Phi_H \circ \delta \Phi = \delta \Phi \circ \Phi_H = 0$ by definition.
    Squaring gives, $\Phi_{\mathcal{E}} \circ \Phi_{\mathcal{E}} = \Phi_H + \delta \Phi \circ \delta \Phi$.
    By assumption, we have $-\varepsilon \Phi_H(\rho) \preceq \delta \Phi(\rho) \preceq \varepsilon \Phi_H(\rho)$ for any $\rho$.
    We can break $\delta \Phi(\rho) =  \rho_+ -  \rho_-$ into a positive and negative part, $0 \preceq  \rho_{\pm} \preceq \varepsilon \Phi_H(\rho)$.
    We then have $\delta \Phi( \delta \Phi (\rho) ) = \delta \Phi(\rho_+ ) -  \delta \Phi( \rho_-)$.
    Moreover, since $\Phi_{\mathcal{E}}$ has relative error $\varepsilon$, we can bound $\delta \Phi( \rho_\pm ) \preceq \varepsilon \Phi_H(\rho_\pm)$.
    Combining these two inequalities, we have
    \begin{equation}
        - 2 \varepsilon^2 \Phi_H(\rho) \preceq - \varepsilon \Phi_H(\rho_+) - \varepsilon \Phi_H(\rho_-) \preceq \delta \Phi ( \delta \Phi (\rho) ) \preceq \varepsilon \Phi_H(\rho_+) + \varepsilon \Phi_H(\rho_-) \preceq 2 \varepsilon^2 \Phi_H(\rho),
    \end{equation}
    which completes the proof.
\end{proof}

\subsubsection{Proof of Proposition~\ref{prop: compile sym}: Compiling with symmetric geometrically 5-local gates} \label{sec: compiling constant local}

Let us begin by recalling a few facts about discrete on-site Abelian symmetries.
Any finite Abelian group is isomorphic to a tensor sum of cyclic groups,
\begin{equation}
    G \cong \mathbb{Z}_{p_1} \otimes \mathbb{Z}_{p_2} \otimes \cdots \otimes \mathbb{Z}_{p_l},
\end{equation}
for positive numbers $p_1, \ldots, p_l$.
Each irreducible representation $x$ of an Abelian group has dimension $d_x = 1$.
This implies that each irrep $x$ appears once in the regular representation, so we can write $\mathcal{F}_{\text{reg}} = \{ \ket{x} \}$.
It also implies that for any fixed irrep, the action of any group element amounts to an overall phase.
The label $x$ of each irrep is referred to as its \emph{charge},
\begin{equation}
    x \in \mathbb{Z}_{p_1} \otimes \mathbb{Z}_{p_2} \otimes \cdots \otimes \mathbb{Z}_{p_l},
\end{equation}
and determines the overall phase accrued when a symmetry group element is applied, $R_g \ket{x} = \prod_{j=1}^l e^{i (2\pi/p_j) x_j g_j} \ket{x}$. 
From this definition, it is clear that the charges of irreps add under the tensor product, i.e.~the state $\ket{x_1} \otimes \ket{x_2}$ has charge $x_1 \oplus x_2$, where $\oplus$ denotes addition in the group $G$.

Let us now turn to compilation of our symmetric random unitary ensemble.
We note two facts.
First, the unitary $V$ is equal to the identity for an Abelian group, since the regular representation is already a tensor sum with each irrep appearing once.
Second, we can take the unitary $W$ to be a generalized CNOT ladder.
Namely, we let $W'_{i,i+1} (\ket{ y_i } \otimes \ket{x_{i+1}}) = \ket{y_i} \otimes \ket{x_{i+1} \oplus y_i}$ for each $i$.
After $n$ applications, one finds $W \ket{x_1 \ldots x_n} = \ket{ y_1 \ldots y_n }$,
where $y_i = x_1 \oplus x_2 \oplus \cdots \oplus x_i$ is equal to the charge of all qudits to the left of $i$ (including $i$ itself).
The $n$-th qudit contains the charge $y_n$ of the entire system, as required by the definition of $W$.

To implement our symmetric unitary ensemble, $U = W^\dagger ( \prod_{y_n} U^c_{y_n} ) W$, we first decompose each controlled PRU, $U^c_{y_n}$, as a product of controlled unitary gates of the form,
\begin{equation} \label{eq: control V}
    \dyad{y_n} \otimes V_{i,i+1} + \big(\mathbbm{1}-\dyad{y_n} \big) \otimes \mathbbm{1}_{i,i+1}.
\end{equation}
Here, the projector $\dyad{y_n}$ acts on the $n$-th qudit (or any register that stores the total charge $y_n$), and $V_{i,i+1}$ acts on qudits $i,i+1$.
We can assume that each $V_{i,i+1}$ is a two-qudit geometrically-local gate since such gates are universal for quantum computation.

To proceed, we consider the conjugation of the gate $V_{i,i+1}$ by the final $n-i+1$ layers of the circuit $W$.
That is,
\begin{equation}
    \tilde{V}_{i,i+1} \equiv W^\dagger_{i-1,i} \ldots W^\dagger_{n-1,n} V_{i,i+1} W_{n-1,n} \ldots W_{i-1,i}.
\end{equation}
We let $y_j$ denote the value of qudit $j$ after the $W$ gates are applied (i.e.~where $V_{i,i+1}$ acts), and $x_j$ the value before the $W$ gates are applied.
We note three properties.
First, the light-cone of the $W$ gates implies that $\tilde{V}_{i,i+1}$ can act at most on qudits $\{i-1,i,i+1,i+2\}$.
Second, the qudit $\{ i-1 \}$ must function only as a control qudit for an operation on the remaining $\{i,i+1,i+2\}$.
That is, we can write $\tilde{V}_{i,i+1} = \sum_{y_i} \dyad{y_{i-1}} \otimes \tilde{V}^{y_i}_{i,i+1,i+2}$.
This follows because the value of $y_{i-1}$ is not changed by the gate $W_{i-1,i}$ nor the unitary $V_{i,i+1}$ (since $V_{i,i+1}$ acts only on $y_i$ and $y_{i+1}$ and not $y_{i-1}$).
Third, the controlled operations $\tilde{V}^{y_i}_{i,i+1,i+2}$ are symmetric on $\{i,i+1,i+2\}$.
This follows because $V_{i,i+1}$ cannot change the value of $y_{i+2} = y_{i-1} \oplus x_i \oplus x_{i+1} \oplus x_{i+2}$.
Since $V_{i,i+1}$ also cannot change the value of $y_{i-1}$, it must not change the charge, $x_i \oplus x_{i+1} \oplus x_{i+2}$, on $\{ i,i+1,i+2 \}$.
This implies that $\tilde{V}^{y_i}_{i,i+1,i+2}$ is symmetric.

With these three properties established, our compilation of $U$ is as follows.
We copy the charge $y_n$ of all qudits to a first ancilla register.
We then iterate over all charges $y_n$ and all gates in $U^c_{y_n}$ [Eq.~(\ref{eq: control V})].
For each gate, we first copy the charge $y_{i-1}$ to a second ancilla register.
We then perform the gate $\tilde{V}_{i,i+1}$ as a controlled gate from the first and second ancilla register to qudits $\{ i,i+1,i+2 \}$.
From the discussion above, this is a symmetric gate acting on three geometrically-local system qudits and two ancilla qudits.
We then uncompute the charge $y_{i-1}$ and repeat for all charges $y_n$ and all gates in $U^c_{y_n}$.
After all gates are completed, we uncompute the charge $y_n$.
This completes our compilation of $U$.
Our compilation has a factor of $n$ greater depth compared to the original compilation, due to the need to compute $y_{i-1}$ for every gate.
Hence, if the original circuit depth is $\poly n$, our symmetric compilation also has circuit depth $n \cdot \poly n = \poly n$.

In the above discussion, we neglected to notate the ancilla register $\mathcal{A}$ to reduce notation.
One can easily see that the analysis is unchanged by the addition of new registers that are not acted on by the symmetry group, since these registers are also not acted on by any of the $W_{i,i+1}$. \qed

\subsection{Proof of Theorem~\ref{thm:polylog}: Symmetric PRUs in poly-logarithmic depth} \label{sec: sym PRUs low depth}

The combination of our gluing theorem (Theorem~\ref{thm:main-design}) and our symmetric pseudorandom unitary construction (Theorem~\ref{thm: sym PRU}) immediately implies that symmetric pseudorandom unitaries can be generated in extremely low circuit depths.
This is formalized in Theorem~\ref{thm: low depth PRU} and Corollary~\ref{cor: low depth PRU} below, which together prove Theorem~\ref{thm:polylog} of the main text.
Our proofs of both statements are extremely short building upon our earlier results.

\begin{theorem}[Gluing small symmetric pseudorandom unitaries] \label{thm: low depth PRU}
For any discrete on-site symmetry group $G$.
Let $n$ be the number of qudits in the whole system and $\xi = \omega(\log n)$ be the number of qudits in each local patch.
Suppose each small random unitary in the two-layer brickwork ensemble $\mathcal{E}$ is a $2\xi$-qubit symmetric PRU secure against $\text{\emph{poly}}(n)$-time adversaries.
Then $\mathcal{E}$ is an $n$-qubit symmetric PRU secure against $\text{\emph{poly}}(n)$-time adversaries.
\end{theorem}
\begin{proof}
    Our proof is identical to the proof of Theorem~2 in Ref.~\cite{schuster2024random}.
    The only change is to replace Theorem~1 in Ref.~\cite{schuster2024random}, which applies to random unitaries without symmetries, with Theorem~\ref{thm:main-design} in our work, which applies to random unitaries with symmetries.
\end{proof}

\begin{corollary}[Low-depth pseudorandom unitaries] \label{cor: low depth PRU}
For any discrete on-site symmetry group $G$.
Under the conjecture that no subexponential-time quantum algorithm can solve LWE, random quantum circuits over $n$ qubits can form symmetric PRUs secure against any polynomial-time quantum adversary in circuit depth
$d = \text{\emph{poly}} \log n$, on any circuit geometry.
\end{corollary}
\begin{proof}
    The corollary follows immediately from Theorem~\ref{thm: sym PRU} and Theorem~\ref{thm: low depth PRU}.
    Here, we take the number of qudits in Theorem~\ref{thm: sym PRU} to be $\xi = \poly(\log n)$.
\end{proof}

\section{Controlled pseudorandom unitaries} \label{sec: controlled PRUs}

In this section, we take a brief detour from our analysis of symmetric systems to prove the existence of controlled PRUs.
Controlled PRUs are unitary ensembles that are indistiguishable from controlled Haar-random unitaries by any polynomial-time quantum adversary.
They are also central to our construction of symmetric PRUs in the previous section.
Nonetheless, to this point, the existence of a controlled PRU ensemble is not known in the literature.
In what follows, we analyze controlled variants of the PFC and LRFC PRU ensembles, and prove that both controlled ensembles are controlled PRUs.
Our analysis requires multiple technical innovations compared to existing PRU proofs to handle the control register.

\subsection{Summary of results}

Let us first define several operators and notations.
We consider a Hilbert space $\mathcal{H} = \mathcal{H}_{\mathsf{C}} \otimes \mathcal{H}_{\mathsf{S}}$ where $\mathcal{H}_{\mathsf{C}}$ is a single-qubit control register and $\mathcal{H}_{\mathsf{S}}$ is a system register of dimension $D$.
The Hilbert space is spanned by the computational basis states, $\{ \ket{t, x} \, | \, t \in \{0,1\}, x \in [D] \}$.
We let $F^c$ denote a quantum-secure controlled pseudorandom function (PRF) acting on this Hilbert space in the phase representation,
\begin{equation}
    F^c \ket{0, x} = \ket{0, x}, \quad \quad \quad  F^c \ket{1, x} = (-1)^{f(x)} \ket{1, x},
\end{equation}
where $f: [D] \rightarrow \{0,1\}$.
We let $P^c$ denote a quantum-secure controlled pseudorandom permutation (PRP), 
\begin{equation}
    P^c \ket{0, x} = \ket{0, x}, \quad \quad \quad  P^c \ket{1, x} = \ket{1, P(x)},
\end{equation}
where $P: [D] \rightarrow [D]$ is one-to-one.
For any Hilbert space that factorizes, $\mathcal{H} = \mathcal{H}_L \otimes \mathcal{H}_R$, we let $S^c_L$ denote quantum-secure controlled PRFs, acting in the  representation,
\begin{equation}
    S^c_L \ket{0, x_L, x_R} = \ket{0, x_L, x_R}, \quad \quad \quad S^c_L \ket{1, x_L, x_R} = \ket{1, x_L \oplus f_L(x_R), x_R},
\end{equation}
where $f_L: [D_R] \rightarrow [D_L]$.
We let $S^c_R$ denote the analogous PRF,
\begin{equation}
    S^c_R \ket{0, x_L, x_R} = \ket{0, x_L, x_R}, \quad \quad \quad S^c_R \ket{1, x_L, x_R} = \ket{1, x_L, x_R \oplus f_R(x_L)},
\end{equation}
where $f_R: [D_L] \rightarrow [D_R]$.
Here, $D_L = | \mathcal{H}_L |$, $D_R = | \mathcal{H}_R |$, and $D_L D_R = D = | \mathcal{H} |$.
We note that quantum-secure controlled PRFs and PRPs are easily constructed from quantum-secure PRFs and PRPs without controls by using controlled SWAP gates\footnote{In more detail, to obtain a controlled PRF $F^c$ from a standard PRF $F$, we consider the operation $\text{cSWAP} (\mathds{1} \otimes \mathds{1} \otimes F)  \text{cSWAP}$, where the three subsystems are the control qubit, the system register, and an ancilla register of the same size as the system, respectively. Here, $\text{cSWAP}$ performs a controlled swap between the system and ancilla registers. If one prepares the ancilla register in the zero state, this operator is equal to the desired controlled PRF up to a controlled phase $(-1)^{f(0)}$. This phase can be determined and undone with one additional query to $f$.
\newline 
\newline
\noindent To obtain a controlled PRP $P^c$ from a standard PRP, we consider the operation $\text{cSWAP}  (\mathds{1} \otimes \mathds{1} \otimes P)  \text{cSWAP}$, and assume that the ancilla register again begins in the zero state.
If the control register is in the $1$ state, this is equal to the desired controlled PRP.
If the control register is in the $0$ state, it is equal to the desired PRP on the control and system registers, but leaves the ancilla register in the state $\ket{P(0)}$.
The ancilla register can then be traced out, leaving a controlled PRP on the control and system.
Alternatively, to maintain unitarity on the entire system including the ancilla qubits, one can apply another query to $P$ on an additional ancilla register prepared in the all zero state, use the result to controllably return the original ancilla register to the zero state, and then apply $P^{-1}$ to the additional ancilla register to return that register to the zero state as well.
\newline 
\newline
\noindent The controlled PRFs $O_L^c$ and $O_R^c$ can be obtained from their standard versions in an identical fashion to controlled PRPs.}.
Finally, we let $C^c$ denote a controlled application of any unitary 2-design, 
\begin{equation}
    C^c = \dyad{0} \otimes \mathbbm{1} + \dyad{1} \otimes C,
\end{equation}
where $C \sim \mathfrak{D}$ and $\mathfrak{D}$ is a 2-design.

Definitions in hand, we consider the following two controlled random unitary ensembles.
First, we consider a controlled variant of the Permutation-Function-Clifford (PFC) ensemble~\cite{metger2024simple,ma2024construct},
\begin{equation} \label{eq: Uc PFC}
    U^c = P^c \cdot F^c \cdot C^c,
\end{equation}
where $P^c, F^c, C^c$ are drawn randomly as described above.
The PFC ensemble without controls is the simplest known ensemble that forms a PRU. 
Second, we consider a controlled variant of the Luby-Rackoff-Function-Clifford (LRFC) ensemble~\cite{schuster2025strong},
\begin{equation} \label{eq: Uc RFLC}
    U^c = S^c_R \cdot F^c \cdot S^c_L \cdot C^c.
\end{equation}
The LRFC ensemble was introduced as an alternative to the PFC ensemble, which relies solely on pseudorandom functions and not permutations.
This can be convenient, as existing results on quantum-secure PRPs are more limited than those on PRFs.
In particular, for our application to symmetric pseudorandom unitaries in this work, we require the existence of controlled PRUs on systems composed of qudits of arbitrary dimension.
Existing analyses of quantum-secure PRPs are limited to qu\emph{bit} systems; hence, our main results will utilize only on the LRFC ensemble.
%

Our main result is that both of these controlled random unitary ensembles are controlled PRUs.
\begin{theorem} \label{thm: cPFC}
    The controlled PFC ensemble is a controlled PRU ensemble, with security against any sub-exponential time quantum adversary.
\end{theorem}
\begin{theorem} \label{thm: cRLFC}
    The controlled LRFC ensemble is also a controlled PRU ensemble, with security against any sub-exponential time quantum adversary.
\end{theorem}
\noindent We prove the theorems in the following two subsections.
Both of our proofs follow the concise strategy introduced in Ref.~\cite{cui2025unitary} for proving the adaptive security of PRUs.
We augment this approach with several technical innovations to extend the analysis to controlled random unitaries.
This extension must be performed carefully, since the control register can take exponentially many values, $| \{ 0,1 \}^k | = 2^k$, in an experiment that queries the unitary $k = \poly n$ times.  
Thus, many naive proof approaches will diverge exponentially in $k$.
We solve this by keeping a detailed accounting of the controlled register and leveraging several useful operator inequalities.
Our proof of Theorem~\ref{thm: cPFC} is  simpler and is presented first.
Our proof of Theorem~\ref{thm: cRLFC} then follows by similar steps.

\subsection{Proof of Theorem~\ref{thm: cPFC}: Controlled-PFC ensemble}

We can express the output state of any experiment that queries $U_c$ up to $k$ times as,
\begin{equation} \label{eq: 1}
	\ket{\psi_U}_{\mg{AB}} = (2D)^k \cdot (\mathbbm{1}_{\mg{AB}} \otimes \bra{ \Psi_{\text{Bell}} }_{\mg{XY}} ) (\mathbbm{1}_{\mg{ABY}} \otimes \boldsymbol{U}_{\mg{X}})\ket{\Psi}_{\mg{ABXY}},
\end{equation}
where $\boldsymbol{U} \equiv (U^c)^{\otimes k}$.
We also denote $\bs{P} = (P^c)^{\otimes k}$, $\bs{F} = (F^c)^{\otimes k}$, $\bs{C} = (C^c)^{\otimes k}$.
We refer to Ref.~\cite{schuster2024random} for a visual depiction.
Here, $\mathsf A$ is the combined control and system register, $\mathsf B$ is an arbitrary ancilla register, and $\mathsf X$ and $\mathsf Y$ are ancilla registers of the same size as $k$ copies of $\mathsf A$.
$\ket{\Psi_{\text{Bell}}}_{\mg{XY}}$ is a Bell state between $\mathsf X$ and $\mathsf Y$.
The factor $(2D)^k = | \mathsf X | = | \mathsf Y |$ normalizes the state after the Bell projection.

To write the expected output state in a more compact form, we let
\begin{equation}
    B'_{\mg{XY}} \equiv (2D)^{2k} \cdot \E_{P, f} \big[ ( \bs{F}^\dagger \bs{P}^\dagger \otimes \mathbbm{1} )_{\mg{XY}} \dyad{\Psi_{\text{Bell}}}_{\mg{XY}} ( \bs{P} \bs{F} \otimes \mathbbm{1})_{\mg{XY}} \big]
\end{equation}
denote the Bell projector twirled over $PF$, multiplied by $(2D)^{2k}$.
This yields,
\begin{equation}
	\rho_{\mg{AB}} \equiv \E_{P, f} \big[ \dyad{\psi_U}_{\mg{AB}} \big] = \E_{C} \big[ \tr_{\mg{XY}} \big( B'_{\mg{XY}} \bs{C}_\mg{X} \dyad*{\Psi}_{\mg{ABXY}} \bs{C}^\dagger_{\mg{X}} \big) \big].
\end{equation}
We let $\rho^H_{\mg{AB}} \equiv \E_{U \sim H} \dyad{\psi_U}_{\mg{AB}}$ denote the output of the same experiment when $U_c$ is controlled Haar-random.
We will show that $\lVert \rho - \rho^H \rVert_1 \leq 10 \delta$, where $\delta = k^2/2^n$.
This establishes  security against any sub-exponential time quantum adversary.


We begin by writing down several definitions and facts.
We define the controlled distinct subspace projectors, 
\begin{equation}
    \cPD_{\mg{X}} = \sum_{t \in \{0,1\}^k} \dyad{t}_{\mg{X}_{\mg{C}}} \otimes \left[ \PD_t \otimes \mathbbm{1}_{\bar t} \right]_{\mg{X}_{\mg{S}}},
    \quad \quad \quad \cPD_{\mg{Y}} = \sum_{t \in \{0,1\}^k} \dyad{t}_{\mg{Y}_{\mg{C}}} \otimes \left[ \PD_t \otimes \mathbbm{1}_{\bar t} \right]_{\mg{Y}_{\mg{S}}},
\end{equation}
as well as their product,
\begin{equation}
    \cPD_{\mg{XY}} = \cPD_{\mg{X}} \otimes \cPD_{\mg{Y}}.
\end{equation}
Here, $\mathsf{X}_{\mathsf{L,C}}$ denotes the $k$ control registers in $\mathsf{X}_\mathsf{L}$, and $\mathsf{X}_{\mathsf{L,S}}$ the $k$ system register, and similar for $\mathsf{Y}_\mathsf{R}$.
Here, $\PD_t = \sum_{x \in \text{dist}} \dyad{x}$ is the projector onto the distinct subspace on the system registers of copies $j$ with $t_j = 1$, where sum runs over all distinct bitstrings in $[D_L]^{\otimes |t|}$ or $[D_R]^{\otimes |t|}$, and $\mathbbm{1}_{\bar t}$ is the identity matrix on the copies $j$ with $t_j=0$.

In analogy to $B'$ above, we also define the operator,
\begin{equation}
    B_{\mg{XY}} = D^{2(k-|t|)} \sum_{\substack{ t, t' \in \{0,1\}^k \\ |t| = |t'|}} \left( \dyad{t', t'}{t, t}_{\mg{X}_\mg{C} \mg{Y}_\mg{C}} \otimes \sum_{\pi \in S_{|t|}} \left[ \left( \pi_{t \rightarrow t'} \otimes \pi_{t \rightarrow t'} \right) \otimes \dyad*{ \Psi^{\bar t '}_{\text{Bell}} }{ \Psi^{\bar t}_{\text{Bell}} } \right]_{\mg{X}_\mg{S} \mg{Y}_\mg{S}} \right),
\end{equation}
where $\pi_{t \rightarrow t'}$ is a permutation mapping the system registers of the copies of $\mathsf{X}$ where $t_j=1$ to the those where $t'_j = 1$, and similarly for $\mathsf{Y}$.
With these definitions in hand, our proof uses five facts:

\vspace{3mm}
\noindent \text{(1)}
    $B$ and $B'$ are equal on the controlled distinct subspace up to small relative error, 
    $B_{\mg{XY}} \cPD_{\mg{XY}} \preceq B'_{\mg{XY}} \cPD_{\mg{XY}} \preceq (1+\delta) B_{\mg{XY}} \cPD_{\mg{XY}}$.
    To see this, we express the former as 
    \begin{equation}
    B'_{\mg{XY}} =  \sum_{\substack{ t, t' \in \{0,1\}^k \\ |t| = |t'|}} D^{2(k-|t|)} \cdot \dyad{t, t}{t', t'}_{\mg{X}_\mg{C} \mg{Y}_\mg{C}} \otimes \left[ \left( \tilde{\pi}_{t \rightarrow t'} \otimes \tilde{\pi}_{t \rightarrow t'} \right) \cdot
    \left( B'_t \otimes \dyad*{ \Psi^{\bar t}_{\text{Bell}} }{ \Psi^{\bar t}_{\text{Bell}} } \right) \right]_{\mg{X}_\mg{S} \mg{Y}_\mg{S}} ,
\end{equation}
where $\tilde{\pi}_{t\rightarrow t'}$ is any permutation that maps copies $j$ where $t_j = 1$ to those where $t'_j = 1$. We also abbreviate,
\begin{equation}
     (B'_{t})_{\mg{XY}_t} \equiv D^{2|t|} \E_{P, f} \big[ ( F^{\dagger,\otimes |t|} P^{\dagger,\otimes |t|} \otimes \mathbbm{1} )_{\mg{XY}_t} \dyad{\Psi^t_{\text{Bell}}}_{\mg{XY}_t} ( P^{\otimes |t|} F^{\otimes |t|}\otimes \mathbbm{1})_{\mg{XY}_t} \big].
\end{equation}
We can perform precisely the same decomposition for $B$, replacing $(B'_t)_{\mg{XY}_t}$ with $(B_t)_{\mg{XY}_t} \equiv \sum_{\pi \in S_{|t|}} \pi \otimes \pi$.
Inserting the controlled distinct subspace projectors, we find
\begin{equation} \nonumber
\begin{split}
    B'_t ( \PD_{\mg{X}_{t}} \otimes \PD_{\mg{Y}_{t}} )  
    &= D^{|t|} \E_{P, f}  \sum_{x \,\text{dist}} \sum_{\tilde x}
        (-1)^{f(x)+f(\tilde{x})} \dyad{\tilde x, P \tilde x}{x, P x}_{\mg{XY}_t} \\
    & = D^{|t|} \E_{P} \sum_\pi \sum_{x \,\text{dist}}
        \dyad{\pi x, \pi P  x}{x, P x}_{\mg{XY}_t} \\
    & = (D^{|t|}/\mathfrak{D}_{|t|}) \sum_\pi \sum_{x \,\text{dist}} \sum_{y \,\text{dist}}
        \dyad{ \pi x, \pi y}{x, y}_{\mg{XY}_t}\\
    & = (D^{|t|}/\mathfrak{D}_{|t|}) \sum_\pi (\pi \otimes \pi)_{\mg{XY}_t} \cPD_{\mg{XY}_t} = (D^{|t|}/\mathfrak{D}_{|t|}) (B_t)_{\mg{XY}_t} ( \PD_{\mg{X}_{t}} \otimes \PD_{\mg{Y}_{t}} ),
\end{split}
\end{equation}
where the expectation over $f$ enforces $\tilde x = \pi x$ for some $\pi \in S_{|t|}$, and the expectation over $P$ yields a random distinct bitstring $y$.
In the final line, we use $\sum_{x,y} \dyad{x,y} = \mathbbm{1}$ to eliminate the sums.
We see that both $B'_t$ and $B_t$ are block diagonal in $|t|$, and are equal in each block up to the constant factor $D^{|t|}/\mathfrak{D}_{|t|}$.
The inequality $1 \leq D^{|t|}/\mathfrak{D}_{|t|} \leq 1+\delta$ completes our claim.
    
\vspace{5mm}
\noindent \text{(2.)}
    The twirl of $\cPD_{\mg{X}}$ over a controlled 2-design is close to the identity, $\big\lVert \mathbbm{1} - \E_{C \sim \mathfrak{D}} \big[ \bs{C}_\mg{X}^\dagger \cPD_{\mg{X}} \bs{C}_\mg{X} \big] \big\rVert_\infty \leq \delta$.
    To show this, we bound $\mathbbm{1} - \cPD_{\mg{X}} \leq \sum_{t\in \{0,1\}^k} \sum_{i<j \in t} \dyad{t}_{\mg{X}_{\mg{C}}} \otimes \Pi_{\mg{X}_{\mg{S}}}^{\mathsf{eq},ij}$, where $$\Pi_{\mg{X}_\mg{S}}^{\mathsf{eq},ij} = \sum_{x} \dyad{x, x}_{\mg{X}_{\mg{S},i} \mg{X}_{\mg{S},j}}.$$
    The twirl over a controlled 2-design yields $\E_{C} [ \bs{C}_\mg{X}^\dagger \sum_t \sum_{i<j \in t}  \dyad{t} \otimes \Pi_{\mg{X}_\mg{S}}^{\mathsf{eq},ij} \bs{C}_\mg{X} ] = \sum_t \sum_{i<j \in t} \dyad{t} \otimes
    (\mathbbm{1} + \mathcal{S}_{ij})/(D+1)$, where $\mathcal{S}_{ij}$ is the swap operator.
    Hence, the spectral norm is less $\max_t \sum_{i<j\in t} 2/(D+1) \leq (k^2/2)(2/D) = \delta$.

\vspace{5mm}
\noindent \text{(3.)}
$B$ and $B'$ commute with $\cPD_{\mg{X}}$ and $\cPD_{\mg{Y}}$. 
$B$ also commutes with $\bs{C}$.

\vspace{5mm}
\noindent \text{(4.)} For any positive Hermitian matrices $P,Q,R$, where $P$ and $Q$ commute, we have $\tr(PQR)  \leq \lVert Q \rVert_\infty \tr(P R)$.
This follows from Holder's inequality, $$\tr(PQR) = \text{tr}(Q \sqrt{P} R \sqrt{P}) \leq \lVert Q \rVert_\infty \lVert \sqrt{P} R \sqrt{P} \rVert_1 =  \lVert Q \rVert_\infty \tr(P R).$$

\vspace{5mm}
\noindent \text{(5.)} We have $1-\delta \leq \tr(B^{}_{\mg{XY}} \cPD_{\mg{Y}} \dyad{\Psi} ) \leq 1$.
This is shown below the current proof.

\vspace{5mm}
\noindent We can now prove the claim.
We insert $\mathbbm{1} = \cPD_{\mg{XY}} + (\mathbbm{1}-\cPD_{\mg{XY}})$ to decompose $\rho$ into two terms,
\begin{align*}
	\rho = \rho_{\text{dist}} + \delta\rho \equiv \E_{C} \big[ \tr_{\mg{XY}}( B' \cdot \cPD_{\mg{XY}} \cdot \bs{C} \dyad*{\Psi} \bs{C}^\dagger ) \big] + \E_{C} \big[ \tr_{\mg{XY}}( B' \cdot (\mathbbm{1}-\cPD_{\mg{XY}}) \cdot \bs{C} \dyad*{\Psi} \bs{C}^\dagger ) \big].
\end{align*}
\noindent We bound the second term as follows,
\begin{align*}
    \lVert \delta \rho \rVert_1 = \tr( \delta \rho ) & = \E_{C} \big[ \tr( B' (\mathbbm{1}-\cPD_{\mg{XY}}) \bs{C} \dyad*{\Psi} \bs{C}^\dagger ) \big] && \text{(since $\delta \rho$ is positive)} \\
    & \leq 1 - \E_{C} \big[ \tr \big( B \cdot \cPD_{\mg{XY}} \cdot \bs{C}\dyad*{\Psi} \bs{C}^\dagger \big)  \big], && \text{(from $\tr(\rho) = 1$ and (1.))} \\
    & = 1 - \tr \big( B \cdot \E_{C} \big[ \bs{C}^\dagger \cPD_{\mg{X}} \bs{C} \big] \otimes \cPD_{\mg{Y}} \cdot \dyad*{\Psi} \big), && \text{(from (3.))} \\
    & \leq 1 - (1-\delta) \tr \big( B (\mathbbm{1} \otimes \cPD_{\mg{Y}}) \dyad*{\Psi} \big) && \text{(from (2.), (3.), (4.))} \\
    & = 1 - (1-\delta)^2 \leq 2 \delta. && \text{(from (5.))}
\end{align*}
Meanwhile, the first term is close to a fixed density matrix $\rho_a$, independent of the ensemble $\mathfrak{D}$,
\begin{align*}
    \rho_{\text{dist}} & = \tr_{\mg{XY}} \!  \big( B \cdot \E_{C} \big[ \bs{C}^\dagger \cPD_{\mg{X}} \bs{C} \big] \otimes \cPD_{\mg{Y}} \cdot \dyad*{\Psi} \big) && \text{(up to relative error $\delta$, from (1.))} \\
    & = \tr_{\mg{XY}} \!  \big( B \cdot \mathbbm{1}_{\mg{X}} \otimes \cPD_{\mg{Y}} \cdot \dyad*{\Psi} \big) + \Delta \equiv \rho_a + \Delta, && \text{(from (2.), (3.), (4.))}
\end{align*} 
where $\lVert \Delta \rVert_1 \leq \delta  \tr( B \cdot \mathbbm{1}_{\mg{X}} \otimes \cPD_{\mg{Y}} \cdot \dyad*{\Psi}) \leq \delta$ from (2.), (3.), (4.).
We have $\lVert \rho_{\text{dist}} - \rho_a \rVert_1 \leq 3 \delta$, since the 1-norm error is upper bounded by twice the relative error~\cite{schuster2024random}.

This completes our proof: We have $\lVert \rho - \rho_a \rVert_1 \leq 2\delta + 3\delta$ when $C$ is drawn from any 2-design $\mathfrak{D}$.
Since the Haar ensemble is a 2-design, we have $\lVert \rho - \rho^H \rVert_1 \leq \lVert \rho^H - \rho_a \rVert_1 + \lVert \rho - \rho_a \rVert_1 \leq 4\delta + 6\delta = 10k^2/D$. \qed

\vspace{3mm}
\noindent \emph{Proof of (5.)}---We proceed by induction.
%
Let $\rho_j = \tr_{\mg{AB} \mg{XY}_{>j}}\!( \dyad*{\Psi} )$ denote the reduced density matrix on on $\mathsf{XY}_{\leq j}$.
The key property we use is that $\rho_j$ is maximally mixed on $\mathsf{Y}_j$. 

\vspace{2mm}
\noindent \textbf{Base case ($j=1$): }Here, $B = \mathbbm{1} \otimes \mathbbm{1}$ and $\cPD_{\mg{Y}_{\leq 1}} = \mathbbm{1}_{\mg{Y}_{1}}$. Hence, $\text{tr}(B \cdot \cPD_{\mg{Y}_{\leq 1}} \cdot \dyad*{\Psi}) = 1$.

\vspace{2mm}
\noindent \textbf{Inductive step: }We can first expand,
\begin{equation} \nonumber
    B_{\mg{XY}} \cPD_{\mg{Y}} = \!\!\!\! \sum_{\substack{t,t' \in \{0,1\}^j \\ |t|=|t'|}}  \!\!
    D^{2(j-|t|)}
    \cdot \left( \dyad{t',t'}{t,t}_{\mg{X}_{\mg{C}}\mg{Y}_{\mg{C}}}
    \otimes \sum_{\pi \in S_{|t|}}
    \left[ \left( \pi_{t \rightarrow t'} \PD_{\mg{Y}_{\mg{S},t}}  \otimes \pi_{t \rightarrow t'} \right) \otimes \dyad*{ \Psi^{\bar t '}_{\text{Bell}} }{ \Psi^{\bar t}_{\text{Bell}} } \right]_{\mg{X}_\mg{S} \mg{Y}_\mg{S}}
    \! \right).
\end{equation}
We now perform the trace over $\mathsf{Y}_j$.
Since $\rho_j$ is maximally mixed on $\mathsf{Y}_j$, the trace sets $t_j = t'_j$.
If $t_j = 0$, we find 
\begin{equation}
    \tr_{\mg{XY}_j} \!\! \big( \big[ \dyad{0,0} \otimes \dyad*{\Psi_{\text{Bell}}^j} \big]_{\mg{XY}_j} \rho_j \big) = p_j(0) \cdot \rho^{(0)}_{j-1} / D^2,
\end{equation}
where $p_j(0) \rho^{(0)}_{j-1} \equiv \tr_{\mg{XY}_j}( \dyad{0}_{\mg{X}_{\mg{C},j}} \rho_j)$ is the probability of measuring a $0$ on register $\mathsf{X}_{\mathsf{C},j}$, multiplied by the resulting normalized quantum state.
Meanwhile, if $t_j = 1$, we find
\begin{equation}
    \tr_{\mg{XY}_j} \!\! \big( \big[ \dyad{1,1} \otimes \sum_{\pi \in S_{|t|}}
    \left( \pi_{t \rightarrow t'}   \otimes \pi_{t \rightarrow t'} \PD_{\mg{Y}_{\mg{S},t}} \right) \big]_{\mg{XY}_j} \rho_j. \big),
\end{equation}
Leveraging the identity, $\tr_{\mg{Y}_{\mg{S},j}}( \pi_{\mg{Y}_\mg{S}} \PD_{\mg{Y}_{\mg{S}}} ) = (D-j+1) \pi_{\mg{Y}_{\mg{S},< j}} \PD_{\mg{Y}_{\mg{S},< j}}$ if $\pi(j)=j$, and $\tr_{\mg{Y}_{\mg{S},j}}( \pi_{\mg{Y}_\mg{S}} \PD_{\mg{Y}_{\mg{S}}} ) = 0$ otherwise, we can compute
\begin{equation}
    \tr_{\mg{XY}_j} \!\! \big( \big[ \dyad{1,1} \otimes \sum_{\pi \in S_{|t|}}
    \left( \pi_{t \rightarrow t'} \PD_{\mg{Y}_{\mg{S},t}}  \otimes \pi_{t \rightarrow t'} \right) \big]_{\mg{XY}_j} \rho_j \big) = \frac{(D-j+1)}{D} \cdot p_j(1) \cdot \rho^{(1)}_{j-1},
\end{equation}
where $p_j(1) \rho^{(1)}_{j-1} \equiv \tr_{\mg{XY}_j}( \dyad{1}_{\mg{X}_{\mg{C},j}} \rho_j)$.
We have $p_j(0) + p_j(1) = 1$ by definition.

Iterating $k$ times and summing over each $t \in \{0,1\}^k$, we find
\begin{equation}
    \tr( B_{\mg{XY}} \cPD_{\mg{Y}} \dyad{\Psi} ) = \sum_{t \in \{0,1\}^k} (\mathfrak{D}_{|t|} / D^{|t|}) \cdot p(t)
\end{equation}
where we let $p(t)$ denote the product of each $p_j(0)$ or $p_j(1)$ encountered in the iteration, and we let $\mathfrak{D}_{|t|} = D!/(D-|t|+1)!$ denote the dimension of the distinct subspace on $t$ copies.
We have $\tr( B_{\mg{XY}} \cPD_{\mg{Y}} \dyad{\Psi} ) \leq 1$, because $\sum_t p(t) = 1$ and $\mathfrak{D}_{|t|} \leq D^{|t|}$.
We also have $\tr( B_{\mg{XY}} \cPD_{\mg{Y}} \dyad{\Psi} ) \geq 1-k^2/D$ because $\mathfrak{D}_{|t|} \geq D^{|t|} (1-k^2/D)$. \qed

\subsection{Proof of Theorem~\ref{thm: cRLFC}: Controlled-LRFC ensemble}

We follow the same notation as in the previous section.
We denote $\bs{L} = (O_L^c)^{\otimes k}$, $\bs{F} = (F^c)^{\otimes k}$, $\bs{R} = (S^c_R)^{\otimes k}$, $\bs{D} = (D^c)^{\otimes k}$.
To write the expected output state in a more compact form, we let
\begin{equation}
    B'_{\mg{XY}} \equiv (2D)^{2k} \cdot \E_{f_R, f, f_L} \big[ (\bs{R}^\dagger \bs{F}^\dagger \bs{L}^\dagger \otimes \mathbbm{1} )_{\mg{XY}} \dyad{\Psi_{\text{Bell}}}_{\mg{XY}} ( \bs{L} \bs{F} \bs{R} \otimes \mathbbm{1})_{\mg{XY}} \big]
\end{equation}
denote the Bell projector twirled over $O_R, F, O_L$, multiplied by $(2D)^{2k}$.
This yields,
\begin{equation}
	\rho_{\mg{AB}} \equiv \E_{f_R, f, f_L, C} \big[ \dyad{\psi_U}_{\mg{AB}} \big] = \E_{C} \big[ \tr_{\mg{XY}} \big( B'_{\mg{XY}} \bs{C}_\mg{X} \dyad*{\Psi}_{\mg{ABXY}} \bs{C}^\dagger_{\mg{X}} \big) \big].
\end{equation}
We let $\rho^H_{\mg{AB}} \equiv \E_{U \sim H} \dyad{\psi_U}_{\mg{AB}}$ denote the output of the same experiment when $U_c$ is controlled Haar-random.
We will show that $\lVert \rho - \rho^H \rVert_1 \leq 4 \delta$, where $\delta = k^2/2^n$.
This establishes  security against any sub-exponential time quantum adversary.


We begin by writing down several definitions and facts.
We define the controlled distinct subspace projectors, 
\begin{equation}
    \cPD_{\mg{X}_\mg{L}} = \sum_{t \in \{0,1\}^k} \dyad{t}_{\mg{X}_{\mg{L,C}}} \otimes \left[ \PD_t \otimes \mathbbm{1}_{\bar t} \right]_{\mg{X}_{\mg{L,S}}},
    \quad \quad \quad \cPD_{\mg{Y}_\mg{R}} = \sum_{t \in \{0,1\}^k} \dyad{t}_{\mg{Y}_{\mg{R,C}}} \otimes \left[ \PD_t \otimes \mathbbm{1}_{\bar t} \right]_{\mg{Y}_{\mg{R,S}}},
\end{equation}
as well as their product,
\begin{equation}
    \cPDLR \equiv \cPD_{\mg{X}_\mg{L}} \otimes \cPD_{\mg{Y}_\mg{R}}.
\end{equation}
Here, $\mathsf{X}_{\mathsf{L,C}}$ denotes the $k$ control registers in $\mathsf{X}_\mathsf{L}$, and $\mathsf{X}_{\mathsf{L,S}}$ the $k$ system register, and similar for $\mathsf{Y}_\mathsf{R}$.
Here, $\PD_t = \sum_{x \in \text{dist}} \dyad{x}$ is the projector onto the distinct subspace on the system registers of copies $j$ with $t_j = 1$, where sum runs over all distinct bitstrings in $[D_L]^{\otimes |t|}$ or $[D_R]^{\otimes |t|}$.
Also, $\mathbbm{1}_{\bar t}$ is the identity matrix on the copies $j$ with $t_j=0$.

In analogy to $B'$ above, we also define the operator,
\begin{equation}
    B_{\mg{XY}} = D^{2(k-|t|)} \sum_{\substack{ t, t' \in \{0,1\}^k \\ |t| = |t'|}} \left( \dyad{t', t'}{t, t}_{\mg{X}_\mg{C} \mg{Y}_\mg{C}} \otimes \sum_{\pi \in S_{|t|}} \left[ \left( \pi_{t \rightarrow t'} \otimes \pi_{t \rightarrow t'} \right) \otimes \dyad*{ \Psi^{\bar t '}_{\text{Bell}} }{ \Psi^{\bar t}_{\text{Bell}} } \right]_{\mg{X}_\mg{S} \mg{Y}_\mg{S}} \right),
\end{equation}
where $\pi_{t \rightarrow t'}$ is a permutation mapping the system registers of the copies of $\mathsf{X}$ where $t_j=1$ to the those where $t'_j = 1$, and similarly for $\mathsf{Y}$.
With these definitions in hand, our proof uses five facts:

\vspace{3mm}
\noindent \text{(1)}
    $B$ and $B'$ are equal on the distinct subspace, $B'_{\mg{XY}} \cdot \cPDLR
        =
        B_{\mg{XY}} \cdot \cPDLR$.
    To see this, we express the left side as 
    \begin{equation}
    B'_{\mg{XY}} =  \sum_{\substack{ t, t' \in \{0,1\}^k \\ |t| = |t'|}} D^{2(k-|t|)} \cdot \dyad{t, t}{t', t'}_{\mg{X}_\mg{C} \mg{Y}_\mg{C}} \otimes \left[ \left( \tilde{\pi}_{t \rightarrow t'} \otimes \tilde{\pi}_{t \rightarrow t'} \right) \cdot
    \left( B'_t \otimes \dyad*{ \Psi^{\bar t}_{\text{Bell}} }{ \Psi^{\bar t}_{\text{Bell}} } \right) \right]_{\mg{X}_\mg{S} \mg{Y}_\mg{S}} ,
\end{equation}
where $\tilde{\pi}_{t\rightarrow t'}$ is any permutation that maps copies $j$ where $t_j = 1$ to those where $t'_j = 1$, and we abbreviate
\begin{equation}
     (B'_{t})_{\mg{XY}_t} \equiv D^{2|t|} \E_{f_R, f, f_L} \big[ (O_R^{\dagger,\otimes |t|} F^{\dagger,\otimes |t|} O_L^{\dagger,\otimes |t|} \otimes \mathbbm{1} )_{\mg{XY}_t} \dyad{\Psi^t_{\text{Bell}}}_{\mg{XY}_t} ( O_L^{\otimes |t|} F^{\otimes |t|} O_R^{\otimes |t|} \otimes \mathbbm{1})_{\mg{XY}_t} \big]
\end{equation}
We can perform precisely the same decomposition for $B$, replacing $(B'_t)_{\mg{XY}_t}$ with $(B_t)_{\mg{XY}_t} \equiv \sum_{\pi \in S_{|t|}} \pi \otimes \pi$.
Inserting the distinct subspace projectors, we find
\begin{equation} \nonumber
\begin{split}
    B'_t & ( \PD_{\mg{X}_{\mg{L},t}} \otimes \PD_{\mg{Y}_{\mg{R},t}} )  = D^{|t|} \E_{f_R,f,f_L} \sum_{\substack{x_L, y_R \\ \text{dist}}}
        \sum_{\substack{ \tilde{x}_L, \tilde{y}_R \\ \text{dist}}}
        (-1)^{f(x_L \lVert y_R)+f(\tilde{x}_L \lVert \tilde{y}_R)} \ket{x_L \lVert y_R \oplus f_R(x_L)}_{\mg{X}_t} \otimes  \\
        & \quad  \quad \quad \quad \quad \quad \quad \quad  \quad \quad \quad \quad \quad \quad \quad \ket{x_L \oplus f_L(y_R) \lVert y_R}_{\mg{Y}_t}
        \bra{\tilde{x}_L \lVert \tilde{y}_R \oplus f_R(\tilde{x}_L)}_{\mg{X}_t} \otimes 
        \bra{\tilde{x}_L \oplus f_L(\tilde{y}_R) \lVert \tilde{y}_R}_{\mg{Y}_t} \\
        & = D^{2|t|} \E_{f_R,f_L} \sum_{\substack{x_L, y_R \\ \text{dist}}}
        \sum_{\pi \in S_{|t|}}
        \left( \dyad{x_L \lVert y_R \oplus f_R(x_L)} \pi \right)_{\mg{X}_t} \otimes \\
        & \quad  \quad \quad \quad \quad \quad \quad \quad  \quad \quad \quad \quad \quad \quad \quad  \left( \dyad{ x_L \oplus  f_L(y_R) \lVert 
 y_R} \pi \right)_{\mg{Y}_t},
\end{split}
\end{equation}
where in the second line, taking the expectation over $f$ enforces $x_L \lVert y_R = \pi(x_L \lVert y_R)$ for some $\pi \in S_{|t|}$.
Taking the expectation over $f_L$ and $f_R$ yields,
\begin{equation} 
    B'_t ( \PD_{\mg{X}_{\mg{L},t}} \otimes \PD_{\mg{Y}_{\mg{R},t}} )  = 
    \sum_{\substack{ x_L, y_R \\ \text{dist}}}
    \sum_{ y_L, x_R}
    \left( \dyad{x_L \lVert x_R} \pi \right)_{\mg{X}_t}
    \otimes 
    \left( \dyad{y_L \lVert y_R} \pi \right)_{\mg{Y}_t},
\end{equation}
since $y_R \oplus f_R(x_L)$ is a uniformly random bitstring in $[D_R]^{\otimes |t|}$ when $x_L$ is distinct, and similar for $x_L \oplus f_L(y_R)$.
This is clearly equal to $B_t ( \PD_{\mg{X}_{\mg{L},t}} \otimes \PD_{\mg{Y}_{\mg{R},t}} )$, using the resolution of the identity, $\mathbbm{1}_{\mg{XY}_t} = \sum_{x_L, x_R, y_L, y_R} \dyad{x_L \lVert x_R}_{\mg{X}_t} \otimes \dyad{y_L \lVert y_R}_{\mg{Y}_t}$.
    
\vspace{5mm}
\noindent \text{(2.)}
    The twirl of $\cPD_{\mg{X}_\mg{L}}$ over a controlled 2-design is close to the identity, $\big\lVert \mathbbm{1} - \E_{C \sim \mathfrak{D}} \big[ \bs{C}_\mg{X}^\dagger \cPD_{\mg{X}_\mg{L}} \bs{C}_\mg{X} \big] \big\rVert_\infty \leq \delta_L$.
    To show this, we bound $\mathbbm{1} - \cPD_{\mg{X}_\mg{L}} \leq \sum_{t\in \{0,1\}^k} \sum_{i<j \in t} \dyad{t}_{\mg{X}_{\mg{L,C}}} \otimes \Pi_{\mg{X}_{\mg{L,S}}}^{ij}$, where
    \begin{equation}
        \Pi_{\mg{X}_\mg{L,S}}^{\text{eq},ij} = \sum_{x\in [D_L]} \dyad{x, x}_{\mg{X}_{\mg{L}_\mg{S},i} \mg{X}_{\mg{L}_\mg{S},j}}.
    \end{equation}
    The twirl over a controlled 2-design yields $\E_{C} [ \bs{C}_\mg{X}^\dagger \sum_t \sum_{i<j \in t}  \dyad{t} \otimes \Pi_{\mg{X}_\mg{L,S}}^{\text{eq},ij} \bs{C}_\mg{X} ] = \sum_t \sum_{i<j \in t} \dyad{t} \otimes
    (a\mathbbm{1} + b\mathcal{S}_{ij})$, where $\mathcal{S}_{ij}$ is the swap operator, and $a = (D_R^2D_L-1)/(D_R^2D_L^2-1) \leq 1/D_L$ and $b = (D_R D_L - D_R)/(D_R^2D_L^2-1) \leq 1/D_R D_L$.
    Hence, the spectral norm is less $\max_t \sum_{i<j\in t} (a+b) \leq k^2(a+b)/2 \leq \delta_L (1+1/D_R) \leq \delta_L$.

\vspace{5mm}
\noindent \text{(3.)}
$B$ and $B'$ commute with $\cPD_{\mg{X}_\mg{L}}$ and $\cPD_{\mg{Y}_\mg{R}}$. 
$B$ also commutes with $\bs{C}$.

\vspace{5mm}
\noindent \text{(4.)} For any positive Hermitian matrices $P,Q,R$, where $P$ and $Q$ commute, we have $\tr(PQR)  \leq \lVert Q \rVert_\infty \tr(P R)$.
This follows from Holder's inequality,
$$\tr(PQR) = \text{tr}(Q \sqrt{P} R \sqrt{P}) \leq \lVert Q \rVert_\infty \lVert \sqrt{P} R \sqrt{P} \rVert_1 =  \lVert Q \rVert_\infty \tr(P R).$$

\vspace{5mm}
\noindent \text{(5.)} We have $1-\delta_R\leq \tr(B_{\mg{XY}} \cPD_{\mg{Y}_{\mg{R}}} \dyad{\Psi} ) \leq 1$.
This is shown below the current proof.

\vspace{5mm}
\noindent We can now prove the claim.
We insert $\mathbbm{1} = \cPDLR + (\mathbbm{1}-\cPDLR)$ to decompose $\rho$ into two terms,
\begin{align*}
	\rho = \rho_{\text{dist}} + \delta\rho \equiv \E_{C} \big[ \tr_{\mg{XY}}( B' \cdot \cPDLR \cdot \bs{C} \dyad*{\Psi} \bs{C}^\dagger ) \big] + \E_{C} \big[ \tr_{\mg{XY}}( B' \cdot (\mathbbm{1}-\cPDLR) \cdot \bs{C} \dyad*{\Psi} \bs{C}^\dagger ) \big].
\end{align*}
\noindent We bound the second term as follows,
\begin{align*}
    \lVert \delta \rho \rVert_1 = \tr( \delta \rho ) & = \E_{C} \big[ \tr( B' (\mathbbm{1}-\cPDLR) \bs{C} \dyad*{\Psi} \bs{C}^\dagger ) \big] && \text{(since $\delta \rho$ is positive)} \\
    & = 1 - \E_{C} \big[ \tr \big( B \cdot \cPDLR \cdot \bs{C}\dyad*{\Psi} \bs{C}^\dagger \big)  \big], && \text{(from $\tr(\rho) = 1$ and (1.))} \\
    & = 1 - \tr \big( B \cdot \E_{C} \big[ \bs{C}^\dagger \cPD_{\mg{X}_\mg{L}} \bs{C} \big] \otimes \cPD_{\mg{Y}_\mg{R}} \cdot \dyad*{\Psi} \big), && \text{(from (3.))} \\
    & \leq 1 - (1-\delta_L) \tr \big( B (\mathbbm{1} \otimes \cPD_{\mg{Y}_\mg{R}}) \dyad*{\Psi} \big) && \text{(from (2.), (3.), (4.))} \\
    & = 1 - (1-\delta_L)(1-\delta_R) \leq \delta_L + \delta_R. && \text{(from (5.))}
\end{align*}
Meanwhile, the first term is close to a fixed density matrix $\rho_a$, independent of the ensemble $\mathfrak{D}$,
\begin{align*}
    \rho_{\text{dist}} & = \tr_{\mg{XY}} \!  \big( B \cdot \E_{C} \big[ \bs{C}^\dagger \cPD_{\mg{X}_\mg{L}} \bs{C} \big] \otimes \cPD_{\mg{Y}_\mg{R}} \cdot \dyad*{\Psi} \big) && \text{(from (1.))} \\
    & = \tr_{\mg{XY}} \!  \big( B \cdot \mathbbm{1}_{\mg{X}_\mg{L}} \otimes \cPD_{\mg{Y}_\mg{R}} \cdot \dyad*{\Psi} \big) + \Delta \equiv \rho_a + \Delta, && \text{(from (2.), (3.), (4.))}
\end{align*} 
where $\lVert \Delta \rVert_1 \leq \delta_L  \tr( B \cdot \mathbbm{1}_{\mg{X}_\mg{L}} \otimes \cPD_{\mg{Y}_\mg{R}} \cdot \dyad*{\Psi}) \leq \delta_L$ from (2.), (3.), (4.).

This completes our proof: We have $\lVert \rho - \rho_a \rVert_1 \leq 2\delta_L + \delta_R$ when $C$ is drawn from any 2-design $\mathfrak{D}$.
Since the Haar ensemble is a 2-design, we have $\lVert \rho - \rho^H \rVert_1 \leq \lVert \rho^H - \rho_a \rVert_1 + \lVert \rho - \rho_a \rVert_1 \leq 4\delta_L + 2\delta_R = 4k^2/D_L + 2k^2/D_R$. \qed

\vspace{3mm}
\noindent \emph{Proof of (5.)}---We proceed by induction, similar to the previous section.
%
Let $\rho_j = \tr_{\mg{AB} \mg{XY}_{>j}}\!( \dyad*{\Psi} )$ denote the reduced density matrix on on $\mathsf{XY}_{\mathsf{R},\leq j}$.
The key property we use is that $\rho_j$ is maximally mixed on $\mathsf{Y}_j$. 

\vspace{2mm}
\noindent \textbf{Base case ($j=1$): }Here, $B = \mathbbm{1} \otimes \mathbbm{1}$ and $\cPD_{\mg{Y}_{\mg{R},\leq 1}} = \mathbbm{1}_{\mg{Y}_{\mg{R},1}}$. Hence, $\tr(B \cdot \cPD_{\mg{Y}_{\mg{R},\leq 1}} \cdot \dyad*{\Psi}) = 1$.

\vspace{2mm}
\noindent \textbf{Inductive step: }We can first expand,
\begin{equation} \nonumber
    B_{\mg{XY}} \cPD_{\mg{Y}_{\mg{R}}} =  \!\!\!\! \sum_{\substack{t,t' \in \{0,1\}^j \\ |t|=|t'|}}  \!\!
    D^{2(j-|t|)}
    \cdot \left( \dyad{t',t'}{t,t}_{\mg{X}_{\mg{C}}\mg{Y}_{\mg{C}}}
    \otimes \sum_{\pi \in S_{|t|}}
    \left[ \left( \pi_{t \rightarrow t'} \PD_{\mg{Y}_{\mg{R,S,t}}}  \otimes \pi_{t \rightarrow t'} \right) \otimes \dyad*{ \Psi^{\bar t '}_{\text{Bell}} }{ \Psi^{\bar t}_{\text{Bell}} } \right]_{\mg{X}_\mg{S} \mg{Y}_\mg{S}}
    \! \right).
\end{equation}
We now perform the trace over $\mathsf{Y}_j$.
Since $\rho_j$ is maximally mixed on $\mathsf{Y}_j$, the trace sets $t_j = t'_j$.
If $t_j = 0$, we find 
\begin{equation}
    \tr_{\mg{XY}_j} \!\! \big( \big[ \dyad{0,0} \otimes \dyad*{\Psi_{\text{Bell}}^j} \big]_{\mg{XY}_j} \rho_j \big) = p_j(0) \cdot \rho^{(0)}_{j-1} / D^2,
\end{equation}
where $p_j(0) \rho^{(0)}_{j-1} \equiv \tr_{\mg{XY}_j}( \dyad{0}_{\mg{X}_{\mg{C},j}} \rho_j)$ is the probability of measuring a $0$ on register $\mathsf{X}_{\mathsf{C},j}$, multiplied by the resulting normalized quantum state.
Meanwhile, if $t_j = 1$, we find
\begin{equation}
    \tr_{\mg{XY}_j} \!\! \big( \big[ \dyad{1,1} \otimes \sum_{\pi \in S_{|t|}}
    \left( \pi_{t \rightarrow t'} \PD_{\mg{Y}_{\mg{R,S,t}}}  \otimes \pi_{t \rightarrow t'} \right) \big]_{\mg{XY}_j} \rho_j. \big),
\end{equation}
Leveraging the identity, $$\tr_{\mg{Y}_{\mg{S},j}}( \pi_{\mg{Y}_\mg{S}} \PD_{\mg{Y}_{\mg{R,S}}} ) = D_L (D_R-j+1) \pi_{\mg{Y}_{\mg{S},< j}} \PD_{\mg{Y}_{\mg{R,S},< j}}$$ if $\pi(j)=j,$ and $\tr_{\mg{Y}_{\mg{S},j}}( \pi_{\mg{Y}_\mg{S}} \PD_{\mg{Y}_{\mg{R,S}}} ) = 0$ otherwise, we can compute
\begin{equation}
    \tr_{\mg{XY}_j} \!\! \big( \big[ \dyad{1,1} \otimes \sum_{\pi \in S_{|t|}}
    \left( \pi_{t \rightarrow t'} \PD_{\mg{Y}_{\mg{R,S,t}}}  \otimes \pi_{t \rightarrow t'} \right) \big]_{\mg{XY}_j} \rho_j \big) = \frac{D_L (D_R-j+1)}{D_L D_R} \cdot p_j(1) \cdot \rho^{(1)}_{j-1},
\end{equation}
where $p_j(1) \rho^{(1)}_{j-1} \equiv \tr_{\mg{XY}_j}( \dyad{1}_{\mg{X}_{\mg{C},j}} \rho_j)$.
We have $p_j(0) + p_j(1) = 1$ by definition.

Iterating $k$ times and repeating for every $t \in \{0,1\}^k$, we find
\begin{equation}
    \tr( B_{\mg{XY}} \cPD_{\mg{Y}_{\mg{R}}} \dyad{\Psi} ) = \sum_{t \in \{0,1\}^k} (\mathfrak{D}_R^{|t|} / D_R^{|t|}) \cdot p(t)
\end{equation}
where we let $p(t)$ denote the product of each $p_j(0)$ or $p_j(1)$ encountered in the iteration, and we let $\mathfrak{D}_R^{|t|} = ((D_R)!/(D_R-|t|+1)!)$ denote the dimension of the distinct subspace on $[D_R]^{\otimes |t|}$.
We have $\tr( B_{\mg{XY}} \cPD_{\mg{Y}_{\mg{R}}} \dyad{\Psi} ) \leq 1$, because $\sum_t p(t) = 1$ and $\mathfrak{D}_R^{|t|} \leq D_R^{|t|}$.
We also have $\tr( B_{\mg{XY}} \cPD_{\mg{Y}_{\mg{R}}} \dyad{\Psi} ) \geq 1-k^2/D_R$ because $\mathfrak{D}_R^{|t|} \geq D_R^{|t|} (1-k^2/D_R)$. \qed

\section{Translation-invariant random unitaries} \label{sec: TI}

In this section, we present the extension of our results to translation-invariant random unitaries.
For simplicity, we focus on random unitaries that are invariant with respect to translations over the same distance, $2\xi$, as the patch size in our two-layer brickwork construction.
We begin in Appendix~\ref{sec: intro TI} by establishing a few basic results on translation-invariant random unitaries.
We then analyze the gluing construction for translation-invariant unitaries in Appendix~\ref{sec: no on site TI}.
We show how this leads to extremely low depth translation-invariant PRUs. 
In Appendix~\ref{sec: on site TI}, we extend our results to translation-invariant unitaries that also have on-site symmetries.
%

\subsection{Approximate translation-invariant Haar twirl} \label{sec: intro TI}

We consider a periodic 1D chain of $n$ qubits.
We suppose that the system is symmetric under translations over a distance $2\xi$.
We let $\mathcal{T}$ denote the operator which implements this translation. 
We have $\mathcal{T}^m = 1$, where $m = n/2\xi$.

We are interested in the Haar measure over unitaries that commute with the translation symmetry operator.
To spare time, we will not analyze the exact expression involving Weingarten functions for the translation-invariant Haar twirl.
Instead, we will begin by positing the following approximate expression for the translation-invariant Haar twirl,
\begin{equation} \label{eq: approx TI Haar twirl}
    \Phi_a^{\text{TI}}(\rho) = \frac{1}{D^k} 
    \sum_\pi \sum_{\bs{t}} 
    \tr \bigg( \rho \cdot \pi^{-1} \cdot \bigotimes_{j=1}^k \mathcal{T}^{-t_j}  \bigg) 
    \cdot 
    \bigotimes_{j=1}^k \mathcal{T}^{t_j}
    \cdot 
    \pi.
\end{equation}
Here, $\pi \in S_k$ is summed over all permutations amongst the $k$ copies, and $\bs{t} = (t_1,\ldots,t_k)$ is summed over all translations, over distance $2\xi t_j $, in each copy $j$.
%

The key technical result of this subsection is the following lemma.
The lemma allows one to bound the relative error to the \emph{exact} translation-invariant Haar twirl, by bounding the additive error to the \emph{approximate} translation-invariant Haar twirl on the EPR state.

\begin{lemma}[Translation-invariant unitary designs from EPR states] \label{lemma: relative error to additive TI}
    An ensemble $\mathcal{E}$ of translation-invariant unitaries forms a translation-invariant $\varepsilon$-approximate unitary $k$-design with relative error
    \begin{equation}  \label{eq: relative EPR TI}
        \varepsilon =  \frac{3 D^{2k}}{n^k k!} \left\lVert [\delta \Phi \otimes \mathbbm{1}](P_{\text{EPR}}) \right\rVert_\infty,
    \end{equation}
    where $\delta \Phi = \Phi_\mathcal{E} - \Phi_a^{\text{\emph{TI}}}$, and $P_{\text{EPR}}$ is the EPR state on $\mathcal{H}^{\otimes k} \otimes \mathcal{H}^{\otimes k}$.
\end{lemma}
\noindent Curiously, the lemma can be proven even without knowing that $\Phi^{\text{TI}}_a$ is close to the translation-invariant Haar twirl in relative error.
The only property of $\Phi^{\text{TI}}_a$ that we require is that it is invariant under the Haar twirl, $\Phi^{\text{TI}}_H \circ \Phi^{\text{TI}}_a = \Phi^{\text{TI}}_a$.
This is convenient, as it allows us to sidestep a brute force Weingarten analysis of the exact twirl.
\begin{proof}
    The first part of our proof proceeds exactly analogously to the proof of Lemma~2 in Ref.~\cite{schuster2024random} and the proof of Lemma~\ref{lemma: relative error to additive} in this work.
    The sole difference is that the non-zero eigenvalue of $[\Phi_a^{\text{TI}} \otimes \mathbbm{1}](P_{\text{EPR}})$ is now equal to $n^k k! / D^{2k}$, owing to the $n^k k!$ permutations appearing $\Phi_a^{\text{TI}}$.
    This immediately shows that $\Phi_\mathcal{E}$ is close to the \emph{approximate} Haar twirl, 
    \begin{equation} \label{eq: close TI}
        (1-\varepsilon') \Phi^{\text{TI}}_a \preceq \Phi_\mathcal{E} \preceq (1+\varepsilon') \Phi^{\text{TI}}_a,
    \end{equation}
    up to relative error $\varepsilon' = (D^{2k}/n^k k!) \left\lVert [\delta \Phi \otimes \mathbbm{1}](P_{\text{EPR}}) \right\rVert_\infty$.

    To bound the relative error between $\Phi_\mathcal{E}$ and the \emph{exact} Haar twirl $\Phi^{\text{TI}}_H$, we develop the following short-cut.
    We first apply the exact Haar twirl to each term in Eq.~(\ref{eq: close TI}).
    Using $\Phi_H^{\text{TI}} \circ \Phi_\mathcal{E} = \Phi_H^{\text{TI}}$ (from the definition of the Haar measure) and $\Phi_H^{\text{TI}} \circ \Phi_a^{\text{TI}} = \Phi_a^{\text{TI}}$ (since $\Phi_a^{\text{TI}}(\cdot)$ is a sum of operators that commute with any translation-invariant unitary), we find,
    \begin{equation} \label{eq: phi E phi TI a}
        (1-\varepsilon') \Phi^{\text{TI}}_a \preceq \Phi_H^{\text{TI}} \preceq (1+\varepsilon') \Phi^{\text{TI}}_a.
    \end{equation}
    Combining this with Eq.~(\ref{eq: close TI}), we have
    \begin{equation} 
        (1-3\varepsilon') \Phi^{\text{TI}}_H \preceq \frac{1-\varepsilon'}{1+\varepsilon'} \Phi^{\text{TI}}_a \preceq \Phi_\mathcal{E} \preceq \frac{1+\varepsilon'}{1-\varepsilon'} \Phi^{\text{TI}}_H \preceq (1+3\varepsilon') \Phi^{\text{TI}}_H,
    \end{equation}
    where the outside inequalities hold for $3 \varepsilon' \leq 1$.
    Setting $\varepsilon = 3\varepsilon'$ completes the proof.
\end{proof}

\subsection{Translation-invariant PRUs in extremely low depth} \label{sec: no on site TI}

We now show that the two-layer brickwork construction naturally allows one to construct translation-invariant random unitaries.
As before, we break the system into $2m = n/\xi$ patches of $\xi$ qubits each, and apply a translation-invariant version of the two-layer circuit,
\begin{equation} \label{eq: TI TL BW circuit}
    U = U_2^{\otimes m}  \cdot U_1^{\otimes m},
\end{equation}
where the first layer is a tensor product of $m$ applications of the unitary $U_1$ acting on all even nearest-neighbor pairs of patches, and the second layer is a tensor product of $m$ applications of $U_2$ acting on all odd nearest-neighbor pairs of patches.
The unitary is clearly symmetric under translations, $\mathcal{T} U = U \mathcal{T}$.

Our main result is the following theorem, which shows that $U$ forms a translation-invariant random unitary design.

\begin{theorem}[Gluing translation-invariant random unitaries] \label{thm: gluing TI}
    For any $\varepsilon \leq 1$.
    Consider the translation-invariant two-layer brickwork circuit, $U$ [Eq.~(\ref{eq: TI TL BW circuit})].
    Suppose that the small random unitaries, $U_1$ and $U_2$, are each drawn from $\frac{\varepsilon}{4}$-approximate unitary $mk$-designs.
    Then $U$ forms an $\varepsilon$-approximate translation-invariant unitary $k$-design, as long as $\xi \geq \log_2(32 n^6 k^6 / \varepsilon)$.
\end{theorem}

\noindent By drawing each of $U_1$ and $U_2$ from a PRU ensemble on $2\xi$ qubits with security against sub-exponential time quantum adversaries, the theorem immediately leads to translation-invariant pseudorandom unitaries in extremely low depth.

\begin{corollary}[Translation-invariant PRUs in extremely low depth]
    Translation-invariant PRUs exist, and can be formed in depth $\poly(\log n)$ in geometrically-local circuits, and $\poly \log \log n$ in all-to-all connected circuits.
\end{corollary}




\begin{proof}[Proof of Theorem~\ref{thm: gluing TI}]
We let $K \equiv mk$, and set $D \equiv 2^\xi$ for the duration of this proof.
By assumption, the twirl when $U_1$ and $U_2$ are drawn from approximate designs is close to the twirl when they are both drawn from the Haar measure, up to relative error $(1+\varepsilon/4)(1+\varepsilon/4)$.
The latter is given by,
\begin{equation} \label{eq: Phi TI E EPR}
\begin{split}
    [ \Phi_{\mathcal{E}} \otimes \mathbbm{1} ](P_{\text{EPR}})
    & =
    \frac{1}{D^{3k}} \sum_{\sigma,\tilde{\sigma} \in S_K}
    \tr_e ( (\mathcal{T} \tilde{\sigma} \mathcal{T}^{-1})^{-1} \sigma )
    \cdot
    \tr_o ( \tilde{\sigma}^{-1} \sigma )
    \cdot
    [ ( \mathcal{T} \tilde{\sigma} \mathcal{T}^{-1} )_e \otimes \tilde{\sigma}_o ]
    \otimes 
    [ \sigma_e \otimes \sigma_o ] \\
    & =
    \frac{1}{D^{2k}} \sum_{\sigma,\tilde{\sigma} \in S_K}
    D^{-| \tilde{\sigma}^{-1}  \mathcal{T}^{-1} \sigma \mathcal{T} | }
    \cdot
    D^{-| \tilde{\sigma}^{-1} \sigma |} 
    \cdot
    [ ( \mathcal{T} \tilde{\sigma} \mathcal{T}^{-1} )_e \otimes \tilde{\sigma}_o ]
    \otimes 
    [ \sigma_e \otimes \sigma_o ],
\end{split}
\end{equation}
up to an additional relative error $(1+ K^2 / D^{2})^2$ from replacing each small Haar twirl with its approximate form (from Lemma~2 of Ref.~\cite{schuster2024random}).
Here, $\sigma,\tilde{\sigma} \in S_K$ are summed over all possible permutations of the $K = mk$ patches within the $k$ copies.
We can also consider the analogous expression for the approximate Haar twirl,
\begin{equation} \label{eq: Phi TI a EPR}
\begin{split}
    [ \Phi_{a}^{\text{TI}} \otimes \mathbbm{1} ](P_{\text{EPR}})
    & =
    \frac{1}{D^{2k}} \!\! \sum_{\sigma  \, : \,  [\sigma,\mathcal{T}^{\otimes k}] = 0  } \!\!
    [ \sigma_e \otimes \sigma_o ]
    \otimes 
    [ \sigma_e \otimes \sigma_o ],
\end{split}
\end{equation}
where $\sigma$ is summed over the $n^k k!$ permutations in $S_K$ that commute with $\mathcal{T}^{\otimes k}$.
Note that this summation is precisely equivalent to the summation over $\pi$ and $\bs{t}$ in Eq.~(\ref{eq: approx TI Haar twirl}), upon identifying $\sigma = \bigotimes_{j=1}^k \mathcal{T}^{t_j} \cdot \pi$.
The terms in Eq.~(\ref{eq: Phi TI a EPR}) exactly correspond to the terms in Eq.~(\ref{eq: Phi TI E EPR}) with $\sigma = \tilde{\sigma}$ and $[\sigma,\mathcal{T}^{\otimes k}] = 0$.
Taking the difference between the two expressions and applying the triangle inequality, we have,
\begin{equation}
    \left\lVert [ \delta \Phi \otimes \mathbbm{1} ](P_{\text{EPR}}) \right\rVert_\infty
     \leq
    \frac{1}{D^{2k}} \left( \sum_{\sigma,\tilde{\sigma} \in S_K}
    D^{-| \tilde{\sigma}^{-1}  \mathcal{T}^{-1} \sigma \mathcal{T} | }
    D^{-| \tilde{\sigma}^{-1} \sigma |} 
    - n^k k! \right),
\end{equation}
where the negative term accounts for the $n^k k!$ permutations with $\sigma = \tilde{\sigma}$ and $[\sigma,\mathcal{T}^{\otimes k}] = 0$.

We bound the remaining sum in two steps. First, to compute the sum over $\tilde \sigma$, we prove the following proposition.
\begin{proposition} \label{prop: perm bound 1}
    $\sum_{\Delta \in S_K} D^{-|\Delta|} D^{-| \Delta^{-1} \tau |} \leq 8 (K^2 /2D)^{-|\tau|}$ for any permutation $\tau \in S_K$.
\end{proposition}
\begin{proof}
    Let us write $\sum_{\Delta \in S_K} D^{-|\Delta|} D^{-| \Delta^{-1} \tau |} = \sum_{r_1,r_2} N(r_1,r_2) D^{-r_1 - r_2}$, where $N(r_1,r_2)$ counts the number of $\Delta$ with both $|\Delta| = r_1$ and $| \Delta^{-1} \tau | = r_2$.
    We can upper bound,
    $$N(r_1,r_2) \leq \min( N_1(r_1) , N_2(r_2) ),$$
    where $N_1(r_1)$ counts the number of $\Delta$ with $|\Delta| = r_1$, and $N_2(r_2)$ counts the number with $|\Delta^{-1} \tau| = r_2$.
    We have $N_1(r_1) \leq (K^2/2)^{r_1}$, since there are at most ${K \choose 2} \leq K^2/2$ ways to place each of $r_1$ transpositions, to obtain $\Delta$ from the identity permutation.
    Similarly, $N_2(r_2) \leq (K^2/2)^{r_2}$.
    Thus, we have $\sum_{\Delta \in S_K} D^{-|\Delta|} D^{-| \Delta^{-1} \tau |} \leq \sum_{r_1,r_2} (K^2/2)^{\min(r_1,r_2)} D^{-r_1-r_2}$.

    To proceed, let us re-arrange the sums via $\sum_{r_1,r_2} = \sum_r \sum_{r_1+r_2=r}$.
    Note that for a fixed $r \equiv r_1+r_2$, there are $r+1$ valid values of $r_1,r_2$.
    We can quickly upper bound the sum over these $r_1,r_2$, 
    \begin{equation}
        \sum_{r_1+r_2=r} (K^2/2)^{\min(r_1,r_2)} \leq 2 \sum_{r_1=0}^{\lfloor r/2 \rfloor} (K^2/2)^{r_1} = 2 \frac{(K^2/2)^{\lfloor r/2 \rfloor}-2/K^2}{1-2/K^2} \leq 4 (K^2/2)^{r/2},
    \end{equation}
    where the final inequality assumes $K^2 \geq 4$.
    This yields, $\sum_{\Delta \in S_K} D^{-|\Delta|} D^{-| \Delta^{-1} \tau |} \leq 4 \sum_{r} (K^2/2)^{r/2} D^{-r}$.
    Now, we have $|\Delta^{-1} \tau | \geq |\tau| - |\Delta|$, since each transposition in $\Delta$ can decrease the distance of $\Delta^{-1} \tau$ by at most one.
    Thus, the sum is restricted to $r = r_1 + r_2 \geq |\tau|$.
    We can compute,
    \begin{equation}
        4 \sum_{r= |\tau|}^{2K} (K^2/2)^{r/2} D^{-r} \leq 4 \left( \frac{K^2}{2D} \right)^{|\tau|} \frac{1}{1-K^2/2D} \leq 8 \left( \frac{K^2}{2D} \right)^{|\tau|},
    \end{equation}
    where the final inequality assumes $K^2/D \leq 1$. This completes the proof.
\end{proof}
\noindent Applying Proposition~\ref{prop: perm bound 1} to our spectral norm bound, with $\tau = \sigma^{-1} \mathcal{T}^{-1} \sigma \mathcal{T}$, yields,
\begin{equation}
    \left\lVert [ \delta \Phi \otimes \mathbbm{1} ](P_{\text{EPR}}) \right\rVert_\infty
     \leq 
    \frac{8}{D^{2k}} \sum_{\sigma \, : \, [ \sigma, \mathcal{T}^{\otimes k}] \neq 0}
    (K^2/2D)^{-| \sigma^{-1}  \mathcal{T}^{-1} \sigma \mathcal{T} | }.
\end{equation}
As our second step, we compute the remaining sum using the following counting argument.

\begin{proposition} \label{prop: perm bound 2}
    The number of permutations $\sigma \in S_K$ with $|\sigma^{-1} \mathcal{T}^{-1} \sigma \mathcal{T}| = r$ is upper bounded by $(n^k k!) K^{4r}$.
\end{proposition}
\begin{proof}
    Any permutation of distance $r$ has at most $2r$ points that are not fixed.
    In the present context, a point $(i,j) \in [n] \otimes [k]$ is fixed if $\sigma((i,j)+(1,0)) = \sigma((i,j))+(1,0)$.
    There are ${K \choose 2r} \leq K^{2r}$ ways to choose the location of the non-fixed points (i.e. the $(i,j)$ for which this condition is not satisfied), and at most $K!/(K-2r)! \leq K^{2r}$ locations they can be sent to by $\sigma$.
    Thus, there are at most $K^{4r}$ configurations.
    
    For the remaining, fixed, points, we note that, as soon as $\sigma((i,j))$ is decided for all the non-fixed points in a given copy $j$, the values of $\sigma((i,j))$ for the remaining fixed points in that copy are fully determined.
    This follows because the value of $\sigma((i,j))$ for each fixed point is determined by the value at $(i+1,j)$, which is determined by $(i+2,j)$, and so on, until one reaches a non-fixed point.
    The only exception to this argument is if all points in a given copy $j$ are fixed.
    In this case, we are free to choose $\sigma((i,j))$ for a \emph{single} point in copy $j$, and the remaining $n-1$ points are fully determined by this choice.
    Our choice must place $\sigma((i,j))$ into a new copy where no points have been placed thus far, because the fixed point condition will place all $n$ qubits from copy $j$ onto this new copy.
    There are at most $k$ copies with no non-fixed points.
    Thus, there are at most $n^k k!$ ways to place the fixed points.
    
    Combining our two upper bounds, we find that there are at most $(n^k k!) K^{4r}$ ways to place the fixed and non-fixed points of $\sigma$.
    This completes our proof. 
\end{proof}

\noindent Applying Proposition~\ref{prop: perm bound 2} to our spectral norm bound yields,
\begin{equation}
    \left\lVert [ \delta \Phi \otimes \mathbbm{1} ](P_{\text{EPR}}) \right\rVert_\infty
     \leq 
    \frac{n^k k!}{D^{2k}} \cdot 8 \sum_{r=1}^{\infty}(K^6/2D)^{- r }
    = \frac{n^k k!}{D^{2k}} \cdot 8 K^6/D,
\end{equation}
where the final inequality holds for $K^6 \leq D$.
Applying Lemma~\ref{lemma: relative error to additive TI} shows that our ensemble is a translation-invariant unitary $k$-design with relative error $(1+\varepsilon/4)^2 (1+ K^2 / D^{2})^2 (1+8K^6/D)$.
We set $8K^6/D \leq \varepsilon/4$, i.e.~$\xi \geq \log_2(32 n^6 k^6/\varepsilon)$, which also yields $K^2/D^2 \leq \varepsilon/(32 D K^4) \leq \varepsilon/1088$, assuming trivially that $D, K \geq 2$.
The inequality $(1+\varepsilon/4)^3 (1+\varepsilon/1088)^2 \leq 1+\varepsilon$ for $\varepsilon \leq 1$ completes our proof.
\end{proof}

\subsection{Symmetric translation-invariant PRUs in extremely low depth} \label{sec: on site TI}

We conclude our discussion of translation-invariant random unitaries by showing that the results of the previous two subsections also extend to translation-invariant random unitaries with any discrete on-site symmetry.
We will move relatively quickly through the details in this section, since the derivations exactly mimic those previous.

We consider a periodic 1D chain of $n$ qubits. 
We suppose that the system possesses a translation symmetry $\mathcal{T}$ over a distance $2\xi$, as well as a discrete on-site symmetry $R_g = \bigotimes_{i=1}^n R^i_g$, where $g \in G$ is the symmetry group element.
As in Appendix~\ref{sec: no on site TI}, we will begin by positing an expression for the symmetric translation-invariant Haar twirl,
\begin{equation}
    \Phi_a^{G,\text{TI}}(\rho) = \frac{1}{D^k} 
    \sum_\pi \sum_{\bs{t}} \sum_{\bs{g}}
    \tr \bigg( \rho \cdot \pi^{-1} \cdot \bigotimes_{j=1}^k \mathcal{T}^{-t_j}  \cdot R_{\bs{g}}^{-1} \bigg) 
    \cdot 
    R_{\bs{g}}
    \cdot
    \bigotimes_{j=1}^k \mathcal{T}^{t_j}
    \cdot 
    \pi.
\end{equation}
The expression is precisely analogous to our previous formulas for the approximate symmetric and translation-invariant Haar twirls.
Next, we can establish an analogous technical lemma for bounding the relative error of any symmetric translation-invariant unitary design.
\begin{lemma}[Symmetric translation-invariant unitary designs from EPR states] \label{lemma: relative error to additive TI sym}
    For any discrete on-site symmetry group $G$. 
    An ensemble $\mathcal{E}$ of symmetric translation-invariant unitaries forms a symmetric translation-invariant $\varepsilon$-approximate unitary $k$-design with relative error
    \begin{equation}  \label{eq: relative EPR TI 2}
        \varepsilon =  \frac{3 D^{2k}}{|G|^k n^k k!} \left\lVert [\delta \Phi \otimes \mathbbm{1}](P_{\text{EPR}}) \right\rVert_\infty,
    \end{equation}
    where $\delta \Phi = \Phi_\mathcal{E} - \Phi_a^{G,\text{\emph{TI}}}$, and $P_{\text{EPR}}$ is the EPR state on $\mathcal{H}^{\otimes k} \otimes \mathcal{H}^{\otimes k}$.
\end{lemma}
\begin{proof}
    The proof of the lemma proceeds step-by-step in the exact same manner as the proof of Lemma~\ref{lemma: relative error to additive TI}.
    The additional factor of $|G|^k$ arises from the sum over symmetry operators $R_{\bs{g}}$.
\end{proof}

We can now analyze the symmetric translation-invariant two-layer circuit.
Similar to in Appendix~\ref{sec: no on site TI}, we have $U = U_2^{\otimes m}  \cdot U_1^{\otimes m}$, where $U_1, U_2$ are symmetric random unitaries on $2\xi$ qubits each, and $U_1^{\otimes m}$ acts on even pairs of patches and $U_2^{\otimes m}$ acts on odd pairs of patches.
Our main result is as follows.
\begin{theorem}[Gluing symmetric translation-invariant random unitaries] \label{thm: gluing TI sym}
    For any discrete on-site symmetry group $G$ and any $\varepsilon \leq 1$.
    Consider the symmetric translation-invariant two-layer brickwork circuit, $U$, described above.
    Suppose that the small random unitaries, $U_1$ and $U_2$, are each drawn from $\frac{\varepsilon}{4}$-approximate symmetric unitary $mk$-designs.
    Then $U$ forms an $\varepsilon$-approximate symmetric translation-invariant unitary $k$-design, as long as $\xi \geq \log_2(32 |G| n^6 k^6 / \varepsilon)$.
\end{theorem}
\noindent By drawing each of $U_1$ and $U_2$ from a symmetric PRU ensemble on $2\xi$ qubits with security against sub-exponential time quantum adversaries (Appendix~\ref{sec: symmetric PRUs}), Theorem~\ref{thm: gluing TI sym} immediately allows one to construct symmetric translation-invariant pseudorandom unitaries in extremely low depth.

\begin{corollary}[Symmetric translation-invariant PRUs in extremely low depth]
    Symmetric translation-invariant PRUs exist, and can be formed in depth $\poly(\log n)$ in geometrically-local circuits, and $\poly \log \log n$ in all-to-all connected circuits.
\end{corollary}

\begin{proof}[Proof of Theorem~\ref{thm: gluing TI sym}]
    We proceed as in the proof of Theorem~\ref{thm: gluing TI}.
The twirl when both $U_1$ and $U_2$ are drawn from the symmetric Haar ensembles is approximated by,
\begin{equation} \label{eq: Phi TI E EPR sym}
\begin{split}
    &[ \Phi_{\mathcal{E}}  \otimes \mathbbm{1} ](P_{\text{EPR}}) \\
    & =
    \frac{1}{D^{3k}} \sum_{\sigma,\tilde{\sigma}} \sum_{\bs{g},\tilde{\bs{g}}}
    \tr_e ( (\mathcal{T} R_{\tilde{\bs{g}}}^{} \tilde{\sigma} \mathcal{T}^{-1})^{-1} R_{\bs{g}} \sigma )
    \cdot
    \tr_o ( \tilde{\sigma}^{-1} R_{\tilde{\bs{g}}}^{-1} R_{\bs{g}}^{} \sigma )\\
    & \quad\quad\quad \cdot
    [ ( \mathcal{T} R_{\tilde{\bs{g}}} \tilde{\sigma} \mathcal{T}^{-1} )_e \otimes (R_{\tilde{\bs{g}}} \tilde{\sigma})_o ]
    \otimes 
    [ (R_{\bs{g}} \sigma)_e \otimes (R_{\bs{g}} \sigma)_o ], \\
\end{split}
\end{equation}
up to relative error $(1+|G|K^2/D^2)^2$.
The pair of traces are equal to
\begin{align}
    \tr_o ( \tilde{\sigma}^{-1} R_{\tilde{\bs{g}}}^{-1} R_{\bs{g}}^{} \sigma )
    & = D^{-| \tilde{\sigma}^{-1} \sigma |} 
    \cdot
    \prod_{\ell \in \sigma\tilde{\sigma}^{-1}} \delta \bigg( \prod_{i \in \ell} \tilde{g}_i g_i^{-1} = e \bigg) \\
    \tr_e ( (\mathcal{T} R_{\tilde{\bs{g}}}^{} \tilde{\sigma} \mathcal{T}^{-1})^{-1} R_{\bs{g}} \sigma )
    & = 
    D^{-| \tilde{\sigma}^{-1}  \mathcal{T}^{-1} \sigma \mathcal{T} | }
    \cdot
    \prod_{\ell \in \tilde{\sigma}^{-1}  \mathcal{T}^{-1} \sigma \mathcal{T} } \delta \bigg( \prod_{i \in \ell} \tilde{g}_i g_i^{-1} = e \bigg),
\end{align}
which is identical to the translation-invariant circuit without symmetries, except for the additional condition that the symmetry operators within each cycle $\ell$ product to the identity.
This uniquely fixes $\tilde{\bs{g}}$ within each unit cycle.
The analogous expression for the approximate symmetric translation-invariant Haar twirl is
\begin{equation} \label{eq: Phi TI a EPR 2}
\begin{split}
    [ \Phi_{a}^{G,\text{TI}} \otimes \mathbbm{1} ](P_{\text{EPR}})
    & =
    \frac{1}{D^{2k}} \!\! \sum_{\sigma  \, : \,  [\sigma,\mathcal{T}^{\otimes k}] = 0  } \!\!
    \sum_{\bs{g}}
    [ ( R_{\bs{g}} \sigma )_e \otimes ( R_{\bs{g}} \sigma )_o ]
    \otimes 
    [ ( R_{\bs{g}} \sigma )_e \otimes ( R_{\bs{g}} \sigma )_o ].
\end{split}
\end{equation}
Taking the difference of the two and applying the triangle inequality gives,
\begin{align}
    &\left\lVert [ \delta \Phi \otimes \mathbbm{1} ](P_{\text{EPR}}) \right\rVert_\infty\\
    &\leq
    \frac{1}{D^{2k}} \left( \sum_{\sigma,\tilde{\sigma}} \sum_{\bs{g}, \tilde{\bs{g}}}
    D^{-| \tilde{\sigma}^{-1}  \mathcal{T}^{-1} \sigma \mathcal{T} | }
    D^{-| \tilde{\sigma}^{-1} \sigma |} 
    \!\!\!\!\! \prod_{\ell \in \tilde{\sigma}^{-1}  \mathcal{T}^{-1} \sigma \mathcal{T} } \!\!\!\! \delta \bigg( \prod_{i \in \ell} \tilde{g}_i g_i^{-1} \bigg)
    \! \prod_{\ell \in \sigma\tilde{\sigma}^{-1}} \! \delta \bigg( \prod_{i \in \ell} \tilde{g}_i g_i^{-1} \bigg)
    - |G|^k n^k k! \right),
\end{align}
As in the proof of the symmetric gluing lemma, Lemma~\ref{lemma: AB BC to ABC app}, we can upper bound the sum over $\tilde{\bs{g}}$ by noting that each cycle (of either permutation) uniquely fixes one value of $\tilde{g}_i$.
There are therefore, at most, $|G|^{\min(| \tilde{\sigma}^{-1} \sigma |,| \tilde{\sigma}^{-1}  \mathcal{T}^{-1} \sigma \mathcal{T} |)}$ values that $\tilde{\bs{g}}$ can take.
Using $\min(x,y) \leq x + y$, we can upper bound our previous expression by,
\begin{equation}
\begin{split}
    \left\lVert [ \delta \Phi \otimes \mathbbm{1} ](P_{\text{EPR}}) \right\rVert_\infty
     & \leq
    \frac{1}{D^{2k}} \left( \sum_{\sigma,\tilde{\sigma}} \sum_{\bs{g}}
    (D/|G|)^{-| \tilde{\sigma}^{-1}  \mathcal{T}^{-1} \sigma \mathcal{T} | }
    (D/|G|)^{-| \tilde{\sigma}^{-1} \sigma |} 
    - |G|^k n^k k! \right) \\
    & =
    \frac{|G|^k}{D^{2k}} \left( \sum_{\sigma,\tilde{\sigma}}
    (D/|G|)^{-| \tilde{\sigma}^{-1}  \mathcal{T}^{-1} \sigma \mathcal{T} | }
    (D/|G|)^{-| \tilde{\sigma}^{-1} \sigma |} 
    - n^k k! \right),
\end{split}
\end{equation}
where in the second line we perform the sum over $\bs{g}$, which yields a factor of $|G|^k$.
The sum is identical to the sum that we bounded in the proof of Theorem~\ref{thm: gluing TI}, replacing $D \rightarrow D/|G|$.
Making this replacement in the bounds in Theorem~\ref{thm: gluing TI} completes our proof. 
\end{proof}

\section{Quantum phases of matter} \label{sec: phases of matter}

In this section, we provide full details on the implications of our results to the hardness of recognizing phases of matter in quantum systems.
We begin by precisely stating the definitions of phases of matter that our results apply to.
We then provide detailed proofs of our Theorems~\ref{thm:quantum} and~\ref{thm:polylog} from the main text.

There exist many closely-related definitions of phases of matter in the literature.
Our results apply to any definitions that obey the following criteria.
\begin{definition}
    [Pure state phases of matter]
    Our results apply to any definition of pure state phases of matter such that:
    \begin{itemize}
        \vspace{-1mm}
        \item For each phase, one can associate a reference ``fixed point'' state $\ket{\psi_0}$.
        \vspace{-1mm}
        \item Any state $\ket{\psi}$ that can be mapped to $\ket{\psi_0}$ via a symmetric geometrically-local depth-$\ell$ light-cone-$\xi$ unitary circuit, $U \ket{\psi} = \ket{\psi_0}$, is in the same phase as $\ket{\psi_0}$.
    \end{itemize}
\end{definition}

\begin{definition}
    [Mixed state phases of matter]
    Our results apply to any definition of mixed state phases of matter such that:
    \begin{itemize}
        \vspace{-1mm}
        \item For each phase, one can associate a reference ``fixed point'' state $\rho_0$.
        \vspace{-1mm}
        \item Any state $\rho$ that can be mapped to and from $\rho_0$ via symmetry-covariant geometrically-local depth-$\ell$ light-cone-$\xi$  quantum channels, $\mathcal{C}_1(\rho) = \rho_0$ and $\mathcal{C}_2(\rho_0) = \rho$, is in the same phase as $\rho_0$.
    \end{itemize}
\end{definition}
\noindent To be precise, in the second definition, the only channels that we require are combinations of symmetric local unitary circuits and adding and removing ancilla qubits in the maximally mixed state.

\subsection{Proof of Theorem~\ref{thm:quantum}: Hardness of recognizing pure state phases of matter}

Let us first address the setting when $\xi = \omega(\log n)$. 
Here, Theorem~\ref{thm:quantum} follows immediately from our construction of symmetric PRUs on $n$ qubits in Theorem~\ref{thm:polylog}.
Consider two quantum states $\ket{\psi_{\text{P1}}} = U \ket{\psi_{0,\text{P1}}}$ and $\ket{\psi_{\text{P2}}} = U \ket{\psi_{0,\text{P2}}}$.
Here, P1 and P2 denote two different phases of matter (for example, a trivial and non-trivial phase).
The states $\ket{\psi_{0,\text{P1}}}$ and $\ket{\psi_{0,\text{P2}}}$ denote their fixed point states.
The states $\ket{\psi_{\text{P1}}}$ and $\ket{\psi_{\text{P2}}}$ are random and are obtained by drawing $U$ from the two-layer circuit with light-cone $\xi$, circuit depth $\poly(\xi)$, and security against any $\exp(o(\xi))$-time adversary (Theorem~\ref{thm:polylog}).
Without loss of generality, we assume that $\ket{\psi_{0,\text{P1}}}$ and $\ket{\psi_{0,\text{P2}}}$ are symmetric, since any phase of matter has a symmetric representative.
By the definition of a symmetric PRU, $U$ cannot be distinguished from a symmetric Haar-random unitary by any $\exp(o(\xi))$-time algorithm.
Hence, neither $\ket{\psi_{\text{P1}}}$ nor $\ket{\psi_{\text{P2}}}$ can be distinguished from a Haar-random symmetric state.
Hence, the two random states also cannot be distinguished from each other.
If the phase of matter could be recognized, then one could use this to distinguish the states. 
Hence, the phase of matter cannot be recognized.

We now address smaller values of $\xi$, i.e.~$\xi = \mathcal{O}(\log n)$ or smaller.
In this regime, the required circuit size will always be polynomial or smaller in $n$, but may still grow exponentially in the value of $\xi$.
As before, we draw the random unitary $U$ from the two-layer circuit ensemble with light-cone $\xi$, however $\xi$ is now too small for the two-layer circuit to form a PRU on $n$ qubits.
Let $C$ denote the circuit size of the quantum or classical algorithm.
By assumption, $C$ is smaller than any exponential in $\xi$, hence $C$ is smaller than any constant fractional power of $n$.
Now, any circuit of size $C$ can act on at most $C$ qubits of the unknown quantum state.
Hence,  we can trace out all but $C$ of the small random unitaries in the two-layer circuit.
The output of the experiment is unaffected by the details of the two-layer circuit in the traced out regions.
Thus, the output would be the same if the $\mathcal{O}(n)$ traced out small random unitaries were replaced with $\mathcal{O}(C)$ small random unitaries that glue together the remaining un-traced out regions that the algorithm acts on.
Hence, the output is identical to the output of a fictitious two-layer circuit that acts on only $\mathcal{O}(C)$ qubits instead of $n$ qubits.
The fictitious two-layer circuit forms a symmetric PRU with security against any $\poly(C)$-time adversary whenever $\xi = \omega(\log C)$.
Hence, recognizing the phase of matter requires $\xi = \mathcal{O}(\log C)$, which implies that the circuit size must grow exponentially in $\xi$. \qed

\subsection{Proof of Theorem~\ref{thm:mixed}: Hardness of recognizing mixed state phases of matter}

The proof follows by an identical argument to Theorem~\ref{thm:quantum} with only small modifications.
Let $\rho_0$ denote the fixed point of any mixed-state phase of matter.
Our first step is to take the tensor product of $\rho_0$ with $n$ ancilla qudits in the maximally mixed state, yielding $\tilde{\rho}_0 \equiv \rho_0 \otimes (\mathbbm{1}/d)^{\otimes n}$ where $d$ is the local qudit dimension.
Here, the system symmetry $G$ does not act on the ancilla qubits.
We assume that each ancilla qudit is placed geometrically next to a corresponding qudit in the original system.
Our second step is to apply a two-layer circuit of small symmetric PRUs on $2\xi$ pairs of qudits each.
This yields a random state $\rho = U \tilde{\rho_0} U^\dagger$  in the same phase as $\tilde{\rho_0}$ and hence $\rho_0$. 

We will now show that the random state $\rho$ cannot be distinguished a mixture of maximally mixed states in each symmetry sector. 
Let $\rho_{\text{mixed}} = \bigotimes_\lambda p_\lambda (\mathbbm{1}_{\mathcal{M}_\lambda} \otimes \rho_{\mathcal{F}_\lambda})$ denote the mixed state obtained by completely depolarizing $\rho$ (or equivalently, $\tilde{\rho}_0$) in each symmetry sector.
Here, the factor of $p_\lambda$ is chosen such that $\tr(\rho_{\mathcal{F}_\lambda}) = 1$.
We assume $\xi = \omega(\log n)$; the proof for smaller $\xi$ follows by an identical argument as in the proof of Theorem~\ref{thm:quantum}.
Since $U$ is a symmetric PRU, we know that $\rho$ cannot be distinguished from the $V \tilde{\rho_0} V^\dagger$ where $V$ is a symmetric Haar-random unitary.
It remains only to show that $V \tilde{\rho_0} V^\dagger$ is indistinguishable from $\rho_{\text{mixed}}$.

To show this, we borrow a technique from recent studies of pseudorandom states and unitaries known as the distinct subspace.
We refer to Refs.~\cite{metger2024simple,cui2025unitary} for a thorough introduction.
Let $k$ denote the number of queries that an algorithm makes to the random state of interest.
Let $x^{(j)}$ label a computational basis of $d^n$ bitstrings for the $n$ ancilla qudits in the $j$-th copy of state, where $j = 1,\ldots,k$.
Then, let $\Pi^{\text{dist}}$ denote the projector onto the subspace of the $k$ copies of $n$ ancilla qudits where all $k$ values of  $x^{(j)}$ are distinct.
The projector is upper bounded by the operator inequality, $\mathbbm{1}-\Pi^{\text{dist}} \preceq \sum_{i<j} \Pi^{\text{eq},ij}$, where $\Pi^{\text{eq},ij}$ is the projector onto states where $x^{(i)} = x^{(j)}$.
The first step of our proof is then to project the $k$ copies of the state $\tilde{\rho}_0$ onto the distinct subspace on the ancilla qudits.
This incurs a negligibly small error,
\begin{equation} \label{eq: error distinct mixed}
    \left\lVert \tilde{\rho}^{\otimes k}_0 - \Pi^{\text{dist}} \tilde{\rho}^{\otimes k}_0 \Pi^{\text{dist}}\right\rVert_1 = \tr((\mathbbm{1}-\Pi^{\text{dist}})\tilde{\rho}_0^{\otimes k}) \leq \sum_{i<j} \tr(\Pi^{\text{eq},ij} \tilde{\rho}_0^{\otimes k}) \leq k^2/2d^n.
\end{equation}
The term is upper bounded by $k^2/2d^n$, since the probability for two copies of the $n$ ancilla qudits to be in the same state is $1/d^n$.
The factor of $k(k-1)/2\leq k^2/2$ counts the number of terms in the sum.

The second step of our proof is to compute the twirl of the projected state over the symmetric Haar-random unitary $V$.
From Lemma~\ref{lemma: approximate symmetric Haar twirl}, the twirl is equal to the following up to negligibly small relative error,
\begin{equation}
    \frac{1}{D^k} \sum_{\pi} \sum_{\bs{g}} \tr( \Pi^{\text{dist}} \tilde{\rho}_0^{\otimes k} \Pi^{\text{dist}} \pi^{-1} R_{\bs g}^{-1} ) R_{\bs g} \pi = \frac{1}{D^k}  \sum_{\bs{g}} \tr( \Pi^{\text{dist}} \tilde{\rho}_0^{\otimes k} \Pi^{\text{dist}} R_{\bs g}^{-1} ) R_{\bs g},
\end{equation}
where we use that $\tr_{\text{ancillas}}(\Pi^{\text{dist}} \pi^{-1}) = \delta_{\mathbbm{1},\pi}$.
From Eq.~(\ref{eq: error distinct mixed}), this state is equal to 
\begin{equation}
    \frac{1}{D^k}  \sum_{\bs{g}} \tr(  \tilde{\rho}_0^{\otimes k}  R_{\bs g}^{-1} ) R_{\bs g} = \rho_{\text{mixed}}^{\otimes k}
\end{equation}
up to an additional 1-norm error $k^2/2d$.
This completes the proof.\qed


\section{Classical phases of matter} \label{sec: classical}

In this section, we provide the full details of  our results on the hardness of recognizing phases of matter in classical systems.
We begin by precisely stating our definition of classical phases of matter, which is closely modeled off of standard definitions for quantum phases of matter (restricted to classical states and channels).
We then provide the detailed proof of our Theorems~\ref{thm:classical} from the main text.

To our knowledge, a precise definition of classical phases of matter, in a similar format to the now standard definitions of quantum phases of matter, has not been given in the literature.
We propose the following definition, which is identical to the definition of mixed state quantum phases of matter but with all operations restricted to be entirely classical.
\begin{definition}
    [Classical phases of matter]
    Our results apply to any definition of classical phases of matter such that:
    \begin{itemize}
        \vspace{-2mm}
        \item For each phase, one can associate a reference ``fixed point'' probability distribution $p_0(x)$.
        \vspace{-2mm}
        \item Any distribution $p(x)$ that can be mapped to and from $p_0(x)$ via symmetry-covariant geometrically-local depth-$\ell$ stochastic processes, $\mathcal{M}_1(p) = p_0$ and $\mathcal{M}_2(p_0) = p$, is in the same phase as $p_0(x)$.
    \end{itemize}
\end{definition}
\noindent As in the mixed state quantum setting, the only stochastic processes that we require are combinations of symmetric local reversible circuits and adding and removing ancilla bits in the maximally mixed state.

\subsection{Proof of Theorem~\ref{thm:classical}: Hardness of recognizing classical phases of matter}


Let $\rho_0$ denote the diagonal density matrix with entries $\bra{x} \rho_0 \ket{x} = p_0(x)$.
Let $\mathcal{M}_P(\rho) = P [ \rho \otimes (\mathbbm{1}/2)^{\otimes n} ] P^\dagger$ denote the classical stochastic process of interest [Fig.~\ref{fig:classical}(a)], where $P = \bigotimes_{\alpha=1}^m P_\alpha$ and each $P_\alpha$ is a pseudorandom permutation acting on $2\xi = 2n/m$ bits.
By assumption, each pseudorandom permutation cannot be distinguished from a truly random permutation by any $\poly n$-time classical or quantum experiment.
Hence, from hereon, we can assume each $P_\alpha$ is truly random.

Consider a classical or quantum experiment that queries $\mathcal{M}(p_0(x))$ at most $k = \poly n$ times and outputs $0$ if it decides that the state is in the trivial phase and $1$ if it decides the state is in the symmetry-breaking phase.
The expected output of the experiment can be written as,
\begin{equation}
    \E_P \left[ \tr( M \cdot \mathcal{M}_P(\rho_0)^{\otimes k} ) \right],
\end{equation}
for some Hermitian observable $M$ acting on all $k$ copies, with $0 \preceq M \preceq 1$.

To proceed, let us define the \emph{symmetric distinct subspace} as the set of all sets of $k$ bitstrings such that no two strings are equal up to the symmetry operation.
That is, $x_i \neq x_j$ for all $1 \leq i < j \leq k$ and $x_i \neq \bar x_j$ as well, where $\bar x_j$ is equal to $x_j$ with all bits flipped (i.e.~acted on by the symmetry operation).
We let $\Pi^{\text{sym-dist}}_\alpha$ denote the projector onto the symmetric distinct subspace on patch $\alpha$; that is, onto all sets of $k$ $2\xi$-bit strings on patch $i$ such that no two strings are equal to one another up to the symmetry.
We also let $\Pi^{\text{sym-dist}}_{*} = \bigotimes_{\alpha=1}^m \Pi^{\text{dist}}_\alpha$ denote the projector onto the distinct subspace on every patch.

The first step of our proof is to replace $\rho_0 \otimes (\mathbbm{1}/2)^{\otimes n}$ with its projection onto the symmetric distinct subspace $\Pi^{\text{sym-dist}}_{*}$.
This incurs an error
\begin{equation} \label{eq: classical 1-norm}
    E \equiv \left\lVert \rho_0 \otimes (\mathbbm{1}/2)^{\otimes n} - \Pi^{\text{sym-dist}}_{*} \big[ \rho_0 \otimes (\mathbbm{1}/2)^{\otimes n} \big] \Pi^{\text{sym-dist}}_{*} \right\rVert_1 = \tr( \big(\mathbbm{1} - \Pi^{\text{sym-dist}}_{*} \big) \big[ \rho_0 \otimes (\mathbbm{1}/2)^{\otimes n} \big] ),
\end{equation}
where we use that $\Pi^{\text{sym-dist}}_{*}$ and the state commute since the state is classical.
The trace measures the probability that any non-distinct set of $k$ $2\xi$-bit strings is observed when measuring $\rho_0 \otimes (\mathbbm{1}/2)^{\otimes n}$.
This can be upper bounded using the inequality $\Pi^{\text{sym-dist}}_{*} \preceq \sum_\alpha \sum_{i < j} \Pi^{\text{eq},ij}_{\alpha}$, where $\Pi^{\text{eq},ij}_{\alpha}$ is the projector onto all states with the same bitstring (up to the symmetry operation) on copies $i$ and $j$.
Here, the first sum is over all patches and the second sum is over all pairs of copies, $1 \leq i < j \leq k$.
The probability to observe the same (up to the symmetry) $2\xi$-bit string on copies $i$ and $j$ is upper bounded by the probability to observe the same (up to the symmetry) $\xi$-bit string when restricting attention to the $\xi$ ancilla qubits.
Since the ancilla qubits are maximally mixed, this probability is equal to $2/2^\xi$, where the factor of $2$ in the numerator accounts for equality up to the symmetry operation.
From this, the 1-norm in Eq.~(\ref{eq: classical 1-norm}) is upper bounded by $m k (k-1) / 2^{\xi}$.

To complete the proof, we note that the action of each symmetric random permutation on a symmetry-distinct set of bitstrings is especially simple,
\begin{equation} \label{eq: classical twirl P}
    \E_{P_\alpha}\big[ \dyad{x_\alpha} \big] = \Pi^{\text{sym-dist}}_\alpha / \mathcal{D}^{\text{sym-dist}}_\alpha, \,\,\,\,\,\, \text{ if } \,\,\,\,\,\, \Pi^{\text{sym-dist}}_\alpha \ket{x_\alpha} = \ket{x_\alpha},
\end{equation}
where $\mathcal{D}^{\text{sym-dist}}_\alpha = 2^{2\xi} (2^{2\xi}-2) (2^{2\xi}-4) \cdots (2^{2\xi}-2k+2) \geq 2^{2\xi}(1-k^2/2^{2\xi})$ is the dimension of the symmetric distinct subspace on patch $\alpha$.
The formula follows because a symmetric random permutation sends each input bitstring to a random output bitstring, up to only the condition that two inputs related by the symmetry operation must go to outputs that are also related by the symmetry.
Since $x_\alpha$ is symmetry-distinct, no two inputs are related by the symmetry.
Hence, all outputs are distinct, not related by the symmetry, and otherwise independently random.
When taken in expectation, this produces the maximally mixed state on the symmetric distinct subspace.

Applying Eq.~(\ref{eq: classical twirl P}) to each patch $\alpha$, we find
\begin{equation}
    \E_P \left[ \left( P \big[\Pi^{\text{sym-dist}}_{*} \big[ \rho_0 \otimes (\mathbbm{1}/2)^{\otimes n} \big] \Pi^{\text{sym-dist}}_{*} \big] P^\dagger \right)^{\otimes k}  \right] = (1-E) \cdot \bigotimes_{\alpha=1}^m \left( \frac{\Pi^{\text{sym-dist}}_\alpha }{ \mathcal{D}^{\text{sym-dist}}_\alpha } \right),  
\end{equation}
where $1-E$ is the normalization of the state, with $E$ defined in Eq.~(\ref{eq: classical 1-norm}).
The state on the right hand side is close to the fixed state,
\begin{equation}
    \bigotimes_{\alpha=1}^m \left( \frac{\Pi^{\text{sym-dist}}_\alpha }{ \mathcal{D}^{\text{sym-dist}}_\alpha } \right),
\end{equation}
up to 1-norm error $E \leq mk(k-1)/2^\xi$.
However, this state is independent of the input probability distribution $p_0(x)$.
Hence, the expected output states of any two different input distributions, $p_0(x)$ and $p_1(x)$, will be equal up to 1-norm error $4 mk(k-1)/2^\xi$.
This immediately implies that their expected output decisions are equal up to $4 m k(k-1)/2^\xi = 1/\omega(\poly n)$ given $\xi = \omega(\log n)$.
Hence, the experiment cannot distinguish the output distributions arising from $p_0(x)$ and $p_1(x)$ in polynomial time.
Setting $p_1(x)$ to be the maximally mixed distribution yields Theorem~\ref{thm:classical}. \qed

\section{Numerics} \label{sec: numerics_appendix}

\subsection{Construction of symmetric Clifford circuits}
\subsection{Additional numerical data}

\bibliography{refs}
\bibliographystyle{unsrt}

\end{document}